\documentclass[11pt]{amsart}
\usepackage{graphicx, hyperref}
\usepackage{color}

\newtheorem{thm}{Theorem}[section]

\newtheorem{lem}[thm]{Lemma}
\newtheorem{prop}[thm]{Proposition}
\theoremstyle{definition}
\newtheorem{defn}[thm]{Definition}

\newtheorem{ass}[thm]{Assumption}

\theoremstyle{remark}
\newtheorem{rem}[thm]{Remark}
\newtheorem{exa}[thm]{Example}
\numberwithin{equation}{section}

\newcommand{\abs}[1]{\left\vert#1\right\vert}
\newcommand{\set}[1]{\left\{#1\right\}}
\newcommand{\Real}{\mathbf R}
\newcommand{\Natural}{\mathbf N}

\newcommand{\such}{\, | \,}
\newcommand{\prob}{\mathbb{P}}

\newcommand{\qprob}{\mathbb{Q}}
\newcommand{\rprob}{\mathbb{R}}
\newcommand{\expec}{\mathbb{E}}
\newcommand{\expecp}{\expec_\prob}
\newcommand{\expecq}{\expec_\qprob}

\newcommand{\var}{\mathbb{V} \mathsf{ar}}
\newcommand{\cov}{\mathbb{C} \mathsf{ov}}

\newcommand{\ud}{\mathrm d}

\newcommand{\normtv}[1]{\left|#1\right|_{\mathsf{TV}}}

\newcommand{\expecqb}{\expec_{\qb}}
\newcommand{\expecqg}{\expec_{\qg}}
\newcommand{\expecqp}{\expec_{\qp}}

\newcommand{\cg}{C^{\diamond}}

\newcommand{\rg}{\rprob^{\diamond}}

\newcommand{\zg}{z^{\diamond}}

\newcommand{\ug}{u^{\diamond}}
\newcommand{\qg}{\qprob^{\diamond}}

\newcommand{\cb}{C^{\mathsf{r}}}
\newcommand{\db}{D^{\mathsf{r}}}

\newcommand{\rb}{\rprob^{\mathsf{r}}}

\newcommand{\zeb}{\zeta}

\newcommand{\qb}{\qprob_i^{\mathsf{r}}}

\newcommand{\cbm}{C_0^{m,\mathsf{r}}}
\newcommand{\dbm}{D_0^{m,\mathsf{r}}}

\newcommand{\zbm}{z_0^{m,\mathsf{r}}}

\newcommand{\qbm}{\qprob_0^{m,\mathsf{r}}}

\newcommand{\cbmk}{C_0^{m_k,\mathsf{r}}}
\newcommand{\dbmk}{D_0^{m_k,\mathsf{r}}}
\newcommand{\zbmk}{z_0^{m_k,\mathsf{r}}}

\newcommand{\qbmk}{\qprob^{m_k,\mathsf{r}}}

\newcommand{\cbom}{C_0^{m,\mathsf{r}}}

\newcommand{\cgm}{C^{m, \diamond}}
\newcommand{\dgm}{D^{m, \diamond}}

\newcommand{\rgm}{\rprob^{m, \diamond}}

\newcommand{\zgm}{z^{m, \diamond}}

\newcommand{\ugm}{u^{m, \diamond}}
\newcommand{\qgm}{\qprob^{m, \diamond}}

\newcommand{\cgmk}{C^{m_k, \diamond}}
\newcommand{\dgmk}{D^{m_k, \diamond}}
\newcommand{\zgmk}{z^{m_k, \diamond}}

\newcommand{\cgi}{C^{\infty, \diamond}}

\newcommand{\zgi}{z^{\infty, \diamond}}

\newcommand{\qgi}{\qprob^{\infty, \diamond}}

\newcommand{\cp}{C^{\ast}}
\newcommand{\qp}{\qprob^{\ast}}
\newcommand{\up}{u^{\ast}}

\newcommand{\PP}{\mathcal{P}}
\newcommand{\cpk}{C^{m, \ast}}
\newcommand{\cpm}{C^{m, \ast}}
\newcommand{\qpk}{\qprob^{m, \ast}}
\newcommand{\upk}{u^{m, \ast}}
\newcommand{\qpm}{\qprob^{m, \ast}}
\newcommand{\upm}{u^{m, \ast}}

\newcommand{\cpi}{C^{\infty, \ast}}

\newcommand{\qpi}{\qprob^{\infty, \ast}}

\newcommand{\qz}{\qprob(z)}
\newcommand{\expecqz}{\expec_{\qz}}

\newcommand{\DI}{\Delta^I}

\newcommand{\oprob}{\overline{\prob}}

\newcommand{\oexpec}{\overline{\expec}}

\newcommand{\tz}{\widetilde{z}}
\newcommand{\tze}{\widetilde{z}_0^{\infty}}

\newcommand{\hzi}{\widehat{z}_0^\infty}
\newcommand{\tD}{\widetilde{D}}

\newcommand{\hDi}{\widehat{D}^\infty_0}
\newcommand{\hCi}{\widehat{C}^\infty_0}
\newcommand{\tC}{\widetilde{C}}

\newcommand{\plim}{\Lb^0 \text{-} \lim}

\newcommand{\plimm}{\plim_{m \to \infty}}
\newcommand{\plimk}{\plim_{k \to \infty}}

\newcommand{\limn}{\lim_{n \to \infty}}
\newcommand{\limk}{\lim_{k \to \infty}}
\newcommand{\limm}{\lim_{m \to \infty}}

\newcommand{\pare}[1]{\left(#1\right)}
\newcommand{\bra}[1]{\left[#1\right]}
\newcommand{\dbra}[1]{[\kern-0.15em[ #1 ]\kern-0.15em]}
\newcommand{\dbraco}[1]{[\kern-0.15em[ #1 [\kern-0.15em[}
\newcommand{\dbraoc}[1]{]\kern-0.15em] #1 ]\kern-0.15em]}
\newcommand{\C}{\mathcal{C}}

\newcommand{\K}{\mathcal{K}}

\newcommand{\Lb}{\mathbb{L}}

\newcommand{\U}{\mathbb{U}}
\newcommand{\V}{\mathbb{V}}

\newcommand{\dfn}{\mathrel{\mathop:}=}

\newcommand{\lz}{\Lb^0}

\newcommand{\li}{\Lb^\infty}

\newcommand{\kin}{k \in \Natural}
\newcommand{\mina}{m \in \Natural}

\newcommand{\indic}{\mathbb{I}}

\newcommand{\iii}{i \in I}
\newcommand{\jii}{j \in I}

\newcommand{\sumi}{\sum_{\iii}}

\newcommand{\dis}{\ell}

\newcommand{\ent}{\mathcal{H}}

\newcommand{\lzi}{(\Lb^0)^I}

\begin{document}

\title[Equilibrium in risk-sharing games]{Equilibrium in risk-sharing games}%

\author{Michail Anthropelos}
\address{Michail Anthropelos, Department of Banking and Financial Management, University of Piraeus}
\email{anthropel@unipi.gr}

\author{Constantinos Kardaras}
\address{Constantinos Kardaras, Statistics Department, London School of Economics}
\email{k.kardaras@lse.ac.uk}

\thanks{M. Anthropelos acknowledges support from the Research Center of the University of Piraeus. C. Kardaras acknowledges support from the MC grant FP7-PEOPLE-2012-CIG, 334540.}%

\keywords{Nash equilibrium, risk sharing, heterogeneous beliefs, reporting of beliefs}
\date{\today}%
\begin{abstract}
The large majority of risk-sharing transactions involve few
agents, each of whom can heavily influence the structure and the
prices of securities. This paper proposes a game where agents'
strategic sets consist of all possible sharing securities and
pricing kernels that are consistent with Arrow-Debreu sharing
rules. First, it is shown that agents' best response problems have
unique solutions. The risk-sharing Nash equilibrium admits a
finite-dimensional characterisation and it is proved to exist for
arbitrary number of agents and be unique in the two-agent
game. In equilibrium, agents declare beliefs on future
random outcomes different than their actual probability
assessments, and the risk-sharing securities are endogenously
bounded, implying (among other things) loss of efficiency. In
addition, an analysis regarding extremely risk tolerant agents
indicates that they profit more from the Nash risk-sharing
equilibrium as compared to the Arrow-Debreu one.
\end{abstract}

\maketitle


\section*{Introduction}

The structure of securities that optimally allocate risky
positions under heterogeneous beliefs of agents has been a
subject of ongoing research. Starting from the seminal works of
\cite{Bor62}, \cite{Arr63}, \cite{BuhlJew79} and \cite{Buhl84}, the existence and characterisation of welfare risk sharing of random positions in a variety of models has been extensively studied---see, among others, \cite{BarElKar05}, \cite{JouiSchTou08}, \cite{Acc07}, \cite{FilSci08}. On the other hand,
discrepancies amongst agents regarding their assessments on the
probability of future random outcomes reinforce the existence of mutually
beneficial trading opportunities (see e.g.~\cite{Var85}, \cite{Var89},
\cite{BilChaGilTall00}). However, market imperfections---such as
asymmetric information, transaction costs and oligopolies---spur agents to act strategically and prevent markets from reaching maximum efficiency. In the financial risk-sharing
literature, the impact of asymmetric or private information has been addressed under both static and dynamic models (see, among others, \cite{NacNoe94}, \cite{MarRah00},
\cite{Par04}, \cite{Axe07}, \cite{Wil11}). The
importance of frictions like transaction costs has be highlighted
in \cite{AllGal91}; see also \cite{CarRosWer12}.

The present work aims to contribute to the risk-sharing literature by focusing on how over-the-counter (OTC) transactions with a small number of agents motivate strategic behaviour. The vast majority of real-world sharing instances involves only a few participants, each of whom may influence the way heterogeneous risks and beliefs are going to be allocated. (The seminal papers \cite{Kyl89} and \cite{Vay99} highlight such transactions.) As an example, two financial institutions with possibly different beliefs, and in possession of portfolios with random future payoffs, may negotiate and design innovative asset-backed securities that mutually share their defaultable assets. Broader discussion on risk-sharing innovative securities is given in the classic reference \cite{AllGal94} and in \cite{Tuf03}; a list of widely used such securities is provided in \cite{Finn92}.

As has been extensively pointed out in the literature (see, for example, \cite{Var89} and \cite{SchXio03}), it is reasonable, and perhaps even necessary, to assume that agents have heterogeneous beliefs, which we identify with subjective probability measures on the considered state space.  In fact, differences in subjective beliefs do not necessarily stem from asymmetric information; agents usually apply different tools or models for the analysis and interpretation of common sets of information.

Formally, a risk-sharing transaction consists of security payoffs and their prices, and since only few institutions (typically, two) are involved, it is natural to assume that no social planner for the transaction exists, and that the equilibrium valuation and payoffs will result as the outcome of a symmetric game played among the participating institutions. Since institutions' portfolios are (at least, approximately) known, the main ingredient of risk-sharing transactions leaving room for strategic behaviour is the beliefs that each institution reports for the sharing. We propose a novel way of modelling such strategic actions where the agents' strategic set consists of the beliefs that each one chooses to declare (as opposed to their actual one) aiming to maximise individual utility, and the induced game leads to an equilibrium sharing. Our main insights are summarised below.

\subsection*{Main contributions}
The payoff and valuation of the risk-sharing securities are
endogenously derived as an outcome of agents' strategic behaviour,
under constant absolute risk-aversion (CARA) preferences. To the
best of our knowledge, this work is the first instance that models the way
agents choose the beliefs on future uncertain events that they are
going to declare to their counterparties, and studies whether such
strategic behaviour results in equilibrium. Our results
demonstrate how the game leads to risk-sharing
inefficiency and security mispricing, both of which are
quantitatively characterised in analytic forms. More
importantly, it is shown that equilibrium securities have
\textit{endogenous limited liability}, a feature that, while
usually suboptimal, is in fact observed in practice.

Although the
agents' set of strategic choices is infinite-dimensional, one of
our main contributions is to show that Nash equilibrium admits a
finite-dimensional characterisation, with the dimensionality being one less than the number of participating agents. Not only does our characterisation provide a concrete algorithm for calculating the equilibrium transaction, it also allows to
prove existence of Nash equilibrium for arbitrary number
of players. In the important case of two participating agents, we
even show that Nash equilibrium is unique. It has to
be pointed out that the aforementioned results are obtained under
complete generality on the probability space and the involved
random payoffs---no extra assumption except from CARA preferences is imposed. While certain qualitative analysis could be potentially carried out without the latter assumption on the entropic form of agent utilities, the advantage of CARA preferences utilised in the present paper is that they also allow for substantial quantitative analysis, as workable expressions are obtained for Nash equilibrium.

Our notion of Nash risk-sharing equilibrium highlights the
importance of agents' risk tolerance level. More precisely, one of
the main findings of this work is that agents with
sufficiently low risk aversion will prefer the
risk-sharing game rather than the outcome of an Arrow-Debreu equilibrium that would have resulted from absence of strategic behaviour. Interestingly, the result is valid
irrespective of their actual risky position or their subjective
beliefs. It follows that even risk-averse agents, as long as their
risk-aversion is sufficiently low, will prefer risk-sharing
markets that are thin (i.e., where participating agents are few and have the power to influence the transaction), resulting in aggregate loss of risk-sharing welfare.

\subsection*{Discussion}

Our model is introduced in Section \ref{sec: A-D equilibrium}, and
consists of a two-period financial economy with uncertainty,
containing possibly infinite states of the world. Such infinite-dimensionality is essential in our framework, since in general the risks that agents encounter do not have a-priori bounds, and we do not wish to enforce any restrictive assumption on the shape of the probability distribution or the support of agents' positions. Let us also note that, even if the analysis was carried out in a simpler set-up of a finite state space, there would not be any significant simplification in the mathematical treatment.

In the economy we consider a finite number of agents, each
of whom has subjective beliefs (probability measure) about the events at the time of uncertainty resolution. We also allow agents to be endowed with a (cumulative, up to the point of uncertainty resolution) random endowment.

Agents seek to increase their expected utilities through trading securities that allocate the discrepancies of their beliefs and risky exposures in an optimal way. The possible disagreement on agents' beliefs is assumed on the whole probability space, and not only on the laws of the shared-to-be risky positions. Such potential  disagreement is important: it alone can give rise to mutually beneficial trading opportunities, even if agents have no risky endowments to share, by actually designing securities with payoffs written on the events where probability assessments are different. 

Each sharing rule consists of the security payoff that each agent is going to obtain and a valuation measure under which all imaginable securities are priced. The sharing rules that efficiently allocate any submitted discrepancy of beliefs and risky exposures are the ones stemming from Arrow-Debreu equilibrium. (Under CARA preferences, the optimal sharing rules have been extensively studied---see, for instance, \cite{Bor62}, \cite{BuhlJew79} and \cite{BarElKar05}.) In principle, participating agents would opt for the highest possible aggregate benefit from the risk-sharing transaction, as this would increase their chance for personal gain. However, in the absence of a social planner that could potentially impose a truth-telling mechanism, it is reasonable to assume that agents do not negotiate the rules that will allocate the submitted endowments and beliefs. In fact, we assume that agents adapt the specific sharing rules that are consistent with the ones resulting from Arrow-Debreu equilibrium, treating reported beliefs as actual ones, since we regard these sharing rules to be the most natural and universally regarded as efficient.

Agreement on the structure of risk-sharing securities is also consistent with what is observed in many OTC transactions involving security design, where the contracts signed by institutions are standardised and adjusted according to required inputs (in this case, the agents' reported beliefs). Such pre-agreement on sharing rules reduces negotiation time, hence the related transaction costs. Examples are asset-backed securities, whose payoffs are backed by issuers' random incomes, traded among banks and investors in a standardised form, as well as credit derivatives, where portfolios of defaultable assets are allocated among financial institutions and investors.

Combinations of strategic and competitive stages are widely used in the literature of financial
innovation and risk-sharing, under a variety of different guises. The majority of this literature distinguishes participants among designers (or issuers) of securities and investors who trade them. In \cite{DuffJac89}, a security-design game is played among exchanges, each aiming to maximise internal transaction volume; while security design throughout exchanges is the outcome of non-competitive equilibrium, investors trade securities in a competitive manner. Similarly, in \cite{Bis98}, Nash equilibrium determines not only the designed securities among financial intermediaries, but also the bid-ask spread that price-taking investors have to face in the second (perfect competition) stage of market equilibrium. In \cite{CarRosWer12}, it is entrepreneurs who strategically design securities that investors with non-securitised hedging needs competitively trade. In \cite{RahZig09}, the role of security-designers is played by arbitrageurs who issue innovated securities in segmented markets. Mixture of strategic and competitive stages has also been used in models with asymmetric information. For instance, in \cite{Brai05} a two-stage equilibrium game is used to model security design among agents with private information regarding their effort. In a first stage, agents strategically issue novel financial securities; in the second stage, equilibrium on the issued securities is formed competitively.

Our framework models oligopolistic OTC security design, where participants are \textit{not} distinguished regarding their information or ability to influence market equilibrium. Agents mutually agree to apply Arrow-Debreu sharing rules, since these optimally allocate whatever is submitted for sharing, and also strategically choose the inputs of the sharing-rules (their beliefs, in particular).

Given the agreed-upon rules,
agents propose accordingly consistent securities and valuation
measures, aiming to maximise \textit{their own} expected utility. As
explicitly explained in the text, proposing risk-sharing
securities and a valuation kernel is in fact equivalent to agents
reporting beliefs to be considered for sharing. Knowledge of the probability assessments of the counterparties may result in a readjustment of the probability measure an agent is going to report for the transaction. In effect, agents form a game by responding to other agents' submitted
probability measures; the fixed point of this game (if it exists)
is called \textit{Nash risk-sharing equilibrium}.

The first step of analysing Nash risk-sharing equilibria is to
address the well-posedness of an agent's best response problem,
which is the purpose of Section \ref{sec: best_response}. Agents
have motive to exploit other agents' reported beliefs and hedging
needs and drive the sharing transaction as to maximise their own
utility. Each agent's strategic choice set consists of all
possible probability measures (equivalent to a baseline measure),
and the optimal one is called \textit{best probability response}.
Although this is a highly non-trivial infinite-dimensional
maximisation problem, we use a bare-hands approach to establish that it admits a unique solution. It is shown that the beliefs that
an agent declares coincide with the actual ones only in the
special case where the agent's position cannot be improved by any
transaction with other agents. By resorting to examples, one may gain more intuition on how future risk appears under the lens of agents' reported beliefs. Consider, for instance, two financial institutions adapting distinct models for estimating the likelihood of the involved risks. The sharing contract designed by the institutions will result from individual estimation of the joint distribution of the shared-to-be risky portfolios. According to the best probability response procedure, each institution tends to use less favourable assessment for its own portfolio than the one based on its actual beliefs, and understates the downside risk of its counterparty's portfolio. Example \ref{ex: best response} contains an illustration of such a case.

An important consequence of applying the best
probability response is that the corresponding security that the
agent wishes to acquire has bounded liability. If only one agent applies the proposed strategic behaviour, the received security payoff is bounded below (but not necessarily bounded above). In fact, the arguments and results of the best response problem receive extra attention and discussion in the paper, since
they demonstrate in particular the value of the proposed strategic
behaviour in terms of utility increase. This
situation applies to markets where one large
institution trades with a number of small agents, each of
whom has negligible market power.

A Nash-type game occurs when all agents apply the
best probability response strategy. In Section \ref{sec: Nash}, we
characterise Nash equilibrium as the solution of a certain
finite-dimensional problem. Based on this characterisation, we establish existence of Nash risk-sharing equilibrium for an arbitrary (finite)
number of agents. In the special case of two-agent games, the Nash
equilibrium is shown to be unique. The finite-dimensional
characterisation of Nash equilibrium also provides an
algorithm that can be used to approximate the Nash equilibrium
transaction by standard numerical procedures, such as Monte Carlo simulation.

Having Nash equilibrium characterised, we are able to further
perform a joint qualitative and quantitative analysis. Not only
do we verify the expected fact that, in any non-trivial case, Nash
risk-sharing securities are different from the Arrow-Debreu ones,
but we also provide analytic formulas for their shapes. Since the securities that correspond to the best probability response are bounded from below, the application of such strategy from all the agents yields that the Nash risk-sharing market-clearing securities are also bounded from above. This comes in stark contrast to Arrow-Debreu equilibrium, and
implies in particular an important loss of efficiency. We measure
the risk-sharing inefficiency that is caused by the game via the
difference between the aggregate monetary utilities at
Arrow-Debreu and Nash equilibria, and provide an analytic
expression for it. (Note that inefficient allocation of risk in symmetric-information thin market models may also occur when securities are exogenously given---see e.g.~\cite{RosWer15}. When securities are endogenously designed, \cite{CarRosWer12} highlights that imperfect competition among issuers results in risk-sharing inefficiency, even if securities is traded among perfectly competitive investors.)

One may wonder whether the revealed agents' subjective beliefs in Nash equilibrium are far from their actual subjective probability measures, which would be unappealing from a modelling viewpoint. Extreme departures from actual beliefs are endogenously
excluded in our model, as the distance of the truth from reported beliefs in Nash equilibrium admits a-priori bounds. Even though agents are free to choose any probability measure that supposedly represents their beliefs in a risk-sharing transaction, and they do indeed end up choosing probability measures different than their actual ones, this departure cannot be arbitrarily large if the market is to reach equilibrium.

Turning our attention to Nash-equilibrium valuation, we show
that the pricing probability measure can be written as a certain convex combination of the individual agents' marginal indifference valuation measures. The weights of this convex combination depend on agents' relative risk tolerance coefficients, and, as it turns out, the Nash-equilibrium valuation measure is closer to the marginal valuation measure of the more
risk-averse agents. This fact highlights the importance of risk tolerance coefficients in assessing the gain or loss of utility for individual agents in Nash risk-sharing equilibrium; in fact, it implies that more risk tolerant agents tend to get better cash compensation as a result of the Nash game than what they would get in Arrow-Debreu equilibrium.

Inspired by the involvement of the risk tolerance coefficients in
the agents' utility gain or loss, in Section \ref{sec:
extrema} we focus on induced Arrow-Debreu and Nash equilibria of two-agent games, when one of the agents' preferences approach risk neutrality.
We first establish that both equilibria converge to well-defined
limits. Notably, it is shown that an extremely risk tolerant agent
drives the market to the same equilibrium regardless of
whether the other agent acts strategically or plainly submits true subjective beliefs. In other words, extremely risk tolerant agents tend to dominate the risk-sharing transaction. The study of limiting equilibria indicates that, although there is loss of aggregate utility when agents act strategically, there is always utility gain in the Nash transaction as compared to Arrow-Debreu equilibrium for the extremely risk-tolerant agent, regardless of the risk tolerance level and subjective beliefs of the other agent. Extremely risk-tolerant agents are willing to undertake more risk
in exchange of better cash compensation; under the risk-sharing
game, they respond to the risk-averse agent's hedging needs and beliefs by
driving the market to higher price for the security they short.
This implies that agents with sufficiently high risk tolerance---although still not risk-neutral---will prefer thin markets. The case where both acting agents uniformly approach
risk-neutrality is also treated, where it is shown that the
limiting Nash equilibrium sharing securities equal half of the
limiting Arrow-Debreu equilibrium securities, hinting towards the
fact that Nash risk-sharing equilibrium results in loss of trading
volume.

For convenience of reading, all the proofs of the paper are
placed in Appendix \ref{sec: appe}.

\section{Optimal Sharing of Risk}\label{sec: A-D equilibrium}

\subsection{Notation}

The symbols ``$\Natural$'' and ``$\Real$'' will be used to denote the set of all natural and real numbers, respectively. As will be evident subsequently in the paper, we have chosen to use the symbol ``$\rprob$'' to denote (reported, or revealed) probabilities.

In all that follows, random variables are defined on a standard
probability space $(\Omega, \, \mathcal{F}, \, \prob)$. We stress that \emph{no} finiteness restriction is enforced on the state space $\Omega$. We use $\PP$ for the class of all probabilities that are equivalent to the baseline probability $\prob$. For $\qprob \in \PP$, we use ``$\expecq$'' to denote expectation under $\qprob$. 
The space $\Lb^0$ consists of all (equivalence classes, modulo
almost sure equality) finitely-valued random variables endowed
with the topology of convergence in probability. This
topology does not depend on the representative probability from
$\PP$, and $\Lb^0$ may be infinite-dimensional. For $\qprob
\in \PP$, $\Lb^1 (\qprob)$ consists of all $X \in \lz$ with
$\expecq \bra{|X|} < \infty$. We use $\li$ for the subset of
$\lz$ consisting of essentially bounded random variables.

Whenever $\qprob_1 \in \PP$ and $\qprob_2 \in \PP$, $\ud \qprob_2 / \ud \qprob_1$ denotes the (strictly
positive) density of $\qprob_2$ with respect to $\qprob_1$. The \emph{relative entropy} of $\qprob_2 \in
\PP$ with respect to $\qprob_1 \in \PP$ is defined via
\[
\ent(\qprob_2 \such \qprob_1 ) \dfn \expec_{\qprob_1}\left[
\frac{\ud \qprob_2}{\ud \qprob_1} \log \left( \frac{\ud
\qprob_2}{\ud \qprob_1}\right) \right] = \expec_{\qprob_2}\left[
\log \left( \frac{\ud \qprob_2}{\ud \qprob_1}\right) \right] \in
[0, \infty].
\]

For $X \in \lz$ and $Y \in \lz$, we write $X
\sim Y$ if and only if there exists $c \in \Real$ such that $Y = X + c$. In particular, we shall use this notion of equivalence to ease notation on probability densities: for
$\qprob_1 \in \PP$ and $\qprob_2 \in \PP$, we shall write $\log
\pare{\ud \qprob_2 / \ud \qprob_1} \sim \Lambda$ to mean that
$\exp(\Lambda) \in \Lb^1 (\qprob_1)$ and $\ud \qprob_2 /  \ud
\qprob_1 = (\expec_{\qprob_1} \bra{\exp(\Lambda)})^{-1}
\exp(\Lambda)$.

\subsection{Agents and preferences}

We consider a market with a single future period, where all
uncertainty is resolved. In this market, there are $n + 1$
economic agents, where $n \in \Natural = \set{1, 2, \ldots}$; for
concreteness, define the index set $I = \set{0, \ldots, n}$. Agents derive utility only from the consumption of a
num\'eraire in the future, and all considered security payoffs are
expressed in units of this num\'eraire. In particular, future deterministic amounts have the same present value for
the agents. The preference structure of agent $\iii$ over future
random outcomes is numerically represented via the concave
\emph{exponential} utility functional
\begin{equation}\label{eq: utility functional}
\Lb^0 \ni X \mapsto \U_i (X) \dfn - \delta_i \log \expec_{\prob_i} \bra{\exp \pare{-  X /
\delta_i}} \in [- \infty, \infty),
\end{equation}
where $\delta_i \in (0, \infty)$ is the agent's \emph{risk tolerance} and $\prob_i \in \PP$ represents the agent's \emph{subjective beliefs}.
For any $X \in \lz$, agent $\iii$ is indifferent between the cash amount $\U_i (X)$ and the corresponding risky position $X$; in other words, $\U_i (X)$ is the \emph{certainty equivalent} of $X \in \lz$ for agent $\iii$. Note that the functional $- \U_i$ is an \emph{entropic risk measure} in the terminology of convex risk measure literature---see, amongst others, \cite[Chapter 4]{FollSch04}.

Define the aggregate risk tolerance $\delta \dfn
\sumi \delta_i$, as well as the relative risk tolerance $\lambda_i \dfn \delta_i / \delta$ for all $\iii$. Note that $\sumi \lambda_i = 1$. Finally, set $\delta_{-i} \dfn \delta - \delta_i$ and $\lambda_{-i} \dfn 1 - \lambda_i$, for all $\iii$.

\subsection{Subjective probabilities and endowments} \label{subsec: endow}

Preference structures that are numerically represented via \eqref{eq: utility functional} are rich enough to include the possibility of already existing portfolios of random positions for acting agents. To wit, suppose that $\widetilde{\prob}_i\in\PP$ are the actual subjective beliefs of agent $\iii$, who also carries a risky future payoff in units of the num\'eraire. Following standard terminology, we call this cumulative (up to the point of resolution
of uncertainty) payoff \emph{random endowment}, and denote it by $E_i \in \Lb^0$. In this set-up, adding on top of $E_i$ a payoff $X \in \Lb^0$ for agent $\iii$ results in numerical utility equal to $\widetilde{\U}_i(X) \dfn - \delta_i \log \expec_{\widetilde{\prob}_i} \bra{\exp \pare{-  (X+E_i) / \delta_i}}$. Assume that $\widetilde{\U}_i(0) > - \infty$, i.e., that $\exp \pare{- E_i / \delta_i} \in \Lb^1 \big( \widetilde{\prob}_i \big)$. Defining $\prob_i \in \PP$ via $\log \big( \ud \prob_i  / \ud  \widetilde{\prob}_i \big) \sim - E_i / \delta_i$ and $\U_i$ via \eqref{eq: utility functional}, $\U_i(X) = \widetilde{\U}_i(X) - \widetilde{\U}_i(0)$ holds for all $X \in \lz$. Hence, hereafter, the probability $\prob_i$ is understood to incorporate any possible random endowment of agent $\iii$, and utility is measured in \emph{relative} terms, as difference from the baseline level $\widetilde{\U}_i(0)$.

Taking the above discussion into account, we stress that agents are
completely characterised by their risk tolerance level and (endowment-modified) subjective beliefs, i.e., by the collection of pairs $(\delta_i,
\prob_i)_{\iii}$. In other aspects, and unless otherwise noted, agents
are considered symmetric (regarding information, bargaining power,
cost of risk-sharing participation, etc).

\subsection{Geometric-mean probability}  \label{subsec:geo_mean_prob}
We introduce a method that produces a geometric mean of probabilities which will play central role in our discussion. Fix $(\rprob_i)_{\iii} \in \PP^I$. In view of H\"older's inequality, $\prod_{\iii} \pare{\ud \rprob_i / \ud \prob}^{\lambda_i} \in \Lb^1(\prob)$ holds. Therefore, one may define $\qprob \in \PP$ via $\log \pare{\ud \qprob / \ud \prob} \sim \sum_{\iii} \lambda_i \log \pare{\ud \rprob_i / \ud \prob}$. Since $\sum_{\iii} \lambda_i \log \pare{\ud \rprob_i / \ud \qprob} \sim 0$, one is allowed to formally write
\begin{equation} \label{eq:geom-mean-prob}
\log \ud \qprob  \sim \sum_{\iii} \lambda_i \log \ud \rprob_i.
\end{equation}
The fact that $\ud \rprob_i / \ud \qprob \in \Lb^1(\qprob)$ implies $\log_+ \pare{\ud \rprob_i / \ud \qprob} \in \Lb^1(\qprob)$,
and Jensen's inequality gives $\expecq \bra{\log \pare{\ud \rprob_i / \ud \qprob}} \leq 0$, for all $\iii$. Note that \eqref{eq:geom-mean-prob} implies that the existence of $c \in \Real$ such that $\sum_{\iii} \lambda_i \log \pare{\ud \rprob_i / \ud \qprob} = c$ holds; therefore, one actually has $\expecq \bra{\log \pare{\ud \rprob_i / \ud \qprob}} \in (- \infty, 0]$, for all
$\iii$. In particular, $\log \pare{\ud \rprob_i / \ud \qprob} \in
\Lb^1(\qprob)$ holds for all $\iii$, and
\[
\ent(\qprob \such \rprob_i) = - \expec_{\qprob} \bra{\log \pare{\ud \rprob_i / \ud \qprob}} < \infty, \quad \forall \iii.
\]

\subsection{Securities and valuation} 

Discrepancies amongst agents' preferences provide incentive to design securities, the trading of which could be mutually beneficial in terms of risk reduction. In principle, the ability to design and trade securities in any desirable way essentially leads to a \emph{complete} market. In such a market, transactions amongst agents are characterised by a valuation measure (that assigns prices to all imaginable securities), and a collection of the
securities that will actually be traded. Since all
future payoffs are measured under the same num\'{e}raire,
(no-arbitrage) valuation corresponds to taking expectations with
respect to probabilities in $\PP$. Given a valuation measure,
agents agree in a collection $(C_i)_{\iii} \in \lzi$ of zero-value
securities, satisfying the market-clearing condition $\sum_{\iii}
C_i = 0$. The security that agent $\iii$ takes a long position as part of the transaction is $C_i$. 

As mentioned in the introductory section, our model could find applications in OTC markets. For instance, the design of asset-backed securities involves only a few number of financial institutions; in this case, $\prob_i$ stands for the subjective beliefs of each institution $\iii$ and, in view of the discussion of \S \ref{subsec: endow}, further incorporates any existing portfolios that back the security payoffs. In order to share their risky positions, the institutions agree on prices of future random payoffs and on the securities they are going to exchange. Other examples are the market of innovated credit derivatives or the market of asset swaps that involve exchange of random payoff and a fixed payment.

\subsection{Arrow-Debreu equilibrium} \label{subsec: AD}

In the absence of any kind of strategic behaviour in designing securities, the agreed-upon transaction amongst agents will actually form an Arrow-Debreu equilibrium. The valuation measure will determine both trading and indifference prices, and securities will be constructed in a way that maximise each agent's respective utility. 
\begin{defn}
$\pare{\qp, \pare{\cp_i}_{\iii}} \in \PP \times \lzi$ will be called an
\textsl{Arrow-Debreu equilibrium} if:
\begin{enumerate}
    \item $\sumi \cp_i =
0$, as well as $\cp_i \in \Lb^1(\qp)$  and $\expec_{\qp} \bra{\cp_i} =
0$, for all $\iii$, and
    \item for all $C \in \Lb^1 (\qp)$ with $\expecqp \bra{C} \leq 0$, $\U_i (C) \leq \U_i (\cp_i)$ holds for all $\iii$.
\end{enumerate}
\end{defn}

Under risk preferences modelled by \eqref{eq: utility functional}, a unique Arrow-Debreu equilibrium may be explicitly obtained. In other guises, Theorem \ref{thm: AD} that follows has appeared in many works---see for
instance \cite{Bor62}, \cite{BuhlJew79} and \cite{Buhl84}. Its proof is based on standard arguments; however, for reasons of completeness, we provide a short argument in \S \ref{subsec: proof_of_AD}.

\begin{thm} \label{thm: AD}
In the above setting, there exists a unique Arrow-Debreu equilibrium $\pare{\qp, \pare{\cp_i}_{\iii}}$. In fact, the valuation measure $\qp \in \PP$ is such that
\begin{equation} \label{eq: qp}
\log \ud \qp  \sim \sumi \lambda_i \log \ud \prob_i,
\end{equation}
and the equilibrium market-clearing securities $\pare{\cp_i}_{\iii} \in (\lz)^I$ are given by
\begin{equation} \label{eq: opt_contr}
\cp_i \dfn \delta_i \log (\ud \prob_i / \ud \qp) + \delta_i \ent (\qp \such \prob_i ), \quad \forall \iii,
\end{equation}
where the fact that $\ent (\qp \such \prob_i ) < \infty$ holds for all $\iii$ follows from \S \ref{subsec:geo_mean_prob}.
\end{thm}

The securities that agents obtain at Arrow-Debreu equilibrium described in \eqref{eq: opt_contr} provide  higher payoff on events where their individual subjective probabilities are higher than the ``geometric mean'' probability $\qp$ of \eqref{eq: qp}. In other words, discrepancies in beliefs result in allocations where agents receive higher payoff on their corresponding relatively more likely events.

Note also that the securities traded at Arrow-Debreu equilibrium have an interesting decomposition. Since $\U_i(\delta_i \log (\ud \prob_i / \ud \qp)) = - \delta_i \log \expec_{\prob_i} \bra{\ud \qp / \ud \prob_i} = 0 = \U_i(0)$, agent $\iii$ is indifferent between no trading and the first ``random'' part $\delta_i \log (\ud \prob_i / \ud \qp)$ of the security $\cp_i$. The second ``cash'' part $\delta_i \ent (\qp \such \prob_i )$ of $\cp_i$ is always nonnegative, and represents the monetary gain of agent $\iii$ resulting from the Arrow-Debreu transaction. After this transaction, the position of agent $\iii$ has certainty equivalent
\begin{equation} \label{eq:opt_util_price_agent_AD}
\up_i \dfn \U_i \pare{\cp_i}= \delta_i \ent (\qp \such \prob_i ), \quad \forall \iii.
\end{equation}
The aggregate agents' monetary value resulting from the Arrow-Debreu transaction equals
\begin{equation} \label{eq:opt_util_aggr_AD}
\up \dfn \sumi \up_i = \sumi \delta_i \ent (\qp \such \prob_i ).
\end{equation}

\begin{rem}
In the setting and notation of \S \ref{subsec: endow}, let $(E_i)_{\iii}$ be the collection of agents' random endowments. Furthermore, suppose that agents share common subjective beliefs; for concreteness, assume that $\widetilde{\prob}_i = \prob$, for all $i \in I$. In this case, and setting $E\dfn \sumi E_i$, the equilibrium valuation measure of \eqref{eq: qp} satisfies $\log \pare{\ud \qp / \ud \prob} \sim -E/\delta$ and equilibrium securities of \eqref{eq: opt_contr} are given by $\cp_i = \lambda_i E - E_i - \expecqp \bra{\lambda_i E - E_i}$, for all $\iii$. In particular, note the well-known fact that the payoff of each shared security is a linear combination of the agents' random endowments.
\end{rem}

\begin{rem} \label{rem:pareto}
Since $\cp_i / \delta_i \sim - \log(\ud \qp / \ud \prob_i)$, it is straightforward to compute
\begin{equation} \label{eq:relative_util}
\U_i (C_i) - \U_i (\cp_i) = - \delta_i \log \expecqp \bra{\exp \pare{ - \frac{C_i - \cp_i}{\delta_i}}}, \quad \forall C_i \in \lz, \quad \forall \iii.
\end{equation}
In particular, an application of Jensen's inequality gives $\U_i
(C_i) - \U_i (\cp_i) \leq \expecqp \bra{C_i - \cp_i} =
\expecqp \bra{C_i}$ for $C_i \in \Lb^1(\qp)$, with equality if and
only if $C_i \sim \cp_i$. The last inequality shows that $\cp_i$
is indeed the optimally-designed security for agent $\iii$ under
the valuation measure $\qp$. Furthermore, for any collection
$(C_i)_{\iii}$ with $\sum_{\iii} C_i = 0$ and $C_i \in \Lb^1(\qp)$
for all $\iii$, it follows that $\sum_{\iii} \U_i (C_i) \leq
\sum_{\iii} \U_i (\cp_i) = \up$. A standard argument using
the monotone convergence theorem extends the previous inequality
to
\[
\sum_{\iii} \U_i (C_i) \leq \sum_{\iii} \U_i (\cp_i), \quad \forall (C_i)_{\iii} \in \lzi \text{ with } \sum_{\iii} C_i = 0,
\]
with equality if and only if $C_i \sim \cp_i$ for all $\iii$.
Therefore, $(\cp_i)_{\iii}$ is a maximiser of the functional
$\sum_{\iii} \U_i (C_i)$ over all $(C_i)_{\iii} \in (\lz)^I$ with
$\sum_{\iii} C_i = 0$. In fact, the collection of all such
maximisers is $(z_i + \cp_i)_{\iii}$ where $(z_i)_{\iii}
\in \Real^I$ is such that $\sumi z_i = 0$. It can be shown that all Pareto optimal securities are exactly of this form; see e.g., \cite[Theorem 3.1]{JouiSchTou08} for a more general result.
Because of this Pareto optimality, the collection
$\pare{\qp, (\cp_i)_{\iii}}$ usually comes under the appellation
of (welfare) optimal securities and valuation measure,
respectively.

Of course, not every Pareto optimal allocation $(z_i + \cp_i)_{\iii}$, where $(z_i)_{\iii}$ is such that $\sumi z_i = 0$, is economically reasonable. A minimal ``fairness'' requirement that has to be imposed is that the position of each agent after the transaction is at least as good as the initial state. Since the utility comes only in the terminal time, we obtain the requirement $z_i
\geq - \up_i$, for all $\iii$. While there
may be many choices satisfying the latter
requirement in general, the choice $z_i = 0$ of Theorem \ref{thm: AD} has the cleanest economic interpretation in terms of complete financial market equilibrium.
\end{rem}

\begin{rem} \label{rem:no_trade_1}
If we ignore potential transaction costs, the cases where an agent has no motive to enter in a risk-sharing transaction are extremely rare. Indeed, agent $\iii$ will \emph{not} take part in the Arrow-Debreu transaction if and only
if $C_i = 0$, which happens when $\prob_i = \qp$. In
particular, agents will already be in Arrow-Debreu equilibrium and no transaction will take place if and only if they all share the same subjective beliefs.
\end{rem}

\section{Agents' Best Probability Response} \label{sec: best_response}

\subsection{Strategic behaviour in risk sharing} \label{subsec: strategic}

In the Arrow-Debreu setting, the resulting equilibrium is based on
the assumption that
agents do not apply any kind of strategic behaviour. However, in
the majority of practical risk-sharing situations, the modelling assumption of absence
of agents' strategic behaviour is unreasonable, resulting, amongst other things, in overestimation of market efficiency. When securities are negotiated among agents, their design and valuation will depend not only on their existing risky portfolios, but also on the beliefs about the future outcomes they will report for sharing. In general, agents will have incentive to report subjective beliefs that may differ from their true views about future uncertainty; in fact, these will also depend on subjective beliefs reported by the other parties.

As discussed in \S \ref{subsec: AD}, for a given set of agents' subjective beliefs, the optimal sharing rules are governed by the mechanism resulting in Arrow-Debreu equilibrium, as these are
the rules that efficiently allocate discrepancies of risks and beliefs among agents. It is then reasonable to assume that, in absence of a social planner, agents adapt this sharing mechanism for any collection $(\rprob_i)_{\iii} \in \PP^I$ of subjective probabilities  they choose to report---see also the related discussion in the introductory section). More precisely, in accordance to \eqref{eq: qp} and \eqref{eq: opt_contr}, the agreed-upon valuation measure $\qprob \in \PP$ is such that $\log \ud \qprob  \sim \sumi \lambda_i \log \ud \rprob_i$, and the collection of securities that agents will trade are $C_i \dfn \delta_i \log (\ud \rprob_i / \ud \qprob) + \delta_i \ent (\qprob \such \rprob_i )$, $\iii$.

Given the consistent with Arrow-Debreu equilibrium sharing rules, agents respond to subjective beliefs that other agents have reported, with the goal to maximise their individual utility. In this way, a game is formed, with the probability family $\PP$ being the agents' set of strategic choices. The subject of the present Section \ref{sec: best_response} is to analyse the behaviour of individual agents, establish their best response problem and show its well-posedness. The definition and analysis of the Nash risk-sharing equilibrium is taken up in Section \ref{sec: Nash}.

\subsection{Best response} \label{subsec:best_respo}

We shall now describe how agents respond to the reported
subjective probability assessments from their counterparties. For the purposes of \S \ref{subsec:best_respo}, we fix an agent $\iii$ and a collection of reported probabilities $\rprob_{-i} \dfn \pare{\rprob_j}_{j \in I \setminus \set{i}} \in \PP^{I \setminus \set{i}}$ of the remaining agents, and seek the subjective probability that is going to be submitted by agent $\iii$. According to the rules described in \S \ref{subsec: strategic}, a reported probability $\rprob_i \in \PP$ from agent $\iii$ will lead to entering a long position on the
security with payoff
\[
C_i \dfn \delta_i \log \pare{\ud \rprob_i / \ud \qprob^{(\rprob_{-i}, \rprob_i)} } + \delta_i \ent \pare{\qprob^{(\rprob_{-i}, \rprob_i)} \such \rprob_i },
\]
where $\qprob^{(\rprob_{-i}, \rprob_i)} \in \PP$ is such that
\[
\log \ud \qprob^{(\rprob_{-i}, \rprob_i)} \sim \lambda_i \log \ud \rprob_i + \sum_{j \in I \setminus \set{i}} \lambda_j \log \ud \rprob_j.
\]
By reporting subjective beliefs  $\rprob_i \in \PP$, agent $\iii$ also indirectly affects the geometric-mean valuation probability $\qprob^{(\rprob_{-i}, \rprob_i)}$, resulting in a highly non-linear overall effect in the security $C_i$. With the above understanding, and given  $\rprob_{-i} \dfn \pare{\rprob_j}_{j \in I \setminus \set{i}} \in \PP^{I \setminus \set{i}}$, the \emph{response function}  of agent $\iii$ is defined to be
\begin{align*}
\PP \ni \rprob_i \mapsto \V_i (\rprob_i; \rprob_{-i}) &\equiv \U_i \pare{\delta_i \log \pare{\ud \rprob_i / \ud \qprob^{(\rprob_{-i}, \rprob_i)} }   + \delta_i \ent \pare{\qprob^{(\rprob_{-i}, \rprob_i)} \such \rprob_i } } \\
&  = - \delta_i \log \expec_{\prob_i} \bra{\frac{\ud \qprob^{(\rprob_{-i}, \rprob_i)}}{\ud \rprob_i}} + \delta_i \ent \pare{\qprob^{(\rprob_{-i}, \rprob_i)} \such \rprob_i },
\end{align*}
where the fact that $\ent \pare{\qprob^{(\rprob_{-i}, \rprob_i)} \such \rprob_i } < \infty$ follows from the discussion of \S \ref{subsec:geo_mean_prob}. The problem of agent $\iii$ is to report the subjective probability that maximises the certainty equivalent of the resulting position after the transaction, i.e.,  to identify $\rb_i \in \PP$ such that
\begin{equation} \label{eq: best_response}
\V_i (\rb_i; \rprob_{-i}) = \sup_{\rprob_i \in \PP} \V_i (\rprob_i;
\rprob_{-i}).
\end{equation}
Any $\rb_i \in \PP$ satisfying \eqref{eq: best_response}
shall be called \emph{best probability response}.

In contrast to the majority of the related literature, the agent's strategic set of choices in our model may be of infinite dimension. This generalisation is important from a methodological viewpoint; for example, in the setting of \S \ref{subsec: endow} it allows for random endowments with infinite support, like ones with the Gaussian distribution or arbitrarily fat tails, a substantial feature in the modelling of risk.

\begin{rem} \label{rem:no_prob_constraints}
The best response problem \eqref{eq: best_response} imposes no constrains
on the shape of the agent's reported subjective probability, as long as it
belongs to $\PP$. In principle, it is possible for agents
to report subjective views that are considerably far from their actual ones. Such severe departures may be deemed unrealistic and are undesirable from a modelling point of view. However, as will be argued in \S \ref{subsubsec:never_true}, extreme responses are endogenously excluded in our set-up.
\end{rem}

We shall show in the sequel (Theorem \ref{thm:best_response}) that best responses in \eqref{eq: best_response} exist and are unique. We start with a result which gives necessary and sufficient conditions for best probability response.

\begin{prop} \label{prop: best_response_first_ord}
Fix $\iii$ and $\rprob_{-i} \equiv (\rprob_j)_{j \in I \setminus \set{i}} \in \PP^{I \setminus \set{i}}$. Then, $\rb_i \in \PP$ is best probability response for agent $\iii$ given $\rprob_{-i}$ if and only if the random variable $\cb_i \dfn \delta_i \log \pare{\ud \rb_i / \ud \qprob^{(\rprob_{-i}, \rb_i)} } + \delta_i \ent \pare{\qprob^{(\rprob_{-i}, \rb_i)} \such \rb_i }$ is such that $\cb_i > - \delta_{-i}$ and
\begin{equation} \label{eq:best_resp_foc}
\frac{\cb_i}{\delta_i} + \lambda_{-i} \log \pare{1 + \frac{\cb_i}{\delta_{-i}} } \sim  - \sum_{j \in I \setminus \set{i}} \lambda_j \log \pare{\frac{\ud \rprob_j}{\ud \prob_i}}.
\end{equation}
\end{prop}

The proof of Proposition \ref{prop: best_response_first_ord} is given in \S \ref{subsec: proof_of_best_resp_prop}. The necessity of the stated conditions for best response follows from applying first-order optimality conditions. Establishing the sufficiency of the stated conditions is certainly non-trivial, due to the fact that it is far from clear (and, in fact, not known to us) whether the response function is concave.

\begin{rem}\label{rem:no trade in best}
In the context of Proposition \ref{prop: best_response_first_ord}, rewriting \eqref{eq:best_resp_foc} we obtain that
\begin{equation} \label{eq:best_resp_foc-alt}
\frac{\cb_i}{\delta_i} + \lambda_{-i} \log \pare{1 + \frac{\cb_i}{\delta_{-i}} } \sim  - \log \pare{\frac{\ud \qprob^{(\rprob_{-i}, \rb_i)}}{\ud \prob_i}} + \lambda_i \log \pare{\frac{\ud \rb_i}{\ud \prob_i}}.
\end{equation}
Using also the fact that $\cb_i / \delta_i \sim \log \pare{\ud \rb_i / \ud \qprob^{(\rprob_{-i}, \rb_i)} }$, it follows that
\begin{equation} \label{eq:best_resp_prob_dens}
\log \pare{\frac{\ud \rb_i}{\ud \prob_i}} \sim - \log \pare{1 + \frac{\cb_i}{\delta_{-i}} }.
\end{equation}
Hence, $\rb_i = \prob_i$ holds if and only if $\log\pare{1 + \cb_i / \delta_{-i} } \sim 0$, which holds if and only if $\cb_i = 0$. (Note that $\cb_i \sim 0$ implies $\cb_i = 0$, since the expectation of $\cb_i$ under
$\qprob^{(\rprob_{-i}, \rb_i)}$ equals zero.) In words, the best probability response and actual subjective probabilities of an agent agree if and only if the agent has no incentive to participate in the risk-sharing transaction, given the reported subjective beliefs of other agents. Hence, in \emph{any} non-trivial cases, agents' strategic behaviour implies a departure from reporting their true beliefs.

Plugging \eqref{eq:best_resp_prob_dens} back to \eqref{eq:best_resp_foc-alt}, and using also \eqref{eq:best_resp_foc}, we obtain
\begin{equation} \label{eq:best_resp_qprob_dens}
\log \pare{\frac{\ud \qprob^{(\rprob_{-i}, \rb_i)}}{\ud \prob_i}} \sim - \frac{\cb_i}{\delta_i} - \log \pare{1 + \frac{\cb_i}{\delta_{-i}} } \sim - \lambda_i \log \pare{1 + \frac{\cb_i}{\delta_{-i}} }  + \sum_{j \in I \setminus \set{i}} \lambda_j \log \pare{\frac{\ud \rprob_j}{\ud \prob_i}},
\end{equation}
providing directly the valuation measure $\qprob^{(\rprob_{-i}, \rb_i)}$ in terms of the security $\cb_i$.
\end{rem}

\begin{rem}\label{rem: probability of best response}
A message from \eqref{eq:best_resp_prob_dens} is that, according to their best response process, agents will report beliefs that understate (resp., overstate) the probability of their payoff being
high (resp., low) relatively to their true beliefs. Such behaviour is clearly driven by a desired post-transaction utility increase.

More importantly, and in sharp contrast to the securities $\pare{\cp_i}_{\iii}$ formed in Arrow-Debreu equilibrium, the security that agent $\iii$ wishes to enter, after taking into account the aggregate reported beliefs of the rest and declaring subjective probability $\rb_i$, has limited liability, as it is bounded from below by the constant $-\delta_{-i}$.
\end{rem}

\begin{rem}\label{rem:best response different belief}
Additional insight regarding best probability responses may be obtained resorting to the discussion of \S\ref{subsec: endow}, where $\prob_i$ incorporates the random endowment $E_i \in \lz$ of agent $\iii$, in the sense that $\log \big( \ud \prob_i / \ud \widetilde{\prob}_i \big) \sim - E_i/\delta_i$, where $\widetilde{\prob}_i$ denotes the subjective probability of agent $\iii$. It follows from \eqref{eq:best_resp_prob_dens} that $ \log \big( \ud \rb_i/\ud \widetilde{\prob}_i \big) \sim -E_i/\delta_i - \log \pare{1 + \cb_i/\delta_{-i}}$. It then becomes apparent that, when agents share their risky endowment, they tend to put more weight on the probability of
the downside of their risky exposure, rather than the upside. For an illustrative situation, see Example \ref{ex: best response} later on.
\end{rem}

\begin{rem}
In the course of the proof of Proposition \ref{prop: best_response_first_ord}, the constant in the equivalence \eqref{eq:best_resp_foc} is explicitly computed; see \eqref{eq:best_resp_z}. This constant has a particularly nice economic interpretation in the case of two agents. To wit, let $I = \set{0,1}$, and suppose that $\rprob_1 \in \PP$ is given. Then, from the vantage point of agent $0$, \eqref{eq:best_resp_foc} becomes
\[
\frac{\cb_0}{\delta_0} + \lambda_1 \log \pare{1 + \frac{\cb_0}{\delta_1} } = \zeb_0  - \lambda_1 \log \pare{\frac{\ud \rprob_1}{\ud \prob_0}},
\]
where the constant $\zeb_0 \in \Real$ is such that
\[
\zeb_0 = - \log \expec_{\prob_0} \bra{\exp \pare{- \frac{\cb_0}{\delta_0} }} + \log \expec_{\rprob_1} \bra{\exp \pare{ \frac{\cb_0}{\delta_{1}} } } = \frac{\U_0 (\cb_0)}{\delta_0} - \frac{\U_1 (-\cb_0; \rprob_1)}{\delta_1}.
\]
where $\U_1 (\cdot; \rprob_1)$ denotes the utility functional of a ``fictitious'' agent with representative pair $(\delta_1, \rprob_1)$. In words, $\zeta_0$ is the post-transaction difference, denominated in units of risk tolerance, of the utility of agent $0$  from the utility of agent $1$ (who obtains the security $- \cb_0$), provided that the latter utility is measured with respect to the reported, as opposed to subjective, beliefs of agent $1$. In particular, when agent $1$ does not behave strategically, in which case $\rprob_1 = \prob_1$, it holds that $\zeb_0 = \U_0 (\cb_0) / \delta_0 -  \U_1 (-\cb_0) / \delta_1$.
\end{rem}

Proposition \ref{prop: best_response_first_ord} sets a roadmap for
proving existence and uniqueness in the best response problem via a one-dimensional parametrisation. Indeed, in accordance to \eqref{eq:best_resp_foc}, in order to find a best response we consider for each $z_i \in \Real$ the unique random
variable $C_i(z_i)$ that satisfies the equation $C_i(z_i) / \delta_i + \lambda_{-i} \log \pare{1 + C_i (z_i) / \delta_{-i} } = \lambda_{-i} z_i  - \sum_{j \in I \setminus \set{i}} \lambda_j \log \pare{ \ud \rprob_j / \ud \prob_i}$; then, upon defining $\qprob_i (z_i)$ via $\log \pare{ \ud \qprob_i(z_i) / \ud \prob_i} \sim - \lambda_i \log \pare{1 + C_i (z_i) / \delta_{-i} }  + \sum_{j \in I \setminus \set{i}} \lambda_j \log \pare{ \ud \rprob_j / \ud \prob_i}$ in accordance to \eqref{eq:best_resp_qprob_dens}, we seek $\widehat{z}_i \in \Real$ such that
$C_i(\widehat{z}_i) \in \Lb^1(\qprob_i(\widehat{z}_i))$ and $\expec_{\qprob_i(\widehat{z}_i)} \bra{C_i(\widehat{z}_i)} = 0$ hold. It turns out that there is
a unique such choice; once found, one simply defines $\rb_i$ via $\log \pare{\ud \rb_i / \ud \prob_i} \sim - \log \pare{1 + C_i(\widehat{z}_i) / \delta_{-i} }$, in accordance to \eqref{eq:best_resp_prob_dens}, to obtain the
unique best response of agent $\iii$ given $\rprob_{-i}$. The technical details of the proof of
Theorem \ref{thm:best_response} below are given in \S
\ref{subsec:proof_of_best_resp_thm}.

\begin{thm} \label{thm:best_response}
For $\iii$ and $\rprob_{-i} \equiv (\rprob_j)_{j \in I \setminus \set{i}} \in \PP^{I \setminus \set{i}}$, there exists a unique $\rb_i \in \PP$ such that $\V_i (\rb_i; \rprob_{-i}) = \sup_{\rprob_i \in \PP} \V_i (\rprob_i; \rprob_{-i})$.
\end{thm}

\subsection{The value of strategic behaviour} \label{subsec:market_power}

The increase on agents' utility that is caused by following the best probability response procedure can be regarded as a measure for the value of the strategic behaviour induced by problem \eqref{eq: best_response}. Consider for example the case where only a single agent (say) $0 \in I$ applies the best probability response strategy and the rest of the agents report their true beliefs, i.e., $\rprob_j = \prob_j$ holds for $j \in I \setminus \set{0}$. As mentioned in the introductory section, this is a potential model of a transaction where only agent 0 possesses meaningful market power. Based on the results of \S
\ref{subsec:best_respo}, we may calculate the gains, relative to the Arrow-Debreu transaction, that agent $0$ obtains by incorporating such strategic behaviour (which, among others, implies limited liability of the security the agent takes a long position in). The main insights are illustrated in the following two-agent example.

\begin{exa}\label{ex: best response}
Suppose that $I = \set{0,1}$ and $\delta_0 = 1 = \delta_1$. We shall use the set-up of \S \ref{subsec: endow}, where for simplicity it is assumed that agents have the same subjective probability measure. The agents are exposed to random endowments $E_0$ and $E_1$ that (under the common probability measure) have Gaussian law with mean zero and common variance $\sigma^2 > 0$, while $\rho \in [-1,1]$ denotes the correlation coefficient of $E_0$ and $E_1$. In this case, it is straightforward to check that $\cp_0 = (E_1-E_0)/2$; therefore, after the Arrow-Debreu transaction, the position of agent $0$ is $E_0 + \cp_0 = (E_0 + E_1) / 2$. On the other hand, if agent 1 reports true beliefs, from \eqref{eq:best_resp_foc} the
security $\cb_0$ corresponding to the best probability response of
agent $0$ should satisfy $2 \cb_0 +  \log \pare{1 + \cb_0} = \zeb_0 +E_1$ for appropriate $\zeb_0 \in \Real$ that is coupled with $\cb_0$. For
$\sigma^2 = 1$ and $\rho = -0.5$, straightforward Monte-Carlo
simulation allows the numerical approximation of the probability
density functions (pdf) of $E_0$ and $E_1$  under the best response probability $\rb_0$, illustrated in Figure \ref{fig: rpdf}. As is apparent, the best probability response drives agent 0 in overstating the downside risk of $E_0$ and understating the downside risk of $E_1$.
\begin{figure}[!ht]
\includegraphics[trim = 20mm 0mm 0mm 0mm, clip, scale=0.5]{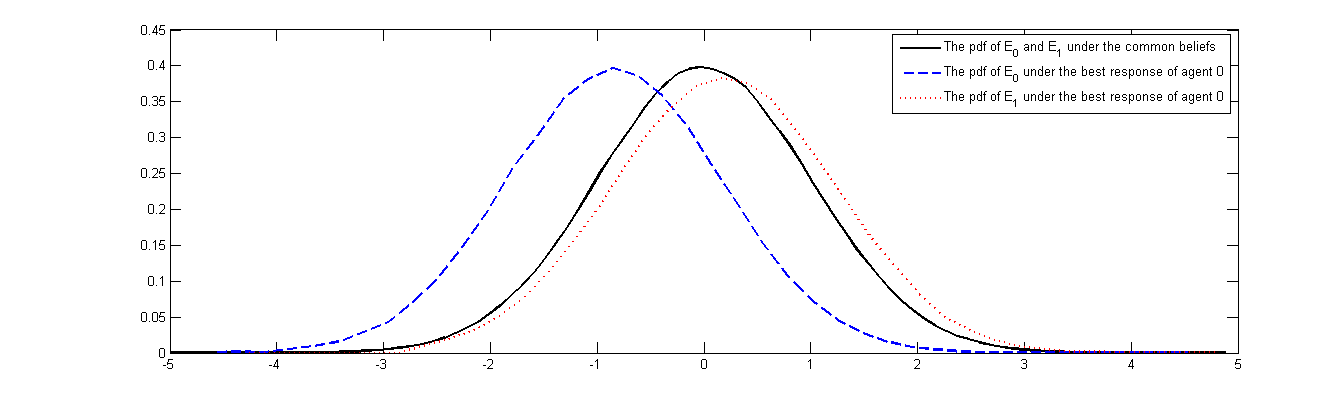}
\caption{{\footnotesize The solid black line is the pdf of endowments $E_0$ and $E_1$ under the agents' common subjective probability measure, while the other
curves illustrate the pdf of $E_0$ (dashed blue) and $E_1$ (dotted red) under the best probability response of agent 0. In
this example, $\sigma^2 = 1$ and $\rho=-0.5$.}}
\label{fig: rpdf}
\end{figure}

The effect of following such strategic behaviour is depicted in Figure \ref{fig: exa1}, where there is comparison between the probability density functions of the positions of agent 0 under (i) no trading; (ii) the Arrow-Debreu transaction; and (iii) the transaction following the application of best response strategic behaviour. As compared to the Arrow-Debreu position, the lower bound of the security $\cb_0$ guarantees a heavier right tail of the agent's position after the best response transaction.

\begin{figure}[!ht]
\includegraphics[trim = 20mm 0mm 0mm 0mm, clip, scale=0.5]{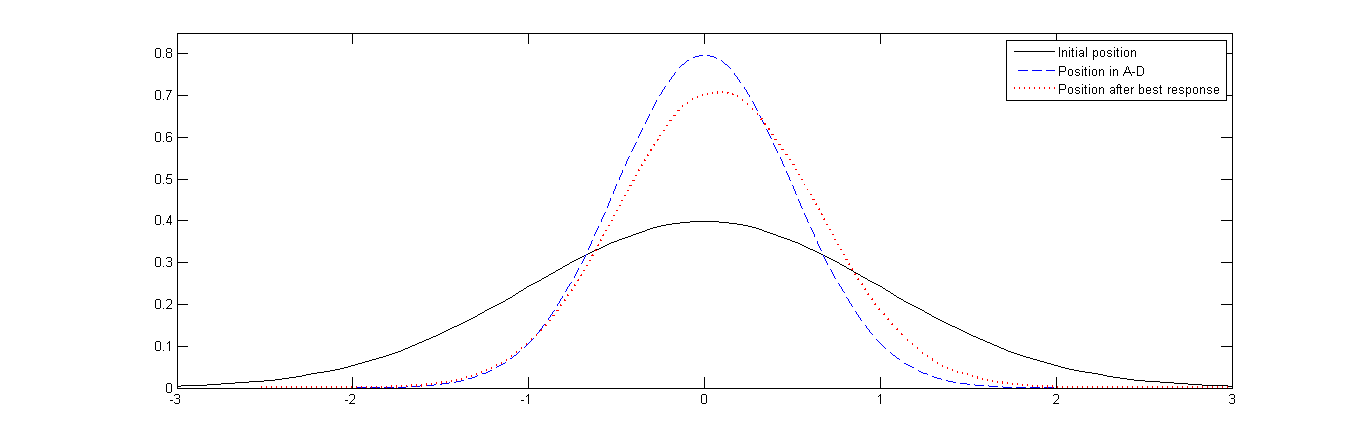}
\caption{{\footnotesize The solid black line is the pdf of the initial position $E_0$, the dashed blue line illustrates the pdf of the position $E_0+\cp_0$ and the dotted red line is the pdf of the position $E_0+\cb_0$, all under the common subjective probability measure. In
this example, $\sigma^2 = 1$ and $\rho=-0.5$.}}
\label{fig: exa1}
\end{figure}
\end{exa}

\section{Nash Risk-Sharing Equilibrium}\label{sec: Nash}

We shall now consider the situation where every single
agent follows the same strategic behaviour indicated by the best
response problem of Section \ref{sec: best_response}. As previously mentioned, sharing securities are designed following the sharing rules determined by Theorem \ref{thm: AD} for any collection of reported subjective views. With the well-posedness of the best response problem established, we are now ready to examine whether the game among agents has an equilibrium point. In view of the analysis of Section \ref{sec: best_response}, individual agents have motive to declare subjective beliefs different than the actual ones. (In particular, and in the setting of \S \ref{subsec: endow}, agents will tend to overstate the probability of their random endowments taking low values.) Each agent will act according to the best response mechanism as in \eqref{eq: best_response}, given what other agents have reported as subjective beliefs. In a sense, the best response mechanism indicates a negotiation scheme, the fixed point (if such exists) of which will produce the Nash equilibrium valuation measure and risk-sharing securities.

Let us emphasise that the actual subjective beliefs of individual players are \emph{not} necessarily assumed to be private knowledge; rather, what is assumed here is that agents have agreed upon the rules that associate any reported subjective beliefs to securities and prices, even if the reported beliefs are not the actual ones. In fact, even if subjective beliefs constitute private knowledge initially, certain information about them will necessarily be revealed in the negotiation process which will lead to Nash equilibrium. 

There are two relevant points to consider here. Firstly, it is unreasonable for participants to attempt to invalidate the negotiation process based on the claim that other parties do not report their true beliefs, as the latter is, after all, a \emph{subjective} matter. This particular point is reinforced from the \emph{a posteriori} fact that reported subjective beliefs in Nash equilibrium do not deviate far from the true ones, as was pointed out in Remark \ref{rem:no_prob_constraints} and is being further elaborated in \S \ref{subsubsec:never_true}. Secondly, it is exactly the limited number of participants, rather than private or asymmetric information, that gives rise to strategic behaviour: agents recognise their ability to influence the market, since securities and valuation become output of collective reported beliefs. Even under the appreciation that other agents will not report true beliefs and the negotiation will not produce an Arrow-Debreu equilibrium, agents will still want to reach a Nash equilibrium, as they will improve their initial position.   
In fact, transactions with limited number of participants typically equilibrate far from their competitive equivalents, as has been also highlighted in other models of thin financial markets with symmetric information structure, like the ones in \cite{CarRosWer12} and \cite{RosWer15}---see also the related discussion in the introductory section.

\subsection{Revealed subjective beliefs}\label{subsec:revealed}
Considering the model from a more pragmatic point of view, one may argue that agents do not actually report subjective beliefs, but rather agree on a valuation measure $\qprob \in \PP$ and zero-price sharing securities $(C_i)_{\iii}$ that clear the market. However, there is a
one-to-one correspondence between reporting subjective beliefs and proposing a valuation measure and securities, as will be described below.

From the discussion of \S \ref{subsec: strategic}, a collection of subjective probabilities $(\rprob_i)_{\iii}$ gives rise to valuation measure $\qprob \in \PP$ such that $\log \ud \qprob  \sim \sumi \lambda_i \log \ud \rprob_i$ and collection $(C_i)_{\iii}$ of securities is such that $C_i \dfn \delta_i \log (\ud \rprob_i / \ud \qprob) + \delta_i \ent (\qprob \such \rprob_i )$, for all $\iii$. Of course, $\sumi C_i = 0$ and $\expecq \bra{C_i} = 0$ holds for all $\iii$. A further technical observation is that $\exp(C_i / \delta_i) \in \Lb^1(\qprob)$ holds for all $\iii$, which is then a necessary condition that an arbitrary collection of market-clearing securities $(C_i)_{\iii}$ must satisfy with respect to an arbitrary valuation probability $\qprob \in \PP$ in order to {be consistent with the  aforementioned risk-sharing mechanism}. The previous observations lead to a definition: for $\qprob \in \PP$, we define the class $\C_\qprob$ of securities that clear the market and are consistent with the valuation measure $\qprob$ via
\[
\C_\qprob \dfn \set{ (C_i)_{\iii} \in \lzi \ \Big| \ \sumi C_i =
0, \text{ and } \exp(C_i / \delta_i) \in \Lb^1(\qprob), \ \expecq \bra{C_i} = 0, \ \forall \iii}.
\]
Note that all expectations of $C_i$ under $\qprob$ in the definition of 
$\C_\qprob$ above are well defined. Indeed, the fact that $\exp(C_i / \delta_i) \in \Lb^1(\qprob)$ in the definition of $\C_\qprob$ implies that $\pare{C_i}_+ \in \Lb^1(\qprob)$ for all $\iii$. From $\sumi C_i = 0$, we obtain $\sumi|C_i| = 2 \sumi \pare{C_i}_+$ and hence $C_i \in \Lb^1(\qprob)$ for all $\iii$.

Starting from a given valuation
measure $\qprob \in \PP$ and securities $(C_i)_{\iii} \in
\C_\qprob$, one may define a collection $(\rprob_i)_{\iii} \in \PP^I$ via $\log (\ud \rprob_i / \ud \qprob) \sim C_i / \delta_i$ for $\iii$, and note that this is the unique collection in $\PP^I$ that results in the valuation probability $\qprob$ and securities $(C_i)_{\iii}$. In this way, the probabilities
$(\rprob_i)_{\iii} \in \PP^I$ can be considered as \emph{revealed} by the valuation measure $\qprob \in
\PP$ and securities $ (C_i)_{\iii} \in \C_\qprob$. Hence, agents proposing risk-sharing securities and a valuation measure is equivalent to them reporting probability beliefs in the transaction. This viewpoint justifies and underlies Definition \ref{def:nash} that follows: the objects of Nash equilibrium are the valuation measure and designed securities, in consistency with the definition of Arrow-Debreu equilibrium.

\subsection{Nash equilibrium and its characterisation}

Following classic literature, we give the formal definition of a Nash risk-sharing equilibrium.

\begin{defn} \label{def:nash}
The collection $\pare{\qg,(\cg_i)_{\iii}} \in \PP \times \lzi$
will be called a \textbf{Nash equilibrium} if
$(\cg_i)_{\iii} \in \C_{\qg}$ and, with $\log (\ud \rg_i / \ud \qg) \sim \cg_i / \delta_i$ for all $\iii$ denoting the
corresponding revealed subjective beliefs, and $\rg_{-i} \dfn (\rg_j)_{j \in I \setminus \set{i}}$ for $\iii$, it holds that
\[
\V_i \pare{\rg_i; \rg_{-i}} = \sup_{\rprob_i \in \PP} \V_i \pare{\rprob_i ;
\rg_{-i}}, \quad \forall \iii.
\]
\end{defn}

A use of Proposition \ref{prop: best_response_first_ord} results
in the characterisation Theorem \ref{thm: nash} below, the proof
of which is given in \S \ref{subsec: proof_of_nash_thm}. For this,
we need to introduce the $n$-dimensional Euclidean space
\begin{equation} \label{eq:simplex}
\DI = \set{z \in \Real^I \ \Big| \ \sum_{\iii} z_i = 0}.
\end{equation}

\begin{thm} \label{thm: nash}
The collection $\pare{\qg, (\cg_i)_{\iii}} \in \PP \times \lzi$ is a Nash equilibrium if and only if the following three conditions hold:
\begin{enumerate}
    \item[(N1)] $\cg_i > - \delta_{-i}$ for all $\iii$, and there exists $\zg = \pare{\zg_i}_{\iii} \in \DI$ such that
\begin{equation} \label{eq:equil_C}
 \cg_i +  \delta_i \log \pare{1 + \frac{\cg_i}{\delta_{-i}} } = \zg_i + \cp_i +  \delta_i \sum_{\jii} \lambda_j \log \pare{1 + \frac{\cg_j}{\delta_{-j}} }, \quad \forall \iii;
\end{equation}
    \item[(N2)] with $\qp \in \PP$ as in \eqref{eq: qp}, i.e., such that $\log \ud \qp  \sim \sumi \lambda_i \log \ud \prob_i$, it holds that
\begin{equation} \label{eq:equil_Q}
    \log \pare{\frac{\ud \qg}{\ud \qp}} \sim - \sum_{\jii} \lambda_j \log \pare{1 + \frac{\cg_j}{\delta_{-j}} };
\end{equation}
    \item[(N3)] $\expecqg \bra{\cg_i}=0$ holds for all $\iii$.
\end{enumerate}
\end{thm}

\begin{rem} \label{rem:no_trade_2}
Suppose that the agents' preferences and risk exposures are such that no trade occurs in Arrow-Debreu equilibrium, which happens when all $\prob_i$ are the same (and equal to, say, $\prob$) for all $\iii$---see Remark \ref{rem:no_trade_1}. In this case, $\qp = \prob$ and $\cp_i = 0$ for all $\iii$. It is then straightforward from Theorem \ref{thm: nash} to see that a Nash equilibrium is also given by $\qg = \prob$ and $\cg_i = 0$ (as well as $\zg_i = 0$) for all $\iii$. In fact, as will be argued in \S \ref{subsubsec: z in compact set}, this is the unique Nash equilibrium in this case. Conversely, suppose that a Nash equilibrium is given by $\qg = \prob$ and $\cg_i = 0$ for all $\iii$. Then, \eqref{eq:equil_Q} shows that $\qp = \qg = \prob$ and \eqref{eq:equil_C} implies that $\cp_i \sim - \zg_i \sim 0$, which means that $\cp_i = 0$ for all $\iii$. In words, the Nash risk-sharing equilibrium involves no risk transfer if and only if the agents are already in a Pareto optimal situation.
\end{rem}

In the important case of two acting agents, since $\cg_0 = - \cg_1$, applying simple algebra in \eqref{eq:equil_C}, we obtain that a Nash equilibrium risk sharing security $\cg_0$ is such that $-\delta_1 < \cg_0 < \delta_0$ and satisfies
\begin{equation}\label{eq: C_two_agent}
\cg_0 + \frac{\delta_0 \delta_1}{\delta} \log \pare{ \frac{1 +
\cg_0 / \delta_1}{1 - \cg_0 / \delta_0}} = \zg_0 + \cp_0.
\end{equation}
In Theorem \ref{thm:nash_exist_uniq}, existence of a unique Nash equilibrium for the two-agent case will be shown. Furthermore, a one-dimensional root-finding algorithm presented in \S \ref{subsec:root_finding} allows to calculate the Nash equilibrium, and further calculate and compare the final
position of each individual agent. Consider for
instance Example \ref{ex: best response} and its symmetric situation that is illustrated in Figure \ref{fig: exa1}, where the limited liability of the security $\cb_0$ implies less variability and flatter right tail of the agent's position. Under the Nash equilibrium, as will be argued in \S \ref{subsubsec:endog_bounds}, security $\cg_0$ is
further bounded from above, which implies that the probability density function of agent's final position is shifted to the left. This fact is illustrated in Figure \ref{fig: exa2}.

\begin{figure}[!ht]
\includegraphics[trim = 20mm 0mm 0mm 0mm, clip, scale=0.5]{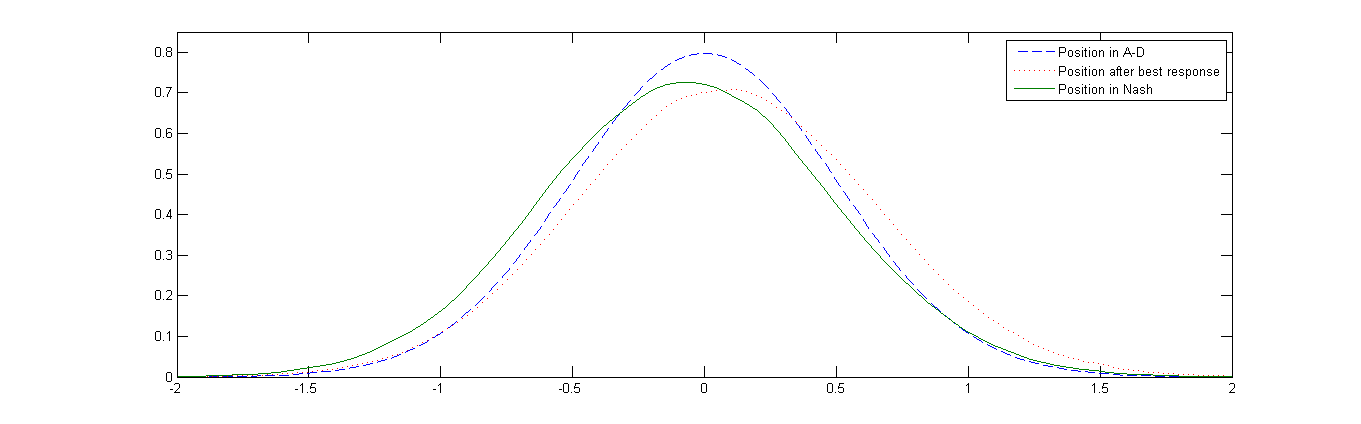}
\caption{{\footnotesize The solid green line is the pdf of the
position $E_0+\cg_0$, the dashed blue line illustrates the pdf
of the position $E_0+\cp_0$ and the dotted red line is the pdf
of the position $E_0+\cb_0$, all under the common subjective probability measure. In this example,
$\sigma_0=\sigma_1=1$, and $\rho=-0.5$.}} \label{fig: exa2}
\end{figure}

Despite the above symmetric case, it is not necessary true that
all agents suffer a loss of utility at the Nash equilibrium risk
sharing. As we will see in the Section \ref{sec: extrema}, for agents with sufficiently large risk tolerance the negotiation game
results in higher utility compared to the one gained through Arrow-Debreu equilibrium.

\subsection{Within equilibrium}\label{subsec: in Nash}

According to Theorem \ref{thm:nash_exist_uniq}, Nash equilibria in the sense of Definition \ref{def:nash} always exist. Throughout \S \ref{subsec: in Nash}, we assume that $\pare{\qg, (\cg_i)_{\iii}}$ is a Nash equilibrium and provide a discussion on certain aspects of it, based on the characterisation Theorem \ref{thm: nash}.

\subsubsection{Endogenous bounds on traded securities} \label{subsubsec:endog_bounds}
As was pointed in Remark \ref{rem: probability of best response}, the security that each agent enters resulting from the best response procedure is
bounded below. When all participating agents
follow the same strategic behaviour, Nash equilibrium securities
are bounded from above as well. Indeed, since the market clears, the security that agents take a long position into is shorted by the rest of the agents, who similarly intend to bound their liabilities. Mathematically, since $\cg_i > -
\delta_{-i}$ is valid for all $\iii$ and $\sum_{\iii} \cg_i = 0$
holds, it also follows that $\cg_i = - \sum_{\jii \setminus
\set{i}} \cg_{j} < \sum_{\jii \setminus \set{i}} \delta_{-j} = (n
-1) \delta + \delta_i$, for all $\iii$. Therefore, a consequence of the agents' strategic behaviour is that Nash risk-sharing securities are \emph{endogenously} bounded. This fact is in sharp contrast with the Arrow-Debreu equilibrium of \eqref{eq: opt_contr}, where the risk transfer may involve securities with unbounded payoffs. An immediate consequence of the bounds on the securities is that the potential gain from the Nash risk-sharing transaction is also endogenously bounded. Naturally, the resulting endogenous bounds are an indication of how the game among agents restricts the risk-sharing transaction, which in turn may be a source of large loss of efficiency. The next example is an illustration of the such inefficiency in a simple symmetric setting. Later on, in Figure \ref{fig: exa2}, the loss of utility in another two-agent
example is visualised.

\begin{exa} 
Let $X \in \lz$ have the standard (zero mean, unit standard deviation) Gaussian law under the baseline probability $\prob$. For $\beta \in \Real$, define $\prob^\beta \in \PP$ via $\log \pare{\ud \prob^\beta / \ud \prob }\sim \beta X$; under $\prob^\beta$, $X$ has the Gaussian law with mean $\beta$ and unit standard deviation. Fix $\beta > 0$, and set $\prob_0 \dfn \prob^\beta$ and $\prob_1 \dfn \prob^{-\beta}$. In this case, it is straightforward to compute that $\cp_0
= \beta X = - \cp_1$. It also follows that $\up_0 = \beta^2 / 2 = \up_1$. If $\beta$ is large, the discrepancy between the agents' beliefs results in large monetary profits to both after the Arrow-Debreu transaction. On the other hand, as will be established in Theorem \ref{thm:nash_exist_uniq}, in case of two agents there exists a unique Nash equilibrium. In fact, in this symmetric case we have that $-1 < \cg_0 < 1$, and it can be checked that (see also \eqref{eq: C_two_agent} later)
\[
\cg_0 + \frac{1}{2} \log \pare{\frac{1 + \cg_0}{1 - \cg_0}} = \beta X.
\]
The loss of efficiency caused by the game becomes greater with increasing values of $\beta > 0$. In fact, if $\beta$ converges
to infinity, it can be shown that $\cg_0$ converges to
$\mathsf{sign} (X) = \indic_{\set{X > 0}} - \indic_{\set{X < 0}}$; furthermore, both $\U_0 \pare{\cg_0}$ and $\U_1 \pare{\cg_1}$ will converge to $1$, which demonstrates the tremendous inefficiency of the Nash equilibrium transaction as compared to the Arrow-Debreu one.
\end{exa}

Note that the endogenous bounds $- \delta_{-i} < \cg_i < (n -1) \delta + \delta_i$ depend only on the risk tolerance profile of the agents,
and \emph{not} on their actual beliefs (or risk exposures). In addition,
these bounds become stricter in games where quite risk-averse agents are playing, as they become increasingly hesitant towards undertaking risk.

\subsubsection{If trading, you never reveal your true beliefs} \label{subsubsec:never_true}
As discussed in Remark \ref{rem:no trade in best}, agents' best probability response differ from their actual subjective beliefs in any situation where risk transfer is involved. This
result becomes more pronounced when we consider the Nash
risk-sharing equilibrium. To wit, if $(\rg_i)_{\iii}$ are revealed subjective beliefs corresponding to a Nash equilibrium, it is as a consequence of Theorem \ref{thm: nash} (see also \eqref{eq:best_resp_prob_dens}) that
\begin{equation} \label{eq:reported_nash_agent}
\log \pare{\frac{\ud \rg_i}{\ud \prob_i}} \sim - \log \pare{1 + \frac{\cg_i}{\delta_{-i}}}, \quad \forall \iii.
\end{equation}
Note that $\rg_i = \prob_i$ holds if and only if $\cg_i = 0$ for any fixed $\iii$; therefore, whenever agents take part (by actually trading) in Nash equilibrium, their reported subjective beliefs are \emph{never} the same as their actual ones.

Even though in any non-trivial trading situation agents will
report different subjective beliefs from their actual ones, we
shall argue below that  \eqref{eq:reported_nash_agent} imposes
\textit{endogenous} constraints on the magnitude of the possible
discrepancy; the discussion the follows expands on Remark \ref{rem:no_prob_constraints}. Start by writing \eqref{eq:reported_nash_agent} as
$\log \pare{ \ud \prob_i / \ud \rg_i } = - \kappa^\diamond_i +
\log \pare{1 + \cg_i / \delta_{-i}}$, where $\kappa^\diamond_i
\dfn \log \expec_{\rg_i} \bra{1 + \cg_i / \delta_{-i}}$, and note
that $\expec_{\rg_i} \bra{\cg_i} \geq - \delta_i \log
\expec_{\rg_i} \bra{\exp \pare{ -\cg_i / \delta_i}}  \geq 0$
holds, where we have used Jensen's inequality and the fact that
$\pare{\qg, (\cg_i)_{\iii}}$ is an Arrow-Debreu equilibrium for
the fictitious agents' preference pairs $(\delta_i,
\rg_i)_{\iii}$. It follows that $\kappa^\diamond_i \geq 0$ holds,
which implies that $\ud \prob_i / \ud \rg_i \leq 1 + \cg_i /
\delta_{-i}$, for all $\iii$. Defining weights $(\alpha_i)_{\iii}$
via $\alpha_i \dfn \delta_{-i}/n\delta = \lambda_{-i}/n$ for all
$\iii$ (noting that $0 < \alpha_i < 1/n$ holds for all $\iii$, and
that $\sum_{\iii} \alpha_i = 1$), a use of the market-clearing
condition $\sumi \cg_i = 0$ gives $\sumi \alpha_i \pare{\ud
\prob_i / \ud \rg_i}  \leq 1$. One can obtain a corresponding
lower bound. Indeed, using the endogenous bounds $\cg_i \leq (n-1)
\delta + \delta_i$, it follows that $\kappa^\diamond_i \leq - \log
\alpha_i$ for all $\iii$, which gives $\ud \prob_i / \ud \rg_i
\geq \alpha_i (1 + \cg_i / \delta_{-i}) = \alpha_i + \cg_i / (n
\delta)$. Using again the market-clearing condition $\sumi \cg_i =
0$, it follows that $\sumi \pare{\ud \prob_i / \ud \rg_i} \geq 1$.
To recapitulate,
\[
\sumi \alpha_i \frac{\ud \prob_i}{\ud \rg_i}  \leq 1 \leq \sumi \frac{\ud \prob_i}{\ud \rg_i}
\]
holds, which imposes considerable a-priori restrictions on the
likelihood ratios $\ud \prob_i / \ud \rg_i$ for all $\iii$. (For example, there are no events for which all agents will overstate or understate their likelihood as compared to their actual subjective beliefs.) In
particular, since $1 / \alpha_i = n / \lambda_{-i}$, we obtain
that
\begin{equation} \label{eq:lower_bound_dens}
\frac{\ud \prob_i}{\ud \rg_i} \leq \frac{n}{\lambda_{-i}}, \quad \forall \iii.
\end{equation}
The above upper bound on the likelihood of $\prob_i$ with respect to $\rg_i$ only depends on the number of remaining agents $n$ and the relative risk tolerance coefficient of the agents; it does not depend neither the aggregate risk tolerance level $\delta$ nor the actual subjective beliefs of other agents. Furthermore, note also that bound \eqref{eq:lower_bound_dens} implies that $\ent(\prob_i \such  \rg_i) = \expec_{\prob_i} \bra{\log(\ud \prob_i / \ud \rg_i)} \leq \log (n / \lambda_{-i})$. The latter gives an a-priori endogenous estimate on the distance of the truth from the reported beliefs in Nash equilibrium.


\subsubsection{Loss of efficiency} \label{subsubsec:loss_of_eff}

As already mentioned, agents' strategic behaviour results in risk-sharing \textit{inefficiency}, which, since utilities $(\U_i)_{\iii}$ are
numerically represented by certainty equivalents, can be measured through the difference of the aggregate monetary utility under the Arrow-Debreu transaction and the aggregate  monetary  utility under the Nash equilibrium risk-sharing transaction. Note that similar measures of inefficiency have been used in risk-sharing literature---see e.g., \cite{Vay99}
or \cite{AchBis05}. Mathematically, the loss of efficiency
equals $\up - \ug = \sumi \up_i - \sumi \ug_i$, where
$(\up_i)_{\iii}$ and $\up$ are defined in
\eqref{eq:opt_util_price_agent_AD} and
\eqref{eq:opt_util_aggr_AD}, while
\[
\ug_i \dfn \U_i \pare{\cg_i}, \quad \forall \iii, \quad \text{and} \quad \ug \dfn \sumi \ug_i.
\]
From \eqref{eq:relative_util}, \eqref{eq:equil_C} and \eqref{eq:equil_Q}, it follows that
\begin{align*}
\nonumber \ug_i - \up_i &= - \delta_i \log \expecqp \bra{\exp \pare{ - \frac{\cg_i - \cp_i}{\delta_i}}} \\
&= \zg_i - \delta_i \log \expecqp \bra{\pare{1 + \frac{\cg_i}{\delta_{-i}}} \prod_{\jii} \pare{1 + \frac{\cg_j}{\delta_{-j}}}^{-\lambda_j}}\\
&= \zg_i - \delta_i \log \expecqg \bra{1 + \frac{\cg_i}{\delta_{-i}}} - \delta_i \log \expecqp \bra{\prod_{\jii} \pare{1 + \frac{\cg_j}{\delta_{-j}}}^{-\lambda_j}}, \quad \forall \iii.
\end{align*}
Recalling that $\expecqg \bra{\cg_i} = 0$ holds for all $\iii$,
and noting the equality
\[
\expecqp \bra{\prod_{\jii} \pare{1 + \frac{\cg_j}{\delta_{-j}}}^{-\lambda_j}} = \expecqg \bra{\prod_{\jii} \pare{1 + \frac{\cg_j}{\delta_{-j}}}^{\lambda_j}}^{-1}
\]
which holds in view of \eqref{eq:equil_Q}, we obtain
\begin{equation} \label{eq:util_loss_nash_indiv}
\ug_i - \up_i = \zg_i + \lambda_i \delta \log \expecqg \bra{\prod_{\jii} \pare{1 + \frac{\cg_j}{\delta_{-j}}}^{\lambda_j}}, \quad \forall \iii.
\end{equation}
Adding up \eqref{eq:util_loss_nash_indiv} over all $\iii$ and using the fact that $\sum_{\iii} \zg_i = 0$, one obtains an analytic expression of the loss of efficiency caused by the game:
\begin{equation} \label{eq:util_loss_nash}
\ug - \up = \delta \log \expecqg \bra{\prod_{\iii} \pare{1 + \frac{\cg_i}{\delta_{-i}}}^{\lambda_i}}.
\end{equation}
Since $\prod_{\iii} \pare{1 + \cg_i / \delta_{-i}}^{\lambda_i}
\leq \sum_{\iii} \lambda_i \pare{1 + \cg_i / \delta_{-i}} = 1 +
\sum_{\iii} \lambda_i \cg_i / \delta_{-i}$ and from the fact that $\expecqg
\bra{\cg_i} = 0$ holds for all $\iii$, we indeed have $\ug \leq \up$ (which was anyway known from Remark \ref{rem:pareto}); furthermore, the equality $\ug = \up$ happens if and only if $\cg_i
= 0$ holds for all $\iii$, which happens if and only if $\cp_i =
0$ holds for all $\iii$---see Remark \ref{rem:no_trade_2}. In other words, the Nash risk-sharing equilibrium always implies a strict loss of
efficiency, except for the case where there is no trading within
Nash equilibrium (which is equivalent to the case where there is
no trading within Arrow-Debreu equilibrium as well).

\subsubsection{A priori information on $\zg$} \label{subsubsec: z in compact set}

From \eqref{eq:util_loss_nash_indiv} and \eqref{eq:util_loss_nash}, one obtains
\begin{equation} \label{eq:loss_decomp}
\ug_i - \up_i = \zg_i + \lambda_i \pare{ \ug - \up}, \quad \forall \iii.
\end{equation}
The above equality implies an economic interpretation for $\zg_i = \lambda_i \pare{ \up - \ug} + \pare{\ug_i - \up_i}$. Indeed, $ \lambda_i \pare{ \up - \ug}$ is the fraction, corresponding to agent $\iii$, of the aggregate loss of utility caused by forming a Nash, instead of Arrow-Debreu, equilibrium; on the other hand, $\ug_i - \up_i$ is the difference between the utility that agent $\iii$ acquires in Nash equilibrium from the Arrow-Debreu one.

Although the aggregate utility $\ug$ in Nash
equilibrium risk sharing can never be higher than the Arrow-Debreu
aggregate utility $\up$, \emph{it may happen that some agents
benefit from the game}, in the sense that their individual utility
after the negotiation game is higher when compared to the utility
gain of the Arrow-Debreu equilibrium. We will address such cases
in Section \ref{sec: extrema}.

Equation \eqref{eq:loss_decomp} is useful in obtaining tight bounds on $\zg = \pare{\zg_i}_{\iii}$. Using the facts that $\ug_i \geq 0$ for all $\iii$, $\ug \leq \up$, and the equality $\zg_i  = \lambda_i \pare{ \up - \ug} + \ug_i - \up_i$, it follows that
\begin{equation} \label{eq:z_bounds}
\zg_i \geq - \up_i, \quad \forall \iii.
\end{equation}
Combined with $\sum_{\iii} \zg_i = 0$, the previous \emph{a priori} bounds imply that $\zg$ has to live in a compact simplex on $\DI$. The bounds in \eqref{eq:z_bounds} are indeed sharp: in the no-trade setting of Remark \ref{rem:no_trade_2}, it follows that $\up_i = 0$ for all $\iii$, which implies that $\zg_i \geq 0$ should hold for all $\iii$; since $\zg \in \DI$, it follows that $\zg_i = 0$ should hold for all $\iii$. This also shows that the trivial Nash equilibrium obtained in Remark \ref{rem:no_trade_2} is unique.

\subsubsection{Individual marginal indifference valuation} \label{subsubsec:marginal_util_based}

In view of \eqref{eq:util_loss_nash} and the subsequent discussion, and recalling Remark \ref{rem:pareto}, it follows that the allocation in Nash equilibrium fails to be Pareto optimal (except in the trivial no-trade case). Another way to demonstrate the inefficiency  of Nash equilibrium is through the disagreement between the individual agent's marginal (utility) indifference valuation measures after the Nash risk-sharing transaction.

Recall that, given a position $G_i \in \lz$ with $\U_i(G_i) > - \infty$, the \emph{marginal indifference valuation measure}  $\qprob_i = \qprob_i(G_i)$ of agent $\iii$ has the property that the function $\Real \ni q \mapsto \U_i \pare{G_i + q (X - \expec_{\qprob_i} \bra{X}) }$ is maximised at $q = 0$ for all $X \in \Lb^\infty$; in other words, if prices are given by expectations under $\qprob_i$, agent $\iii$ has no incentive to take any position other than $G_i$. Using first-order conditions, it is straightforward to show that $\log \pare{\ud \qprob_i / \ud \prob_i} \sim - G_i / \delta_i$ holds.

In Arrow-Debreu equilibrium, the collection $\pare{\qp_i}_{\iii}$ with $\log \pare{\ud \qp_i / \ud \prob_i} \sim -  \cp_i  / \delta_i$ for $\iii$, which are the individual marginal indifference valuation measures associated with positions $\pare{\cp_i}_{\iii}$ after the Arrow-Debreu risk-sharing transaction, satisfies $\qp_i = \qp$ for all $\iii$: all agents' marginal indifference valuation measures agree. Now, denote the individual agent's marginal indifference valuation measures  after the Nash risk-sharing transaction by $\pare{\qg_i}_{\iii}$, for which $\log
\pare{\ud \qg_i / \ud \prob_i} \sim - \cg_i / \delta_i$
holds for all $\iii$. In view of $\log \pare{\ud \prob_i /\ud \qp} \sim \cp_i / \delta_i$, \eqref{eq:equil_C} and
\eqref{eq:equil_Q}, it follows that $\log \pare{\ud \qg_i / \ud \qg}
\sim \log \pare{1 + \cg_i / \delta_{-i}}$. Since $\expec_{\qg}
\bra{1 + \cg_i / \delta_{-i}} = 1$, it actually follows that
\begin{equation} \label{eq:indifference_val_nash}
\frac{\ud \qg_i}{\ud \qg} = 1 + \frac{\cg_i}{\delta_{-i}}, \quad
\forall \iii.
\end{equation}
Pareto optimality would require all $\pare{\qg_i}_{\iii}$ to
agree, which is possible only if $\cg_i = 0$, for all $\iii$,
i.e., exactly when no trade occurs.

All Nash securities $(\cg_i)_{\iii}$ have zero value under $\qg$. For each individual agent $\iii$, we can measure the marginal indifference value of $\cg_i$ via
\begin{equation}\label{eq:valuation under qgi}
\expec_{\qg_i} \bra{\cg_i} = \expec_{\qg} \bra{\pare{1 + \frac{\cg_i}{\delta_{-i}}} \cg_i} = \frac{1}{\delta_{-i}} \var_{\qg} \pare{\cg_i}, \quad \forall \iii.
\end{equation}
In particular, note that $\expec_{\qg_i} \bra{\cg_i} \geq 0$, with strict inequality if $\cg_i$ is non-zero, for all $\iii$. This observation implies that (except in trivial situations of no trading) \emph{all} agents would be better off if they would take a larger position in their individual securities; for all $a \in \Real_+$ the collection $(a \cg_i)_{\iii}$ of securities clears the market, and for some $a > 1$ this collection of securities would result in higher utility for each agent than using securities $(\cg_i)_{\iii}$. Of course, what prevents agents from doing so is that they would find themselves in (Nash) disequilibrium. The fact that agents will not agree on market-clearing collections $(a \cg_i)_{\iii}$ which for some $a > 1$ would be individually (and therefore, also collectively) preferable also indicates that trading volume within Nash equilibrium tends to be reduced.

The individual marginal indifference valuation measures
$\pare{\qg_i}_\iii$ allow for an interesting expression of the
Nash valuation measure $\qg$. To wit, recall from \S
\ref{subsubsec:never_true} the weights $\alpha_i =
\delta_{-i}/n\delta$ for all $\iii$; then, from
\eqref{eq:indifference_val_nash} and the market clearing condition
$\sum_{\iii} C_i = 0$, it follows that
\begin{equation}\label{eq:qg decomp}
\qg = \sum_{\iii} \alpha_i \qg_i.
\end{equation}
In words, the Nash valuation measure $\qg$ is a
convex combination of the individual agent's marginal
indifference valuation measures, assigning weight $\alpha_i$ to
agent $\iii$. Note also that more risk-averse agents carry more weight;
however, since $\max_{\iii} \alpha_i < 1/n$, $\qg$ is almost equal
to the equally-weighted average of $\pare{\qg_i}_{\iii}$ for large
numbers of agents.

Relation \eqref{eq:qg decomp} highlights the importance of risk tolerance levels regarding the gain or loss of utility for individual agents in Nash equilibrium. Consider for instance the
situation of two interacting agents, with one of them being
considerably more risk tolerant than the other. In this case,
$\qg$ will be very close to the risk-averse agent's marginal
utility-based valuation measure, who will agree with
the quoted prices. On the other hand, the possible discrepancy of $\qg$ from the risk tolerant agent's marginal utility-based valuation is beneficial to this agent, as it allows for the opportunity to purchase a positive-value security for zero price. A limiting instructive scenario along these lines is treated in Section \ref{sec: extrema}.

The marginal indifference valuation measures $\pare{\qg_i}_{\iii}$ of
\eqref{eq:indifference_val_nash} can be used to provide interesting formulas for the utility gain in Nash equilibrium and the utility difference between the Nash and Arrow-Debreu transactions. Note first that \eqref{eq:equil_Q} and \eqref{eq:util_loss_nash} give
\[
\log \pare{\frac{\ud \qp}{\ud \qg}} =  \sum_{\jii} \lambda_j \log \pare{1 + \frac{\cg_j}{\delta_{-j}} } + \frac{\up - \ug}{\delta},
\]
which combined with \eqref{eq:loss_decomp} and the fact that $\cp_i = \delta_i \log \pare{\ud \prob_i / \ud \qp} + \up_i$ implies
that
\begin{equation} \label{eq:marg_util_nash_help}
\cg_i + \delta_i \log \pare{1 + \frac{\cg_i}{\delta_{-i}}}  = \zg_i + \cp_i + \delta_i \sum_{\jii} \lambda_j \log \pare{1 + \frac{\cg_j}{\delta_{-j}} } = \ug_i + \delta_i
\log \pare{\frac{\ud \prob_i}{\ud \qg}}.
\end{equation}
Using further \eqref{eq:indifference_val_nash} and taking expectation with respect to $\qg$ in \eqref{eq:marg_util_nash_help}, we obtain
\[
\ug_i = \delta_i \ent(\qg \such \prob_i) - \delta_i \ent(\qg \such \qg_i), \quad \iii.
\]
The last equality has to be compared with \eqref{eq:opt_util_price_agent_AD}. As in Arrow-Debreu equilibrium, agents in Nash equilibrium  benefit from the distance of the resulting valuation measure from their subjective views; however, unlike the Pareto-optimal efficiency of the Arrow-Debreu transaction, agents in the Nash transaction suffer loss from the distance of the valuation measure from their respective marginal indifference valuation measures.

From \eqref{eq:marg_util_nash_help} and  \eqref{eq:indifference_val_nash} it follows that $\cg_i = - \delta_i \log \pare{\ud \qg_i / \ud \prob_i} + \U_i( C_i)$; combining this with $\cp_i = \delta_i \log \pare{\ud \prob_i / \ud \qp} + \up_i$, we obtain $\cg_i +
\delta_i \log \pare{\ud \qg_i / \ud \qp} = \cp_i + (\ug_i -
\up_i)$, for all $\iii$. Taking expectations with respect to
$\qp$, it follows that
\begin{equation} \label{eq:nash_gain_ent_indiv}
\ug_i - \up_i = \expec_{\qp} \bra{\cg_i} - \delta_i \ent (\qp \such
\qg_i ), \quad \forall \iii.
\end{equation}
The difference of individual agents' utilities in the two
equilibria comes from two distinct sources. The first stems from
the discrepancy (measured via the relative entropy) of the
Arrow-Debreu valuation from the individual marginal
indifference valuation of agent $\iii$ in Nash equilibrium. When
the agents' marginal indifference valuation measure in Nash
equilibrium is close to the Arrow-Debreu measure, his loss of
utility caused by the Nash game is lower. In a sense, this is the
part of aggregate loss of utility that is ``paid'' by agent $\iii$
(see also \eqref{eq:nash_gain_ent_aggr} below). The other term on
the right-hand-side of \eqref{eq:nash_gain_ent_indiv} regards the
price under the Arrow-Debreu valuation measure $\qp$ of the actual
security that agent $\iii$ buys at Nash equilibrium. Recalling that
Nash equilibrium prices of the Nash securities $(\cg_i)_{\iii}$
are zero, positivity of $\expec_{\qp} \bra{\cg_i}$ implies that
the security $\cg_i$ is \textit{undervalued} in Nash equilibrium
transaction. Again, note that if $\qg_i$ is close to $\qp$, the
valuation $\expec_{\qp} \bra{\cg_i}$ tends to be positive, since
$\expec_{\qg_i} \bra{\cg_i}$ is always nonnegative (see
\eqref{eq:valuation under qgi}). To recapitulate the previous
discussion: \textit{agents whose marginal indifference valuation measure
is close to the Arrow-Debreu one tend to benefit from the Nash
game}. As we will see in Section \ref{sec: extrema}, this
happens, for example, when agent $\iii$ is sufficiently risk
tolerant.

Due to the market-clearing condition $\sum_{\iii}\cg_i=0$, the
aggregate loss takes into account only the aggregate discrepancy
of individual marginal measures from the Arrow-Debreu optimal one:
under-valuation of certain securities balances off by
over-valuation of others. Indeed, adding up
\eqref{eq:nash_gain_ent_indiv} over all $\iii$, gives
\begin{equation} \label{eq:nash_gain_ent_aggr}
\up - \ug = \sum_{\iii} \delta_i \ent (\qp \such \qg_i ),
\end{equation}
which measures Nash inefficiency as aggregate discrepancy from
optimal valuation of the individual agents' marginal indifference
valuation in Nash equilibrium. Equation \ref{eq:nash_gain_ent_aggr} is the counterpart of \eqref{eq:opt_util_aggr_AD}, where the inefficiency of complete absence of trading as compared to Arrow-Debreu risk-sharing is considered.

\subsection{Existence and uniqueness of Nash equilibrium via finite-dimensional root finding} \label{subsec:root_finding}

Theorem \ref{thm: nash} is used as a guide in order to search for equilibrium, parametrising candidates for optimal securities using the $n$-dimensional space $\DI$ introduced in \eqref{eq:simplex}. Proposition \ref{prop: C_system_sol} that follows, and whose proof is the content of \S \ref{subsec: proof_C_system_sol}, enables to reduce the search of Nash equilibrium, an inherently infinite-dimensional problem in our setting, to a finite-dimensional one. The latter problem gives the necessary tools for numerical approximations of Nash equilibria (see also Example \ref{ex:three_agents_equil} below).

\begin{prop} \label{prop: C_system_sol}
For all $z \in \DI$ there exists unique $(C_i(z))_{\iii} \in \lzi$ with $C_i(z) > - \delta_{-i}$ and
\begin{equation} \label{eq: C_system}
C_i(z) +  \delta_i \log \pare{1 + \frac{C_i(z)}{\delta_{-i}} } =
z_i + \cp_i +  \delta_i \sum_{\jii} \lambda_j \log \pare{1 + \frac{C_j(z)}{\delta_{-j}} }, \quad \forall \iii.
\end{equation}
(Note that,  necessarily, $\sum_{\iii} C_i(z) = 0$ for all $z \in \DI$.) Furthermore, it holds that
\begin{equation} \label{eq:Qg_z_integr}
\expecqp \bra{ \prod_{\jii}  \pare{1 + \frac{C_j(z)}{\delta_{-j}} }^{- \lambda_j} } < \infty.
\end{equation}
\end{prop}

In the notation of Proposition \ref{prop: C_system_sol}, for each $z \in \DI$, define the probability $\qprob(z)$ via
\begin{equation} \label{eq: Q_z}
\log \pare{ \frac{\ud \qprob(z)}{\ud \qp}} \sim - \sum_{\jii} \lambda_j \log \pare{1 + \frac{C_j(z)}{\delta_{-j}} }.
\end{equation}
The uniform bounds $- \delta_{-i} < C_i(z) < (n-1) \delta + \delta_i$ follow exactly as in \S \ref{subsubsec:endog_bounds}, and imply that $\exp( C_i(z) / \delta_i) \in \Lb^1(\qz)$ holds for all $\iii$ and $z \in \DI$. In particular, $\pare{C_i(z)}_{\iii} \in \C_{\qprob(z)}$ holds for all $z \in \DI$.
In view of Theorem \ref{thm: nash}, Nash equilibria amount to finding $z \in \DI$ such that $\expecqz \bra{C_i(z)} = 0$ holds for all $\iii$. We can in fact define a   function $\dis : \DI \mapsto \Real_+$ that gives a ``distance from equilibrium'' via the formula
\begin{equation}\label{eq: dist}
\dis(z) = - \sumi \delta_{-i} \log \pare{1 + \frac{\expecqz
\bra{C_i(z)}}{\delta_{-i}}}, \quad \forall z \in \DI.
\end{equation}
Since $C_i(z) > - \delta_{-i}$ holds for all $z \in \DI$, $\ell$
is well defined. Furthermore, the inequality $\log(x) \leq x-1$,
valid for all $x \in (0, \infty)$, gives
\[
\dis(z) \geq - \sumi \delta_{-i} \pare{\frac{\expecqz \bra{C_i(z)}}{\delta_{-i}}} = - \expecqz \bra{\sumi C_i(z)} = 0, \quad \forall z  \in \DI,
\]
in view of the fact that $\sum_{\iii} C_i(z) = 0$ for all $z
\in \DI$, which shows that $\dis$ is indeed $\Real_+$-valued. Furthermore, since $\log(x) < x-1$ holds for all $x
\in (0, \infty) \setminus \set{1}$, for any $z \in \DI$ it follows
that $\dis(z) = 0$ is equivalent to $\expecqz \bra{C_i(z)} = 0$
for all $\iii$.

The following result summarises the above discussion.

\begin{prop} \label{prop:nash_root}
With the previous notation, the following are true:
\begin{itemize}
    \item Assume that $\pare{\qg, (\cg_i)_{\iii}}$ is a Nash equilibrium, and let $\zg \equiv \pare{\zg_i}_{\iii} \in \DI$ be as in \eqref{eq:equil_C}. Then, $\dis \pare{\zg} = 0$.
    \item Assume the existence of $\zg \in \DI$ such that $\dis \pare{\zg} = 0$. Then, $\pare{\qprob(\zg), (C_i (\zg))_{\iii}}$ as defined in \eqref{eq: C_system} and \eqref{eq: Q_z} is a Nash equilibrium.
\end{itemize}
\end{prop}

Proposition \ref{prop:nash_root} provides a one-to-one correspondence between Nash equilibria and roots of $\dis$. Recalling the discussion in \S\ref{subsubsec: z in compact set}, any root of $\ell$ belongs to the compact subset of $\DI$ consisting of $(z_i)_{\iii} \in \DI$ with $z_i \geq - \up_i$ for all $\iii$. This fact allows for numerical approximations of Nash equilibria via, for example, Monte-Carlo simulation.

\smallskip

Its practical usefulness notwithstanding, Proposition \ref{prop:nash_root} does not answer the question of actual \emph{existence} of Nash equilibria and, in case of existence, the \emph{uniqueness}. Such issues are settled in Theorem \ref{thm:nash_exist_uniq} that follows, the proof of which is the subject of \S \ref{subsec:proof_of_nash_exist_uniq}.

\begin{thm} \label{thm:nash_exist_uniq}
A Nash risk-sharing equilibrium always exists. When, additionally, $I = \set{0,1}$, Nash risk-sharing equilibrium is necessarily unique.
\end{thm}

The question of uniqueness for three or more agents remains open, and is significantly more challenging from a mathematical perspective. In all cases of numerical simulation that were carried out, we observed (existence and) uniqueness of a Nash equilibrium. The next example is representative.

\begin{exa} \label{ex:three_agents_equil}
Consider a three-agent game with  $\delta_0 = \delta_1 = \delta_2 = 1$. We assume that $\log \pare{\ud \prob_i / \ud \prob} \sim X_i$ holds for $i \in \set{0,1,2}$, where $(X_0, X_1, X_2)$ under baseline probability $\prob$ has a mean-zero normal distribution with $\sigma(X_0) = 0.4$, $\sigma (X_1) = 2.7$, $\sigma(X_2) = 1.1$, $\rho(X_0, X_1)=-0.9$, $\rho(X_0,X_2)=0.7$ and $\rho(X_1,X_2)=-0.3$. In Figure \ref{fig: distance}, we plot the function $\ell$ for different values
of $(z_1, z_2)$, only in the bounded region specified by the inequalities $z_1 \geq - \up_1$, $z_2 \geq - \up_2$, and $z_1 + z_2 \leq \up_0$, where recall that $z_1 + z_2 = - z_0$. As can be seen, there is a unique root of $\ell$ approximately at the vector $\zg=(\zg_0,\zg_1,\zg_2) = (0.14,-0.7,0.56)$.
\begin{figure}[!ht]
\includegraphics[trim = 35mm 0mm 0mm 0mm, clip, scale=0.5]{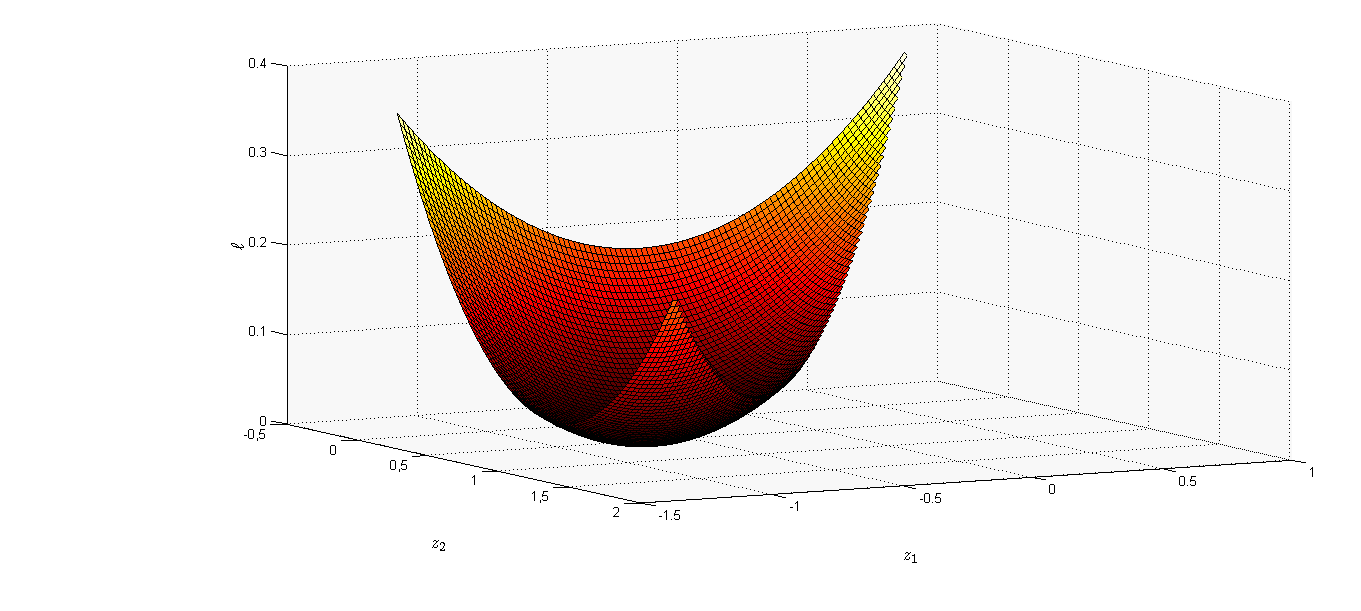}
\caption{{\footnotesize The function $\ell$ for different values
of $z$, corresponding to Example \ref{ex:three_agents_equil}.}}
\label{fig: distance}
\end{figure}
\end{exa}

\section{Extreme Risk Tolerance}\label{sec: extrema}

As discussed in \S \ref{subsubsec:marginal_util_based}, risk-tolerance coefficients are crucial factors in the gain or the loss caused by the game in each agent's utility. In this section, we investigate this issue closer by studying and comparing the Arrow-Debreu and Nash risk-sharing equilibria when agents' risk preferences approach risk neutrality, in the sense that risk tolerance approaches infinity. In order to focus on the economic interpretation of the results, we consider the simplified (but representative) case of two agents.

The analysis that follows examines two cases: firstly, when only one
agent becomes extremely risk tolerant and, secondly, when both agents'
risk tolerance coefficients uniformly approach infinity. Besides the
interest of this analysis in its own right, it also allows us to
substantiate the claim that highly
risk tolerant agents are the ones actually benefit from the risk-sharing game.

\subsection{One extremely risk tolerant agent} \label{subsec: limiting_one}

We start with the two-agent case $I = \set{0,1}$, wherein the risk aversion of only one agent approaches zero. We keep the
risk tolerance $\delta_1$ and subjective probability $\prob_1$ of agent 1
fixed. On the other hand, for agent $0$ we consider a sequence of risk tolerance coefficients $\pare{\delta_0^m}_{\mina}$ with the property that $\limm
\delta^m_0 = \infty$ and a fixed subjective probability $\prob_0$. In this set-up, Theorem \ref{thm: AD} and Theorem \ref{thm:nash_exist_uniq} state that for each $\mina$ there exist a unique
Arrow-Debreu equilibrium $\pare{\qpm, (\cpm_i)_{\iii}}$ and a
unique Nash equilibrium $\pare{\qgm, (\cgm_i)_{\iii}}$. We use agent 0 as the baseline and focus on the securities $\cpm_0$ and $\cgm_0$, since $\cpm_1 = -\cpm_0$ and $\cgm_1 = -\cgm_0$.

We first examine the limiting behaviour of the valuation rule and the securities in the Arrow-Debreu equilibrium transaction. For each $\mina$, from \eqref{eq: qp} we obtain that $\qpm \in \PP$ is such that
$\log \pare{\ud \qpk / \ud \prob_0} \sim \lambda_1^m \log \pare{\ud \prob_1 / \ud \prob_0}$. More precisely, we have
\begin{equation} \label{eq:qpm0}
\frac{\ud \qpk}{\ud \prob_0} = \expec_{\prob_0} \bra{\pare{\frac{\ud \prob_1}{\ud \prob_0}}^{ \lambda_1^m}}^{-1} \pare{\frac{\ud \prob_1}{\ud \prob_0}}^{ \lambda_1^m}.
\end{equation}
Given that $\limm \lambda_1^m = 0$, $\plimm \pare{\ud \qpk / \ud \prob_0} = 1$ readily follows from the dominated convergence theorem---in fact, with $\normtv{\cdot}$ denoting total variation norm, Scheffe's lemma implies that $\limm \normtv{\qpk - \prob_0} = 0$. Since, $\cpm_0 = - \cpm_1 = \delta_1 \log \pare{\ud \qpm / \ud \prob_1} - \delta_1 \ent(\qpk \such \prob_1)$ holds for all $\mina$ and $(\qpk)_{\mina}$ converges to  $\prob_0$, one expects that $\plimm \cpk_0 = \delta_1 \log(\ud \prob_0 / \ud \prob_1) - \delta_1 \ent(\prob_0 \such \prob_1)$. Clearly, for the previous limit to be valid the following (technical) assumption is \emph{necessary}.

\begin{ass} \label{ass:asymptotic}
$\ent(\prob_0 \such \prob_1) < \infty$.
\end{ass}

In \S \ref{subsec: proof_of_prop_limit_AD}, it is shown that the latter assumption is also sufficient for the validity of Proposition \ref{prop: limit_AD} below, giving the limiting valuation and security in Arrow-Debreu equilibrium, as well as the limiting gain of both agents.

\begin{prop}\label{prop: limit_AD}
Under the force of Assumption \ref{ass:asymptotic}, it holds that $\cpi_0 \dfn \plimm \cpk_0 = \delta_1 \log(\ud \prob_0 / \ud \prob_1) - \delta_1 \ent(\prob_0 \such \prob_1)$, $\limm \upm_0 = 0$ and $\limm \upm_1 = \delta_1 \ent(\prob_0 \such \prob_1)$.
\end{prop}

It is indeed expected that the utility gain of a nearly risk neutral agent is almost zero. To see this, compare the limiting valuation measure, which is $\prob_0$, with the limiting utility of agent 0, which is linear expectation with respect to $\prob_0$. On the other hand, the only case where there is no limiting utility gain for agent $1$ is when the two agents' subjective beliefs coincide.

\smallskip

We now turn to Nash risk-sharing equilibrium. From \eqref{eq:
C_two_agent}, we obtain
\[
\cgm_0 + \delta_1 \lambda^m_0 \log \pare{ \frac{1 + \cgm_0 /
\delta_1}{1 - \cgm_0 / \delta^m_0}} = \zgm_0 + \cpm_0, \quad
\forall \mina.
\]
Accepting that the sequence $\pare{\zgm_0}_{\mina}$ converges in
$\Real$ and $\pare{\cgm_0}_{\mina}$ converges in $\lz$ (these
conjectures actually have to be proved as part of Theorem
\ref{thm:limiting_nash_one} below), and given that $\limm
\delta^m_0 = \infty$, $\limm \lambda_0^m = 1$, and $\plimm \cpm_0
= \cpi_0$, the limiting security $\cgi_0 \dfn \plimm \cgm_0$
should satisfy $\cgi_0 + \delta_1 \log \pare{ 1 + \cgm_0 /
\delta_1} = \zgi_0 + \cpi_0$, where $\zgi_0 \dfn \limm \zgm_0$. This
heuristic discussion gives a method to compute the limit. For any
$z \in \Real$, define the random variable $C^{\infty}_0(z)$
satisfying the equation
\begin{equation}\label{eq: limiting security}
C^{\infty}_0(z) + \delta_1 \log \pare{1 +
\frac{C^{\infty}_0(z)}{\delta_1}} = z + \cpi_0.
\end{equation}
Since the function $(- 1, \infty) \ni x \mapsto x + \log \pare{1 +
 x}$ is strictly increasing and continuous and maps $(-1, \infty)$ to $(- \infty, \infty)$, it follows that $C^{\infty}_0(z)$ is a well defined $(- \delta_1, \infty)$-valued random variable for all $z \in \Real$. Then, we should have $\cgi_0 = C^\infty_0 (\zgi_0)$. Although $\zgi_0$ is given as the limit of $\pare{\zgm_0}_{\mina}$, we may actually identify \emph{a priori} what its value will be. To make headway, note that
 from \eqref{eq:equil_Q}
\[
\log \pare{\ud \qgm / \ud \qpm} \sim - \lambda^m_0 \log(1 + \cgm_0
/ \delta_1) - \lambda^m_1 \log(1 - \cgm_0 / \delta^m_0), \quad
\forall \mina,
\]
and the fact that $\limm \normtv{\qpk - \prob_0} = 0$, the limiting Nash valuation probability $\qgi$ should be such that
$\log \pare{\ud \qgi / \ud \prob_0} \sim -\log(1 + \cgi_0 /
\delta_1)$; since $\expec_{\qgi} \bra{\cgi_0} = 0$ is expected to
hold at the limit, we obtain actually that $\expec_{\prob_0} \big[
\pare{1 + \cgi_0 / \delta_1}^{-1} \big] = 1$ would have to be
satisfied. The next result, the proof of which is given in \S
\ref{subsec: proof_of_lem:zgi}, ensures that a unique such
candidate $\zgi_0 \in \Real$ exists.

\begin{lem}\label{lem: zgi}
In the notation of \eqref{eq: limiting security},
there exists a unique $\zgi_0 \in \Real$ satisfying the equality $\expec_{\prob_0} \big[  \pare{1 + C^{\infty}_0(\zgi_0) / \delta_1}^{-1} \big] = 1$.
\end{lem}

Before we state our main result on the limiting behaviour of Nash
equilibrium, we make a final observation. Recall from
\eqref{eq:reported_nash_agent} that $\log \pare{ \ud \rgm_1 / \ud \prob_1} \sim - \log \pare{1 -  \cgm_0 / \delta^m_0}$ holds for all $\mina$. Since
$\limm \delta_0^m = \infty$ and, as it turns out,
$\pare{\cgm_0}_{\mina}$ is convergent, the revealed subjective probability
$\rgm_1$ of agent $1$ when $m$ is large is very close to the
actual $\prob_1$. (There is an alternative way to obtain the same intuition. From \eqref{eq:lower_bound_dens}, note that $\ud \rgm_1 / \ud \prob_1 \geq \lambda_0^m$ holds for all $\mina$. Since $\limm \lambda_0^m = 1$ and $(\ud \rgm_1 / \ud \prob_1)_{\mina}$ has constant unit expectation under $\prob_1$, $\pare{\rgm_1}_{\mina}$ has to converge to $\prob_1$.) This suggests the same asymptotic behaviour in the case discussed in \S \ref{subsec:market_power}, where only agent 0 acts strategically as indicated by the best probability response, while agent 1 reports true subjective beliefs $\prob_1$.
Indeed, the following result, whose proof is given in
\S \ref{subsec: proof_of_limiting_nash_one}, implies that the
limiting security structure is the same, regardless of whether the
risk-averse agent 1 enters in the game or simply reports true
subjective beliefs (in which case, only the approximately risk neutral agent behaves strategically).

\begin{thm} \label{thm:limiting_nash_one}
With the previous notation (in particular, of Lemma \ref{lem:
zgi}), it holds that
\[
\cgi_0 \dfn \plimm \cgm_0 = C^{\infty}_0(\zgi_0) = \plimm \cbom.
\]
\end{thm}

The equality of the limits of $\pare{\cgm_0}_{\mina}$ and
$\pare{\cbom}_{\mina}$ implies that the strategic behaviour of a
risk neutral agent \textit{dominates} the risk sharing
transaction. Intuitively, agents
with high risk tolerance are willing to undertake more risk at the
sharing transaction in return of a higher cash compensation. Thus, at
the limit, the risk neutral agent satisfies the
reported hedging needs of other agents, but achieves better prices by applying the best response strategy. On the
other hand, for the risk averse agent the risk reduction
is more important than a higher price to be paid. As a result,
at the equilibrium the risk averse agent prefers to submit true beliefs,
even though this results in a higher price to be paid to the risk neutral
agent. The situation is totally different in an Arrow-Debreu
equilibrium transaction, where agents act basically as price
takers and the securities and prices are determined by the
efficiency of the transaction.

\smallskip

We argued in Subsection \ref{subsec: in Nash} that in any
risk-transfer situation the Nash equilibrium incurs some loss of
efficiency. Although the aggregate utility is reduced in Nash
equilibrium when compared with the Arrow-Debreu one, certain
agents may obtain higher utility gain in risk-sharing games. In
particular, Proposition \ref{prop: limiting_gain_game} below (the
proof of which is given in \S \ref{subsec:
proof_of_prop_limiting_gain_game}) demonstrates that the agent
with sufficiently high risk tolerance enjoys higher utility at
Nash equilibrium transaction than the utility at the Arrow-Debreu
equilibrium sharing.

\begin{prop} \label{prop: limiting_gain_game}
Define $\qgi \in \PP$ such that $\ud \qgi / \ud \prob_0 = \pare{1 +
\cgi_0 / \delta_1}^{-1}$. Then:
\begin{align*}
\limm \pare{\ugm_0 - \upm_0} &= \ \ \frac{1}{\delta_1} \var_{\qgi}
    \pare{\cgi_0}, \\
\limm \pare{\ugm_1 - \upm_1} &=  - \frac{1}{\delta_1} \var_{\qgi}
\pare{\cgi_0} - \delta_1 \ent \big( \prob_0 \such \qgi \big).
\end{align*}
\end{prop}

The limiting loss for the risk averse agent comes from two sides. The first is $(1 / \delta_1) \var_{\qgi} \pare{\cgi_0}$, which is the limiting gain
of agent $0$. The remaining quantity $\delta_1 \ent \big( \qpi |
\qgi \big)$ is in fact the loss from the applied strategic behaviour as opposed
to sharing in a Pareto optimal way. Both terms are strictly
positive as long as $\cgi_0$ is not identically equal to zero.

The message of Proposition \ref{prop: limiting_gain_game} is
clear. The introduction of strategic behaviour allows agents with
high risk tolerance to achieve better prices that the more risk
averse agents are willing to pay in order to achieve risk reduction. In contrast to the Arrow-Debreu equilibrium where prices are given by the optimal sharing measure, agents with sufficiently high risk tolerance are willing to accept
more risk in the Nash game, since their strategy drives the
market to better cash compensation for them. In fact, a more risk
averse agent not only tends to undertake all the efficiency loss caused by the game, but also fuels the utility gain of the (sufficiently) risk tolerant counterparty.

Recalling the discussion and notation of \S
\ref{subsubsec:marginal_util_based}, we may offer some more detailed
comments. From \eqref{eq:indifference_val_nash} and Proposition
\ref{prop: limiting_gain_game}, it follows that the marginal
valuation measure of agent 0 approaches the limiting optimal
valuation measure $\qpi$. This implies that, for large enough
$\mina$, the security that agent 0 gets in Nash
equilibrium is undervalued---indeed, note that $\expec_{\qpi} \bra{\cgi_0}
= \expec_{\qgi} \bra{\cgi_0 (1 + \cgi_0 / \delta_1)} = (1 /
\delta_1) \var_{\qgi}
\pare{\cgi_0}$. According to \eqref{eq:nash_gain_ent_indiv} and
the discussion that follows, we readily get that utility of agent 0 is increased. For the risk averse agent, the situation is different. From \eqref{eq:qg decomp}, it follows that $\qgm_1$ will be close to $\qgm$ for
large $\mina$, which in turn will be close to $\qgi$. Hence, for
large enough $\mina$, the security received by agent 1 in
Nash equilibrium is overvalued; on top of this, agent 1 also carries all the risk-sharing inefficiency of Nash equilibrium.

\subsection{Both agents being extremely risk tolerant} 

We have seen above that the strategic behaviour of a high risk tolerant agent dominates the Nash game and drives the market to his preferable transaction, regardless the actions of the other agent. Here, we shall examine what happens to the equilibria when both agents approach risk neutrality at the same speed. More precisely, we fix $\lambda_0 \in (0,1)$ and $\lambda_1 \in (0,1)$ with $\lambda_0 + \lambda_1 = 1$ and consider a non-decreasing sequence $\pare{\delta^m}_{\mina}$ such that $\limm \delta^m = \infty$. Define $\delta^m_i \dfn \lambda_i \delta^m$ for all $\mina$ and $i \in \set{0,1}$. In contrast to the set-up of \S \ref{subsec: limiting_one}, here the subjective beliefs of the agents will have to depend on $\mina$. To obtain intuition on why and how the subjective probabilities have to behave, note that according to Theorem \ref{thm: AD}, for all $\mina$ the security $\cpm_0$ is given as a multiple of $\delta_0^m$ of a random variable whose dependence on risk tolerance comes only through $\lambda_0$ and $\lambda_1$. Since the latter weights are fixed for each $\mina$, in order to guarantee that the securities in Arrow-Debreu equilibrium have a well-behaved limit, we make the following assumption.

\begin{ass} \label{ass:asymptotic_both}
For $i \in \set{0,1}$, there exists $\xi_i \in \li$ such that $\expecp \bra{\xi_i} = 0$ and
\[
\log \pare{\frac{\ud \prob_i^m}{\ud \prob}} \sim \frac{\xi_i}{\delta_i^m}, \quad  i \in \set{0,1}, \quad \mina.
\]
\end{ass}
Note that condition $\expecp \bra{\xi_i} = 0$ for $i \in \set{0,1}$ appearing in Assumption \ref{ass:asymptotic_both} is just a normalisation, and does not constitute any loss of generality.

\begin{thm} \label{thm:limiting_contract_both}
In the above set-up, and under Assumption \ref{ass:asymptotic_both},
the sequences $\pare{\cpm_0}_{\mina}$ and
$\left(\cgm_0\right)_{\mina}$ converge in $\Lb^0$ to limiting
securities $\cpi_0$ and $\cgi_0$, where
\[
    \cpi_0 = \lambda_1 \xi_0 - \lambda_0 \xi_1, \quad \cgi_0 = \frac{\lambda_1}{2} \xi_0 - \frac{\lambda_0}{2} \xi_1 = \frac{ \cpi_0}{2}.
\]
\end{thm}

The proof of Theorem \ref{thm:limiting_contract_both} is given in \S \ref{subsec: proof of limiting_contract_both}. Interestingly, the risk neutrality of both agents drives Nash equilibrium to the half of the Arrow-Debreu securities, which is an evidence of the market inefficiency caused by the strategic behaviour of risk-neutral agents. The result of Theorem \ref{thm:limiting_contract_both} is another manifestation of the claim (initially made in \S \ref{subsubsec:marginal_util_based}) that trading volume in Nash equilibrium tends to be lower than Pareto-optimal allocations.



\appendix

\section{Proofs}\label{sec: appe}

\subsection{Proof of Theorem \ref{thm: AD}}
\label{subsec: proof_of_AD}
Suppose that $\pare{\qp, \pare{\cp_i}_{\iii}}$ is an Arrow-Debreu equilibrium. We shall show the necessity of \eqref{eq: qp} and \eqref{eq: opt_contr}. For all $\iii$, note that $\U_i (\cp_i) \geq \U_i (0) = 0$, which implies that $\exp(- \cp_i / \delta_i) \in \Lb^1(\prob_i)$. Fix $X \in \li$ with $\expecqp \bra{X} = 0$ and $\iii$. Since the function $\Real \ni \epsilon \mapsto \U_i(\cp_i + \epsilon X) \in \Real$ has a maximum at $\epsilon = 0$, first order conditions and the dominated convergence theorem, using the fact that $\exp(- \cp_i / \delta_i) \in \Lb^1(\prob_i)$, imply that $\expec_{\prob_i} \bra{\exp \pare{-\cp_i / \delta_i} X } = 0$. The latter equality holds for all $X \in \li$ with $\expec_{\qp} \bra{X} = 0$ and all $\iii$; therefore, $\cp_i \sim \delta_i \log \pare{\ud \prob_i / \ud \qp}$, for all $\iii$. Since $\expec_{\qp} \bra{\cp_i} = 0$, \eqref{eq: opt_contr} follows. Furthermore, the fact that $\sumi \cp_i = 0$ gives $\sumi \delta_i \log \pare{\ud \prob_i / \ud \qp} \sim 0$, from which \eqref{eq: qp} follows.

Assume now that $\pare{\qp, \pare{\cp_i}_{\iii}}$ is given by \eqref{eq: qp} and \eqref{eq: opt_contr}. The fact that $\expec_{\qp} \bra{\cp_i} = 0$ holds for all $\iii$ is immediate by definition. Furthermore, \eqref{eq: qp} and \eqref{eq: opt_contr} give $\sumi \cp_i \sim \sumi \delta_i \log \pare{\ud \prob_i / \ud \qp} \sim \delta \sumi \lambda_i \log \pare{\ud \prob_i / \ud \qp} \sim 0$; together with $\expec_{\qp} \bra{\cp_i} = 0$ for all $\iii$, this implies that $\sumi \cp_i = 0$. The fact that $\cp_i$ is optimal for agent $\iii$ under the valuation measure $\qp$ is argued in Remark \ref{rem:pareto}. We have shown that $\pare{\qp, \pare{\cp_i}_{\iii}}$ given by \eqref{eq: qp} and \eqref{eq: opt_contr} is an Arrow-Debreu equilibrium. The necessity of \eqref{eq: qp} and \eqref{eq: opt_contr} for Arrow-Debreu equilibrium proved in the previous paragraph establishes its uniqueness.

\subsection{Proof of Proposition \ref{prop: best_response_first_ord}} \label{subsec: proof_of_best_resp_prop}

In order to ease the reading, in the course of the proof of Proposition \ref{prop: best_response_first_ord} we shall denote $\qprob^{(\rprob_{-i},\rb_{i})}$ by $\qb$.

\subsubsection{First-order conditions} \label{subsubsec:proof_of_best_resp_prop_foc}

We shall prove here the necessity of the stated conditions for best response. Fix $\iii$ and $\rb_i \in \PP$ such that $\V_i (\rb_i; \rprob_{-i}) = \sup_{\rprob_i \in \PP} \V_i (\rprob_i;\rprob_{-i})$ holds. For $\rprob^0_i \in \PP$ defined via $\log \ud \rprob^0_i \sim (1 / \lambda_{-i}) \sum_{\jii \setminus \set{i}} \lambda_j \log \ud \rprob_j$, the resulting contract for agent $\iii$ would be zero; therefore, $\U_i (\cb_i) = \V_i (\rb_i; \rprob_{-i}) \geq \V_i (\rprob^0_i; \rprob_{-i}) = 0$. In particular we have that $\exp \pare{- \cb_i / \delta_i} \in \Lb^1(\prob_i)$, a fact that will be useful in several places applying the dominated convergence theorem in the sequel.

Fix $X \in \Lb^\infty$. For $\epsilon \in \Real$, define $\rprob_i (\epsilon) \in \PP$ via $\log \pare{\ud \rprob_i (\epsilon) / \ud \rb_i} \sim - \epsilon X / (\lambda_{-i} \delta_i)$. With $\qprob (\epsilon) \equiv \qprob^{(\rprob_{-i},\rprob_i (\epsilon))}$, it follows that $\log \pare{\ud \qprob (\epsilon) / \ud \qb} \sim - \epsilon X / \delta_{-i}$. In accordance with $\cb_i$, define $C_i(\epsilon) = \delta_i \log \pare{\ud \rprob_i (\epsilon) / \ud \qprob (\epsilon)} + \delta_i \ent \pare{\qprob (\epsilon) \such \rprob_i (\epsilon)}$ and then, $C_i(0) = \cb_i$. Noting that
\[
\delta_i \log \pare{\frac{\ud  \rprob_i (\epsilon)}{\ud \qprob (\epsilon)}} = \delta_i \log \pare{\frac{\ud  \rprob_i (\epsilon)}{\ud \rb_i}} + \delta_i \log \pare{\frac{\ud  \rb_i}{\ud \qb}} + \delta_i \log \pare{\frac{\ud  \qb}{\ud \qprob (\epsilon)}} \sim \cb_i - \epsilon  X,
\]
it follows that $C_i(\epsilon) = \cb_i - \epsilon X - \expec_{\qprob (\epsilon)} \bra{\cb_i - \epsilon X}$, where the constant in the equivalence above was cancelled out by definition of $C_i(\epsilon)$. The dominated convergence theorem and simple differentiation, using also the fact that $\expec_{\qb} \bra{\cb_i} = 0$, imply that
\[
C_i'(0) = \frac{\partial C_i(\epsilon)}{ \partial \epsilon} \Big|_{\epsilon = 0} = - X + \expec_{\qb} \bra{\pare{1 + \frac{\cb_i}{\delta_{-i}} }X}.
\]
Since $\V_i(\rprob(\epsilon); \rprob_{-i}) = \U_i (C_i(\epsilon))$ holds for all $\epsilon \in \Real$, another application of the dominated convergence theorem gives
\[
\frac{\partial \V_i(\rprob(\epsilon); \rprob_{-i})}{ \partial \epsilon} \Big|_{\epsilon = 0} = \frac{\expec_{\prob_i} \bra{\exp(- \cb_i / \delta_i)  C_i'(0)}}{\expec_{\prob_i} \bra{\exp(- \cb_i / \delta_i)}}.
\]
Since $\Real \ni \epsilon \mapsto \V_i(\rprob(\epsilon) ; \rprob_{-i})$ is maximised at $\epsilon = 0$, first-order conditions give that
\begin{equation} \label{eq:f-o-1}
0 = \frac{\expec_{\prob_i} \bra{\exp(- \cb_i / \delta_i)  C_i'(0)}}{\expec_{\prob_i} \bra{\exp(- \cb_i / \delta_i)}} = - \frac{\expec_{\prob_i} \bra{\exp(- \cb_i / \delta_i)  X} }{\expec_{\prob_i} \bra{\exp(- \cb_i / \delta_i)}} + \expec_{\qb} \bra{\pare{1 + \frac{\cb_i}{\delta_{-i}} }X}.
\end{equation}
Noting that $\sum_{\jii \setminus \set{i}} \lambda_j \log \pare{ \ud \rprob_j / \ud \prob_i} \sim \log \pare{ \ud \qb / \ud \prob_i} - \lambda_i \log \pare{ \ud \rb_i / \ud \prob_i}$ implies
\[
\lambda_{-i} \log \pare{\frac{\ud \qb}{\ud \prob_i}} - \sum_{\jii \setminus \set{i}} \lambda_j \log \pare{\frac{\ud \rprob_j}{\ud \prob_i}} \sim \lambda_i \log \pare{\frac{\ud \rb_i}{\ud \qb}},
\]
it follows from $\cb_i \sim \delta_i \log \pare{ \ud \rb_i / \ud \qb}$ that
\[
\log \pare{\frac{\ud \qb}{\ud \prob_i}} \sim \frac{\cb_i}{\delta_{-i}} + \frac{1}{\lambda_{-i}} \sum_{\jii \setminus \set{i}} \lambda_j \log \pare{\frac{\ud \rprob_j}{\ud \prob_i}}.
\]
The last equivalence relation allows us to write \eqref{eq:f-o-1} as
\begin{equation} \label{eq:f-o-1.5}
\expec_{\qb} \bra{\pare{-\exp \pare{ \zeb_i - \frac{\cb_i}{\lambda_{-i} \delta_i} - \frac{1}{\lambda_{-i}} \sum_{\jii \setminus \set{i}} \lambda_j \log \pare{\frac{\ud \rprob_j}{\ud \prob_i}}} + 1 + \frac{\cb_i}{\delta_{-i}} }X} = 0,
\end{equation}
where
\begin{equation}  \label{eq:best_resp_z}
\zeb_i = - \log \expec_{\prob_i} \bra{\exp \pare{- \frac{\cb_i}{\delta_i} }} + \log \expec_{\prob_i} \bra{\exp \pare{ \frac{\cb_i}{\delta_{-i}} } \prod_{\jii \setminus \set{i}} \pare{\frac{\ud \rprob_j}{\ud \prob_i}}^{\lambda_j / \lambda_{-i}}}.
\end{equation}
Up to now, $X \in \Lb^\infty$ was fixed, but arbitrary. Ranging $X$ over $\Lb^\infty$ in \eqref{eq:f-o-1.5} gives
\begin{equation} \label{eq:f-o-2}
\exp \pare{ \zeb_i - \frac{\cb_i}{\lambda_{-i} \delta_i} - \frac{1}{\lambda_{-i}} \sum_{\jii \setminus \set{i}} \lambda_j \log \pare{\frac{\ud \rprob_j}{\ud \prob_i}}}  = 1 +
 \frac{\cb_i}{\delta_{-i}}.
\end{equation}
Necessarily, $\cb_i > - \delta_{-i}$ should hold. Taking logarithms and rearranging \eqref{eq:f-o-2} gives \eqref{eq:best_resp_foc}.

\subsubsection{Optimality of candidates for best response}

We now proceed to showing that the necessary conditions for best response are also sufficient. (As mentioned in the discussion following Theorem \ref{thm:best_response}, we have not been able to show whether $\V_i (\cdot; \rprob_{-i})$ is concave; therefore, first order conditions do not immediately imply optimality.) Fixing $\rprob \in \PP$, and assuming the conditions stated, we shall show below that $\V_i (\rprob; \rprob_{-i}) \leq \V_i (\rb_i; \rprob_{-i})$.

Define $X \dfn \lambda_i \log \pare{\ud \rprob / \ud \rb_i}$. Similarly to the arguments in \S \ref{subsubsec:proof_of_best_resp_prop_foc} above, the contract that agent $\iii$ would obtain by responding $\rprob \in \PP$ would be
\[
C_i^X \dfn \cb_i + \delta_{-i} X - \expec_{\qprob^X} \bra{\cb_i + \delta_{-i} X},
\]
where $\qprob^X \in \PP$ is such that $\log(\ud \qprob^X / \ud \qb) \sim X$. It follows that
\begin{align} \label{eq:util_diff}
\V_i (\rprob; \rprob_{-i}) - \V_i (\rb_i; \rprob_{-i}) &= \U_i (C_i^X) - \U_i (\cb_i) \\
\nonumber &= \U_i (\cb_i + \delta_{-i} X ) - \U_i (\cb_i) -
\frac{\expecqb \bra{\exp \pare{X  } (\cb_i + \delta_{-i}
X)}}{\expecqb \bra{\exp \pare{X } } }.
\end{align}

\begin{rem} \label{rem: exp_integrability}
If $\expecqb \bra{\exp \pare{ X_+ } X_+ } = \infty$ was true (equivalently, and since $\exp \pare{ X } X $ is bounded below, if $\expecqb \bra{\exp \pare{  X  } X } = \infty$ was true), one would necessarily have $\expecqb \big[ \exp \pare{ X } \cb_i \big] = - \infty$, which
is impossible in view of $\expecqb \bra{\exp
\pare{  X  }} < \infty$ and $\cb_i > - \delta_{-i}$. It follows
that $\expecqb \bra{\exp
\pare{  X_+ } X_+ } < \infty$.
\end{rem}

Since $\cb_i > - \delta_{-i}$, one may define the $(0,
\infty)$-valued random variable $\db_i \dfn 1 + \cb_i /
\delta_{-i}$. From  \eqref{eq:best_resp_qprob_dens}, it follows that $- \cb_i / \delta_i \sim \log \pare{\ud \qb / \ud \prob_i} + \log \db_i$; since $\expecqb \bra{\db_i} = 1$, it actually holds that $\exp(- \cb_i / \delta_i) = \expec_{\prob_i}\bra{\exp(- \cb_i / \delta_i)}(\ud \qb / \ud \prob_i)\db_i$. Therefore, we obtain
\begin{align*}
\U_i (\cb_i + \delta_{-i} X ) - \U_i (\cb_i) &=- \delta_i \log \expec_{\prob_i} \bra{ \exp \pare{- \frac{ \cb_i+X\delta_{-i}}{\delta_{i}}  } }
+\delta_i \log \expec_{\prob_i} \bra{ \exp \pare{- \frac{ \cb_i}{\delta_{i}}  } }
\\
&= - \delta_i \log \expecqb \bra{ \db_i \exp \pare{- \frac{
\delta_{-i} }{\delta_{i}} X } }.
\end{align*}
Combining the previous, including \eqref{eq:util_diff} and Remark \ref{rem: exp_integrability}, it suffices to show that
\[
- \delta_i \log \expecqb \bra{ \db_i \exp \pare{-
\delta_{-i} X / \delta_i } } \leq \frac{\expecqb \bra{\exp \pare{X  } (\cb_i + \delta_{-i} X)}}{\expecqb \bra{\exp \pare{X  } } }
\]
holds for all $X \in \Lb^0$ with $\expecqb \bra{\exp \pare{
X_+ } X_+ } < \infty$. Since $\db_i > 0$ and $\expecqb \big[ \db_i \big] =
1$, an application of Jensen's inequality under the probability
which has density $\db_i$ with respect to $\qb$ gives $- \delta_i \log \expecqb \bra{ \db_i \exp \pare{- \delta_{-i} X / \delta_i } } \leq \delta_{-i} \expecqb \bra{\db_i X}$. (In particular, $\expecqb \big[ \db_i X_- \big] <
\infty$.) On the other hand, upon defining $\chi \dfn
 \log \expecqb \bra{\exp \pare{ X  }} \in
\Real$, note that
\[
\frac{\expecqb \bra{\exp \pare{X  }  \cb_i }}{\expecqb \bra{\exp \pare{X  } } } = \expecqb \bra{\exp \pare{ X - \chi}
 \cb_i } = \delta_{-i} \expecqb \bra{ \db_i  \pare{\exp \pare{
X - \chi } - 1}},
\]
where in the last equality we have used the facts $\cb_i =
\delta_{-i} (\db_i - 1)$ and $\expecqb \bra{\exp \pare{
X - \chi  } } = 1 = \expecqb \bra{ \db_i }$.
Using the inequality $\exp(x) \geq 1 + x$, valid for all $x \in
\Real$, we obtain that
\[
\frac{\expecqb \bra{\exp \pare{X  }  \cb_i }}{\expecqb \bra{\exp
\pare{X  } } } \geq \delta_{-i} \expecqb \bra{\db_i (X - \chi)} =
\delta_{-i} \expecqb \bra{\db_i X} - \delta_{-i} \log \expecqb
\bra{\exp \pare{ X  }}.
\]
(In particular, $\expecqb \bra{ \db_i X_+ } < \infty$, which
implies that $\expecq \bra{ \db_i X } \in \Real$.) Putting
everything together, it follows that it suffices to show that, for $X \in \Lb^0$ with $\expecqb \bra{\exp \pare{
X_+ } X_+ } < \infty$, $\log \expecqb \bra{\exp \pare{  X  }}
\leq \expecqb \bra{\exp \pare{X  }  X } / \expecqb \bra{\exp \pare{X  } }$ holds. This inequality follows from Jensen's inequality applied to $(0, \infty) \ni z \mapsto \phi(z) = z \log z$, which is a convex function; then, $\phi (\expecqb \bra{\exp(X)}) \leq \expecqb \bra{\phi(\exp(X))}$, which is exactly what was required.

\subsection{Proof of Theorem \ref{thm:best_response}}

\label{subsec:proof_of_best_resp_thm}
Define $R_{-i} \dfn (1 / \lambda_{-i}) \sum_{j \in I \setminus \set{i}} \lambda_j \log \pare{ \ud \rprob_j / \ud \prob_i}$, and note that $\exp(R_{-i}) \in \Lb^1(\prob_i)$ holds in view of H\"older's inequality. For $z_i \in \Real$ implicitly define $C_i (z_i) \in \Lb^0$ as the $(- \delta_{-i}, \infty)$-valued random variable satisfying $(1 / \lambda_{-i})C_i(z_i) / \delta_i + \log \pare{1 + C_i (z_i) / \delta_{-i} } = z_i  - R_{-i}$. Note that the existence and uniqueness of the solution $C_i(z_i)$ for each $z_i \in \Real$ follows from the fact that the function $(- 1, \infty) \ni  y \mapsto (\delta / \delta_i ) y + \log (1 + y)$ is strictly increasing from $- \infty$ to $\infty$. For $z_i \in \Real$, define also  the $(0, \infty)$-valued random variable $D_i (z_i) \dfn 1 + C_i (z_i) / \delta_{-i}$, and note that it is the unique solution to the equation
\begin{equation} \label{eq:D_implicit_response}
\frac{D_i (z_i) - 1}{\lambda_i} +  \log D_i (z_i) = z_i - R_{-i}.
\end{equation}
Observe that $D_i$ (and, as a consequence $C_i$) is increasing as a function of $z_i$. It is also straightforward to check that $\plim_{z_i \to - \infty} D_i (z_i) = 0$ and $\plim_{z_i \to \infty} D_i (z_i) = \infty$.

\begin{lem} \label{lem:best_resp_lem_1}
For all $z_i \in \Real$, it holds that
\begin{equation} \label{eq:best_resp_lem_1}
1 \wedge \exp(z_i - R_{-i}) \leq D_i (z_i) \leq 1 \vee \exp(z_i - R_{-i}).
\end{equation}
In particular, both $D_i (z_i)^{-\lambda_i} \exp(\lambda_{-i} R_{-i})$ and $D_i (z_i)^{\lambda_{-i}} \exp(\lambda_{-i} R_{-i})$ are in $\Lb^1(\prob_i)$.
\end{lem}

\begin{proof}
Fix $z_i \in \Real$. By  \eqref{eq:D_implicit_response}, on the event $\set{D_i (z_i) < 1}$ it holds that $\log D_i (z_i) \geq z_i - R_{-i}$, while on the event $\set{D_i (z_i) \geq 1}$ it holds that $\log D_i (z_i) \leq z_i - R_{-i}$. These observations verify the inequalities \eqref{eq:best_resp_lem_1}. Since $D_i(z_i)^{-1} \leq 1 \vee \exp(- z_i + R_{-i})$ and $\exp(R_{-i}) \in \Lb^1(\prob_i)$, $D_i(z_i)^{-1}  \in \Lb^1(\prob_i)$ follows; then, H\"older's inequality implies that $D_i (z_i)^{-\lambda_i} \exp(\lambda_{-i} R_{-i}) \in \Lb^1(\prob_i)$. Furthermore, $D_i (z_i)^{\lambda_{-i}} \exp(\lambda_{-i} R_{-i}) \leq \exp(\lambda_{-i} z_i) \vee \exp(\lambda_{-i} R_{-i})$, which implies that $D_i (z_i)^{\lambda_{-i}} \exp(\lambda_{-i} R_{-i}) \in \Lb^1(\prob_i)$ because $\exp(R_{-i}) \in \Lb^1(\prob_i)$.
\end{proof}

In view of Lemma \ref{lem:best_resp_lem_1}, for each $z_i \in \Real$ one may define $\qprob_i (z_i) \in \PP$ via
\[
\log \pare{ \ud \qprob_i(z_i) / \ud \prob_i} \sim - \lambda_i \log D_i(z_i) + \lambda_{-i} R_{-i} \sim D_i(z_i) + R_{-i};
\]
in other words, $\log \pare{ \ud \qprob_i(z_i) / \ud \prob_i} \sim - \lambda_i \log \pare{1 + C_i (z_i) / \delta_{-i} }  + \sum_{j \in I \setminus \set{i}} \lambda_j \log \pare{ \ud \rprob_j / \ud \prob_i}$. Furthermore, for $z_i \in \Real$, Lemma \ref{lem:best_resp_lem_1} and in particular the fact that $D_i (z_i)^{\lambda_{-i}} \exp(\lambda_{-i} R_{-i})\in\Lb^1(\prob_i)$ implies that $D_i(z_i) \pare{\ud \qprob_i(z_i) / \ud \prob_i} \in \Lb^1(\prob_i)$, which in turn implies that $D_i(z_i) \in \Lb^1(\qprob_i(z_i))$.
As mentioned in the discussion following the statement of Proposition \ref{prop: best_response_first_ord}, in order to establish Theorem \ref{thm:best_response}, we need to show that there exists a \emph{unique} $\zeb_i \in
\Real$ with the property that $\expec_{\qprob_i (\zeb_i)} \bra{ D_i (\zeb_i)  }
= 1$. Define the function $f_i : \Real \mapsto (0, \infty]$ via
$f_i (z_i) = \expec_{\qprob_i (z_i)} \bra{ D_i (z_i) }$, for $z_i
\in \Real$. Since $D_i(z_i) \in \Lb^1(\qprob_i (z_i))$ for all $z_i \in \Real$, it follows that $f_i(z_i) < \infty$ for all $z_i \in \Real$. It is straightforward to check
$f_i$ is continuous, in view of the dominated convergence theorem and Lemma \ref{lem:best_resp_lem_1}. 

Let $\oprob_i$ be the probability measure in $\PP$ such that $\log \pare{\ud \oprob_i / \ud \prob_i} \sim R_{-i}$. Then, thanks to the equivalence relation $\log \pare{ \ud \qprob_i(z_i) / \ud \prob_i} \sim D_i(z_i) + R_{-i}$, it holds that
\begin{equation}\label{eq:f}
f_i (z_i) = \frac{\expec_{\oprob_i} \bra{\exp(D_i(z_i)) D_i(z_i)}}{\expec_{\oprob_i} \bra{\exp(D_i(z_i))}},
\end{equation}
for all $z_i \in \Real$. In fact,  since the covariance of $\exp(D_i(z_i))$ and $D_i(z_i)$ is non-negative under any probability, we have that $f_i(z_i)\geq \expec_{\oprob_i} \bra{ D_i (z_i)}$, for all $z_i \in \Real$. Using the monotone convergence theorem and the relationship \eqref{eq:f}, $\lim_{z_i \to \infty} f_i(z_i) = \infty$ follows from $\lim_{z_i \to \infty} D_i(z_i) = \infty$. Furthermore, the monotone convergence theorem and the fact that $\lim_{z_i \to - \infty} D_i (z_i) = 0$ imply the limiting relationships $\lim_{z_i \to - \infty} \expec_{\oprob_i} \bra{\exp(D_i(z_i))} = 1$ and $\lim_{z_i \to - \infty} \expec_{\oprob_i} \bra{\exp(D_i(z_i)) D_i(z_i)} = 0$, from which we obtain $\lim_{z_i \to - \infty} f_i (z_i) = 0$. It follows that there exists \emph{at least} one $\zeb_i \in \Real$ such that $f_i (\zeb_i) = 1$. We also claim that $f_i$ is strictly increasing, which implies that $\zeb_i$ is indeed unique. In preparation, note that differentiating \eqref{eq:D_implicit_response} with respect to $z_i$ and rearranging gives $D_i' (z_i) = q_i (D_i (z_i))$, where $(0, \infty) \ni y \mapsto q_i(y) \dfn  \lambda_i y / \pare{\lambda_i + y}$. In particular, since $q_i$ is an increasing function, the covariance between $D_i'(z_i)$ and $D_i(z_i)$ will be non-negative for all $z_i \in \Real$ under any probability. Straightforward computations using the definition of $\qprob_i (z_i)$ for $z_i \in \Real$ give $f_i' (z_i) = \expec_{\qprob_i (z_i)} \bra{D_i'(z_i) + D_i(z_i) D_i'(z_i)} -  \expec_{\qprob_i (z_i)} \bra{D_i(z_i)} \expec_{\qprob_i (z_i)} \bra{D_i'(z_i)} = \expec_{\qprob_i (z_i)} \bra{D_i'(z_i) } + \cov_{\qprob_i (z_i)} \pare{ D_i(z_i), D_i'(z_i) }$. Since $D_i'(z_i) > 0$ and $\cov_{\qprob_i (z_i)} \pare{ D_i(z_i), D_i'(z_i) } \geq 0$ for all $z_i \in \Real$, Theorem \ref{thm:best_response} has been proved.

\subsection{Proof of Theorem \ref{thm: nash}}\label{subsec: proof_of_nash_thm}

Suppose that $\pare{\qg, (\cg_i)_{\iii}}$ is a Nash equilibrium and let $(\rg_i)_{\iii} \in \PP^I$ be the associated revealed subjective beliefs. We shall first prove the relationship \eqref{eq:equil_Q}. In view of
Proposition \ref{prop: best_response_first_ord}, and since $\cg_i / \delta_i \sim \log \pare{\ud \rg_i / \ud \qg}$ and $\sum_{j \in I \setminus \set{i}} \lambda_j \log \pare{\ud \rg_j / \ud \prob_i} = \log \pare{\ud \qg / \ud \prob_i} - \lambda_i \log \pare{\ud \rg_i / \ud \prob_i}$, \eqref{eq:best_resp_foc} gives
\[
- \lambda_{-i} \log \pare{1 + \frac{\cb_i}{\delta_{-i}} } \sim \frac{\cg_i}{\delta_i} + \sum_{j \in I \setminus \set{i}} \lambda_j \log \pare{\frac{\ud \rg_j}{\ud \prob_i}} \sim  \lambda_{-i} \log \pare{\frac{\ud \rg_i}{\ud \prob_i}},
\]
i.e., $\log \pare{ \ud \rg_i / \ud \prob_i} \sim - \log \pare{1 +  \cg_i / \delta_{-i} }$, for all $\iii$, which is \eqref{eq:reported_nash_agent}. Since $\log \pare{ \ud \prob_i / \ud \qp} \sim \cp_i / \delta_i$ holds for all $\iii$ (see \eqref{eq: opt_contr}), it follows that $\lambda_i \log \pare{ \ud \rg_i / \ud \qp} \sim - \lambda_i \log \pare{1 +  \cb_i / \delta_{-i} } + \cp_i / \delta$ for all $\iii$. In turn, and since $\sum_{\jii} \cp_j = 0$, the latter gives  $\log \pare{ \ud \qg / \ud \qp} \sim  \sum_{\jii} \lambda_j \log \pare{ \ud \rg_j / \ud \qp}$, which, in view of \eqref{eq:reported_nash_agent}, is exactly \eqref{eq:equil_Q}.

To prove \eqref{eq:equil_C}, we add $\lambda_i  \log \pare{1 +  \cg_i / \delta_{-i} } \sim - \lambda_i \log \pare{ \ud \rg_i / \ud \prob_i}$ to \eqref{eq:best_resp_foc}, and obtain
\[
\frac{\cg_i}{\delta_i} + \log \pare{1 + \frac{\cg_i}{\delta_{-i}} } \sim - \sum_{j \in I} \lambda_j \log \pare{\frac{\ud \rg_j}{\ud \prob_i}} \sim \log \pare{\frac{\ud \prob_i}{\ud \qg}} \sim
\frac{\cp_i}{\delta_i} - \log \pare{\frac{\ud \qg}{\ud \qp}}.
\]
The latter equivalence, combined with \eqref{eq:equil_Q}, gives \eqref{eq:equil_C} for appropriate $\zg \equiv \pare{\zg_i}_{\iii} \in \Real^I$. The market clearing conditions $\sum_{\iii} \cg_i = 0 = \sum_{\iii} \cp_i$ show that $\sum_{\iii} \zg_i = 0$, i.e., $\zg \in \DI$. Finally, the fact that $\expecqg \bra{\cg_i} = 0$ holds for all $\iii$ results directly from $\pare{\cg_i}_{\iii} \in \C_{\qg}$.

For the proof of the reverse implication, assume that conditions (N1), (N2) and (N3) hold for $\pare{\qg, (\cg_i)_{\iii}}$, and fix $\iii$. Define the associate revealed subjective beliefs $(\rg_i)_{\iii} \in \PP^I$ via $\log( \ud \rg_i / \ud \qg) \sim \cg_i / \delta_i$. Since $\cp_i / \delta_i \sim \log( \ud \prob_i / \ud \qp)$, a combination of \eqref{eq:equil_C} and \eqref{eq:equil_Q}  gives $\log \pare{1 +  \cg_i / \delta_{-i} } \sim - \log \pare{ \ud \rg_i / \ud \prob_i}$. Using again  \eqref{eq:equil_C} and \eqref{eq:equil_Q},
\begin{align*}
\frac{\cg_i}{\delta_i} + \lambda_{-i} \log \pare{1 +  \frac{\cg_i}{\delta_{-i}}  } &\sim \frac{\cp_i}{\delta_i} - \log \pare{ \frac{\ud \qg}{\ud \qp}} - \lambda_i \log \pare{1 +  \frac{\cg_i}{\delta_{-i}}  } \\
&\sim - \log \pare{ \frac{\ud \qg}{\ud \prob_i}} + \lambda_i \log \pare{ \frac{\ud \rg_i}{\ud \prob_i}} \sim - \sum_{\jii \setminus \set{i}} \lambda_j \log \pare{ \frac{\ud \rg_j}{\ud \prob_i}}.
\end{align*}
It follows that the sufficient conditions for optimality of Proposition \ref{prop: best_response_first_ord} are satisfied for each $\iii$. In order to show that $\pare{\qg, (\cg_i)_{\iii}}$ is a Nash equilibrium, it is left to verify that $(\cg_i)_{\iii}\in\C_{\qg}$. Indeed, summing \eqref{eq:equil_C} with respect to $i$ implies that $\sum_{\iii}\cg_i=0$, since $\zg$ is assumed to belong in $\Delta^I$. This fact, together with the requirement $\cg_i>-\delta_{-i}$, implies the uniform boundedness of $\cg_i$; in particular, $\exp(\cg_i/\delta_i)\in\mathbb{L}^1(\qg)$ for all $\iii$. Taking also (3) into account, we conclude that $(\cg_i)_{\iii}\in\C_{\qg}$, which completes the proof.

\subsection{Proof of Proposition \ref{prop: C_system_sol}} \label{subsec: proof_C_system_sol}

Fix $z \in \DI$. Suppose for the moment that a solution to
\eqref{eq: C_system} indeed exists, and set
\begin{equation} \label{eq:L_z}
L(z) \dfn \sum_{\iii} \lambda_i \log \pare{1 + \frac{C_i(z)}{\delta_{-i}} }.
\end{equation}
Then, with $\eta_i : (0, \infty) \mapsto \Real$ defined
via $\eta_i (x) = \delta_{-i} (x - 1) +  \delta_i \log x$, for all
$x \in (0, \infty)$, \eqref{eq: C_system} implies that $\eta_i
\pare{1 + C_i(z)/ \delta_{-i}} = z_i + \cp_i +
\delta_i L(z)$ holds for all $\iii$. With $\theta_i : \Real
\mapsto (0, \infty)$ denoting the inverse of $\eta_i$ for all
$\iii$, it follows that
\begin{equation} \label{eq:C_from_L}
C_i(z) = \delta_{-i} \bra{\theta_i \pare{z_i + \cp_i +  \delta_i
L(z)} - 1}, \quad \forall \iii.
\end{equation}
Plugging back into the definition of $L(z)$ in \eqref{eq:L_z}, we obtain that
\begin{equation} \label{eq:L_z_equation}
L(z) = \sum_{\iii} \lambda_i \log \theta_i \pare{z_i + \cp_i +
\delta_i L(z)}
\end{equation}
should be satisfied.

We now proceed backwards, by showing that \eqref{eq:L_z_equation} has a unique solution. With $z \in \DI$ fixed, define $w : \Omega \times \Real \mapsto \Real$ via $w(y) = y -  \sum_{\iii} \lambda_i \log \theta_i \pare{z_i + \cp_i +
\delta_i y}$ for $y \in \Real$, where the dependence of $w$ in $\omega \in \Omega$ is suppressed. The derivative of $w$ with respect to the spatial coordinate is
\[
w'(y) = 1 - \sum_{\iii} \frac{\lambda_i}{1 + (\delta_{-i} / \delta_i) \theta_i \pare{z_i + \cp_i +  \delta_i y}}=
 \sum_{\iii} \frac{\lambda_{-i}\theta_i \pare{z_i + \cp_i +  \delta_i y}}{1 + (\delta_{-i} / \delta_i) \theta_i \pare{z_i + \cp_i +  \delta_i y}} > 0,
\]
for all $y \in \Real$. Since $\theta_i (y)$ behaves sub-linearly as $y \to \infty$, $\lim_{y \uparrow \infty} w(y) =
\infty$ follows in a straightforward way. Furthermore, from the definition of $\theta_i$ we obtain that $x < \delta_i \log \theta_i(x)$ holds for all $x \in (- \infty, 0)$ and $\iii$. This implies that on the event $\set{y < - \vee_{\jii} \pare{(z_j + \cp_j) /
\delta_j}}$, one has $w (y) < y - \sum_{\iii} (1/\delta)\pare{ z_i + \cp_i
+  \delta_i y} = 0$, which shows both that the equation $w (L(z))
= 0$ has a unique solution, and that
\begin{equation} \label{eq:L_z_ineq}
- L(z) \leq \bigvee_{\jii} \frac{\cp_j + z_j }{\delta_j} \leq \bigvee_{\jii} \frac{z_j }{\delta_j} + \bigvee_{\jii} \frac{\cp_j}{\delta_j},
\end{equation}
Given the existence of a unique $L (z)$ solving \eqref{eq:L_z_equation}, $C_i (z)$ is specified for all $\iii$ via \eqref{eq:C_from_L}.

Since $\exp( \cp_i / \delta_i) \in \Lb^1(\qp)$ holds for all $\iii$, it is straightforward that $\exp(\bigvee_{\iii} \cp_i / \delta_i) \in \Lb^1(\qp)$. Therefore, \eqref{eq:L_z_ineq} and the equality $\prod_{\jii}  \pare{1 +  C_j(z) / \delta_{-j} }^{- \lambda_j} = \exp \pare{- L(z)}$ imply the validity of \eqref{eq:Qg_z_integr}, which concludes the proof.

\subsection{Proof of Theorem \ref{thm:nash_exist_uniq}} \label{subsec:proof_of_nash_exist_uniq}

We first establish the general existence result, and then tackle uniqueness in the two-agent case.

\subsubsection{Proof of existence of Nash equilibrium}

We use notation from Proposition \ref{prop: C_system_sol} and the discussion following it. For each $z \in \DI$ and $\iii$, define $u_i(z) \dfn \U_i (C_i(z))$. Furthermore, for each $z \in \DI$ define $u(z) \dfn \sumi u_(z)$, as well as
\[
\DI \ni z \mapsto \phi_i (z) = u_i(z) - \up_i + \lambda_i \pare{\up - u(z)}, \quad \iii, \quad z \in \DI.
\]
Note that $\sumi \phi_i (z) = 0$ for all $z \in \DI$, so that $\phi \equiv (\phi_i)_{\iii}$ is $\DI$-valued. The obvious continuity of $\DI \ni z \mapsto L(z)$ from \eqref{eq:L_z_equation} and the domination relation given by \eqref{eq:L_z_ineq} allow the application of the dominated convergence theorem to establish that $\phi : \DI \mapsto \DI$ is a continuous function.

\begin{lem} \label{lem:fixed_point}
$\zg \in \DI$ corresponds to a Nash equilibrium if and only if it is a fixed point of $\phi$.
\end{lem}

\begin{proof}
In view of the discussion in \S \ref{subsubsec: z in compact set}, if $\zg \in \DI$ corresponds to a Nash equilibrium, then $\zg$ is a fixed point of $\phi$. Conversely, we shall show that any fixed point of $\phi$ corresponds to a Nash equilibrium point. With $L(z)$ as in \eqref{eq:L_z}, and recalling \eqref{eq: C_system}, start by observing that
\[
C_i(z) - z_i - \up_i = \delta_i \log \pare{\frac{\ud \prob_i}{\ud \qp}} + \delta_i L(z) - \delta_i \log \pare{1 + \frac{C_i(z)}{\delta_{-i}}}, \quad \iii, \quad z \in \DI.
\]
From the last equality, it follows that
\[
u_i(z) - z_i - \up_i = - \delta_i \log \expecqp \bra{\exp (- L(z))\pare{1 + \frac{C_i(z)}{\delta_{-i}}}}, \quad \iii, \quad z \in \DI.
\]
Adding up the previous equality over all agents, we obtain
\[
u(z) - \up = - \delta \sum_{\jii} \lambda_j \log \expecqp \bra{\exp (- L(z))\pare{1 + \frac{C_j(z)}{\delta_{-j}}}}, \quad z \in \DI.
\]
Since $\log \pare{\ud \qprob(z) / \ud \qp} = - L(z) + \lambda(z)$ for appropriate $\lambda(z) \in \Real$ and $\phi_i(z) - z_i = u_i(z) - z_i - \up_i - \lambda_i (u(z) - \up)$ holds for all $\iii$ and $z \in \DI$, it follows that
\[
\phi_i(z) - z_i = - \delta_i \log \expec_{\qprob(z)} \bra{1 + \frac{C_i(z)}{\delta_{-i}}} + \delta_i \sum_{\jii} \lambda_j \log \expec_{\qprob(z)} \bra{1 + \frac{C_j(z)}{\delta_{-j}}}, \quad \iii, \quad z \in \DI,
\]
where note that the quantity $\lambda(z)$ cancels out in the above equation. Now, suppose that $\zg \in \DI$ is a fixed point of $\phi$. From the last equality, it follows that the quantities $\expec_{\qprob(\zg)} \bra{1 + C_i(\zg) / \delta_{-i}}$ have the same value, which we call $x(\zg)$, for all $\iii$. In other words, $\expec_{\qprob(\zg)} \bra{C_i(\zg)} = \delta_{-i} \pare{x(\zg) - 1}$ holds for all $\iii$. Since $\sumi C_i(\zg) = 0$, we obtain that $x(\zg) = 1$, which implies that $\expec_{\qprob(\zg)} \bra{C_i(\zg)} = 0$ holds for all $\iii$, in turn implying that $\zg$ corresponds to a Nash equilibrium.
\end{proof}

In view of Lemma \ref{lem:fixed_point}, existence of a Nash equilibrium will follow if we can show that $\phi$ has at least one fixed point. For any $z \in \DI$ and $\iii$, the strong bound $C_i(z) > - \delta_{-i}$ implies $u_i(z) \geq - \delta_{-i}$. Furthermore, $u(z) \leq \up$ holds for all $z \in \DI$, from aggregate optimality of Arrow-Debreu equilibrium. Therefore, it follows that
\[
\phi_i(z) = u_i(z) - \up_i + \lambda_i \pare{\up - u(z)} \geq - \delta_{-i} - \up_i, \quad \iii, \quad z \in \DI.
\]
Define the set $K \dfn \set{z \in \DI \such z_i \geq - \delta_{-i} - \up_i, \ \forall \iii}$, and note that $K$ is a compact and convex subset of $\DI$. Since $\phi$ is continuous and maps $K$ to $K$, Brower's fixed point theorem implies that $\phi$ has at least one fixed point on $K$, which establishes the claim. (In fact, according to the discussion in \S \ref{subsubsec: z in compact set},  any fixed point has to live in the smaller set $\set{z \in \DI \such z_i \geq - \up_i, \ \forall \iii}$.)

\subsubsection{Proof of uniqueness in the two-agent case}

Note that $(z_0, z_1) \in \DI$ if and only if $z_0 = - z_1$. In
the course of the proof, we identify $\Real$ and $\DI$ via $\Real
\ni z \leftrightarrow (z, - z) \in \DI$, i.e., considering only
the ``zero'' coordinate. Correspondingly, for $z \in \Real$ we
write $C_i(z)$ instead of $C_i \pare{(z, -z)}$ for $i \in
\set{0,1}$; similarly, for $z \in \Real$ we write $L(z)$ 
instead of $L(z, -z)$ of \eqref{eq:L_z}. 

In view of Proposition \ref{prop:nash_root}, as well the equality $C_1(z) = - C_0(z)$ for all $z \in \Real$, we need to prove the existence of a
unique $\zg \in \Real$ such that $\expec_{\qprob (\zg)}
\bra{C_0(\zg)} = 0$ holds; since existence was already established, we shall only focus on uniqueness here. Define the continuous function $\Real \ni z
\mapsto f_0 (z) = \expecqz \bra{C_0(z)}$; then, it suffices to show that $f_0$ is \emph{strictly} increasing.

Recall that $C_1 (z) = - C_0(z)$ holds for all $z \in \Real$,
and rewrite \eqref{eq: C_system} as
\begin{equation} \label{eq: C0(z)}
C_0(z) + \frac{\delta_0 \delta_1}{\delta} \log \pare{ \frac{1 +
C_0(z) / \delta_1}{1 - C_0(z) / \delta_0}} = z + \cp_0, \quad
\forall z \in \Real.
\end{equation}
Differentiating with respect to $z \in
\Real$, after some algebra one obtains that
\[
C_0' (z) = \frac{(\delta_1 + C_0(z)) (\delta_0 - C_0(z))
}{\delta_0 \delta_1 + (\delta_1 + C_0(z)) (\delta_0 - C_0(z))},
\quad \forall z \in \Real.
\]
Since $-\delta_1 <C_0(z)<\delta_0$, $C_0' (z)$ is clearly a $(0, \infty)$-valued random variable for
all $z \in \Real$. Furthermore, recalling that $L(z) = \lambda_0
\log (1 + C_0 (z) / \delta_1) + \lambda_1 \log (1 - C_0 (z) /
\delta_0)$ holds for $z \in \Real$, after differentiation and
algebra we obtain
\[
L'(z) = \frac{(\delta_0 - \delta_1) - C_0(z)
}{(\delta_1 + C_0(z)) (\delta_0 - C_0(z))}  C_0'(z) =
\frac{(\delta_0 - \delta_1) - C_0(z)}{\delta_0
\delta_1 + (\delta_1 + C_0(z)) (\delta_0 - C_0(z))}, \quad \forall
z \in \Real.
\]
In other words, with $q(x) = \pare{x + \delta_1 - \delta_0} / \pare{\delta_0 \delta_1 + (\delta_1 + x) (\delta_0 - x)}$ for $x \in (-\delta_1, \delta_0)$, $L'(z) = - q(C_0(z))$ holds for all $z \in \Real$. Since $q'(x) = \pare{2 \delta_0 \delta_1 + (x + \delta_1 -
\delta_0)^2} / \pare{\delta_0 \delta_1 + (\delta_1 + x) (\delta_0 -
x)}^2 > 0$ holds for all $x \in  (-\delta_1, \delta_0)$, the covariance between $C_0(z)$ and $- L'(z)$ is non-negative under any probability, for all $z \in \Real$.
Continuing, if we take into account that $\log \pare{ \ud\qprob(z) / \ud \prob_i} \sim \log \pare{ \ud\qprob(z) / \ud \qp} + \log \pare{ \ud\qp / \ud \prob_i} \sim - L
(z) - \cp_i / \delta_i$ for all $\iii$, it is straightforward to compute that $f_0' (z) = \expec_{\qprob(z)} \bra{C_0' (z)} - \cov_{\qprob(z)} \pare{C_0 (z), L'(z) }$ holds for all $z \in \Real$. Since $C_0' (z)$ is a $(0, \infty)$-valued random variable and $\cov_{\qprob(z)}
\pare{C_0 (z), L'(z) } \leq 0$ for all $z \in \Real$, the
claim is proved.

\subsection{Proof of Proposition  \ref{prop: limit_AD}}\label{subsec: proof_of_prop_limit_AD}

For the remainder of Appendix \ref{sec: appe}, define $\delta^m \dfn \delta^m_0 +
\delta_1$, $\lambda^m_0 \dfn \delta^m_0 / \delta^m$ and
$\lambda_1^m \dfn \delta_1 / \delta^m = 1 - \lambda_0^m$ for all $\mina$. In view of Theorem \ref{thm: AD} and the fact that $\plimm \pare{\ud \qpk / \ud \prob_0} = 1$, we need to focus on the limit of the sequence $\pare{\ent(\qpk \such \prob_1)}_{\mina}$. For each $\mina$, similarly to the equation \eqref{eq:qpm0}, it holds that
\[
\frac{\ud \qpk}{\ud \prob_1} = \expec_{\prob_1} \bra{\pare{\frac{\ud \prob_0}{\ud \prob_1}}^{ \lambda_0^m}}^{-1} \pare{\frac{\ud \prob_0}{\ud \prob_1}}^{\lambda_0^m}.
\]
Therefore, with $Z \dfn \ud \prob_0 / \ud \prob_1$ and $\phi(x) = x \log x$, it holds that
\[
\ent(\qpk \such \prob_1) = \frac{ \expec_{\prob_1} \bra{\phi(Z^{ \lambda_0^m})} - \phi \pare{\expec_{\prob_1} \bra{ Z^{ \lambda_0^m}}}  }{\expec_{\prob_1} \bra{Z^{ \lambda_0^m}}}
\]
Under Assumption \ref{ass:asymptotic}, the fact that $\limm \lambda_0^m = 1$ and the dominated convergence theorem give that $\limm \ent(\qpk \such \prob_1) = \expec_{\prob_1} \bra{\phi(Z) } = \ent(\prob_0 \such \prob_1)$. Since $\cpm_0 = - \cpm_1 = \delta_1 \log(\ud \qpm / \ud \prob_1) - \delta_1 \ent(\qpm \such \prob_1)$, the limiting relationship $\cpi_0 = \delta_1 \log(\ud \prob_0 / \ud \prob_1) - \delta_1 \ent(\prob_0 \such \prob_1)$ readily follows.

Continuing, for $\mina$, note that $\upk_0 = \U_0 (\cpm_0) \leq \expec_{\prob_0} \bra{\cpm_0}$. Note that there exists $c > 0$ such that $\sup_{\mina} \abs{\cpm_0} \leq c \pare{1 + \abs{\log(\ud \prob_0 / \ud \prob_1)}}$. Under Assumption \ref{ass:asymptotic}, $\log(\ud \prob_0 / \ud \prob_1) \in \Lb^1 (\prob_0)$, which means that the dominated convergence theorem can be applied and gives $\limsup_{m \to \infty} \upk_0 \leq \expec_{\prob_0} \bra{\cpi_0} = 0$. Since $0 \leq \upk_0$ holds for all $\mina$, $\lim_{m \to \infty} \upk_0 = 0$ follows.

Moving on to agent 1, the fact that $\upk_1 = \delta_1 \ent(\qpk \such \prob_1)$ holds for all $\mina$ and the previous discussion gives $\lim_{m \to \infty} \upk_1 = \delta_1 \ent(\prob_0 \such \prob_1)$. The proof is complete.

\subsection{Proof of Lemma \ref{lem: zgi}} \label{subsec: proof_of_lem:zgi}

Since the function $(-1, \infty) \ni x \mapsto x + \log \pare{1 + x}$ is strictly increasing and continuous and maps $(-1, \infty)$ to $(- \infty, \infty)$, it follows that $\plim_{z \to - \infty}  C^{\infty}_0(z) = - \delta_1$ and $\plim_{z \to \infty}  C^{\infty}_0(z) = \infty$. Let $D^{\infty}_0(z) \dfn 1 + C^{\infty}_0(z) / \delta_1$, for all $z \in \Real$. On $\set{C^{\infty}_0(z)
> 0}$ it holds that $1/D^{\infty}_0(z) \leq 1$. On $\set{C^{\infty}_0 (z) \leq 0}$, it holds that $\delta_1 \log D^{\infty}_0(z) \geq C^{\infty}_0(z) + \delta_1 \log D^{\infty}_0(z) = z + \cpi_0$, which implies that $1/D^{\infty}_0(z) \leq \exp \pare{ - (z + \cpi_0) / \delta_1} = \exp \pare{ \ent (\prob_0 \such \prob_1) - z / \delta_1} \pare{\ud \prob_1 / \ud \prob_0}$. Therefore, since  $1/D^{\infty}_0(z) \leq 1 \vee \exp \pare{ \ent (\prob_0 \such \prob_1) - z / \delta_1} \pare{\ud \prob_1 / \ud \prob_0}$ holds everywhere, Assumption \ref{ass:asymptotic} implies that the function $\Real \ni z \mapsto \expec_{\prob_0} \bra{ 1 / D^{\infty}_0(z) } = \expec_{\prob_0} \bra{ \pare{1 + C^{\infty}_0(z) / \delta_1}^{-1} }$ is $(0, \infty)$-valued, continuous, strictly decreasing and, in view of the limiting behaviour of $\Real \ni z \mapsto C^{\infty}_0(z)$ and the monotone convergence theorem, maps $\Real$ to $(0, \infty)$. Therefore, the result follows.

\subsection{Proof of Theorem \ref{thm:limiting_nash_one}}\label{subsec: proof_of_limiting_nash_one}

In order to ease the reading throughout the proof,  for all $\mina$ define the $(0, 1 / \lambda^m_1)$-valued random
variable $\dgm_0 \dfn 1 + \cgm_0 / \delta_1$ and the $(0, \infty)$-valued random variable $\dbm \dfn 1 + \cbm / \delta_1$. We use the obvious notation $\qgm \in \PP$ for $\mina$. As in \eqref{eq:best_resp_qprob_dens}, let $\qbm \in \PP$ be defined via
\[
\log \pare{\frac{\ud \qbm}{\ud \prob_0}} \sim - \lambda^m_0 \log \pare{1 + \frac{\cbm}{\delta_1} }  + \lambda^m_1 \log \pare{\frac{\ud \prob_1}{\ud \prob_0}},
\]
and recall that $\expec_{\qbm} \bra{\cbm} = 0$ for all $\mina$, as follows from Proposition \ref{prop: best_response_first_ord}.

For each $\mina$, define $w^m : (0, \infty) \mapsto \Real$ and $\phi^m : (0, 1 / \lambda^m_1) \mapsto
\Real$ via
\begin{align}
 w^m (x) &= \delta_1 \pare{x - 1} + \lambda_0^m \delta_1 \log(x), \quad x \in (0, \infty), \label{eq: wm} \\
 \phi^m (x) &= \delta_1 \pare{ x - 1 } + \lambda_0^m \delta_1 \log \pare{\frac{\lambda_0^m x}{1 -
\lambda_1^m x}}, \quad x \in (0, 1 / \lambda^m_1), \label{eq: phim}
\end{align}
and note by the equivalence $\cpm_0\sim \delta^m_0\log \pare{\ud \prob_0/\ud \qpm}\sim \delta^m_0\lambda^m_1\log \pare{\ud \prob_0/\ud \prob_1}$, and equations \eqref{eq:best_resp_foc} and \eqref{eq: C_two_agent}, that
\begin{equation} \label{eq: approx_equations}
w^m (\dbm) = \zbm + \cpm_0, \quad \phi^m (\dgm_0) = \zgm_0 + \cpm_0, \quad \forall \mina,
\end{equation}
for an appropriate sequence $\pare{\zbm}_{\mina}$ from Theorem \ref{thm:best_response} and the sequence $\pare{\zgm_0}_{\mina}$  from Theorem  \ref{thm: nash}. The next two results show in particular that $(\zgm_0)_{\mina}$ and  $(\zbm)_{\mina}$ are bounded.

\begin{lem} \label{lem: zg_bdd}
The sequence $(\zgm_0)_{\mina}$ is bounded (in $\Real$), and there
exists $a \in \Real$ such that
\begin{equation} \label{eq:D_domination}
\log \pare{1 / \dgm_0} \leq a + \log \pare{\ud \prob_1 / \ud \prob_0}, \quad \text{holds on } \set{\dgm_0 \leq 1}, \quad \forall \mina.
\end{equation}
\end{lem}

\begin{proof}
Note that $\expec_{\qgm} \bra{\cgm_0} = 0$ is equivalent to
$\expec_{\qgm} \big[ \dgm_0 \big] = 1$, for all $\mina$.

Applying \eqref{eq:equil_C} for $i=1$, and using \eqref{eq:equil_Q} and the fact that $\cpm_1/ \delta_1 \sim \log \pare{\ud \prob_1 / \ud \qpm}$, it follows that
\[
\frac{\cgm_1}{\delta_1} + \log \pare{1 + \frac{\cgm_1}{\delta^m_0}} \sim \frac{\cpm_1}{\delta_1} + \log \pare{\frac{\ud \qpm}{\ud \qgm}} \sim \log \pare{\frac{\ud \prob_1}{\ud \qgm}}.
\]
Coupling the last equivalence with $\cgm_1 / \delta_1 = - \cgm_0 / \delta_1 = 1 - \dgm_0$ and after some extra algebra, we obtain the further equivalence $\log \pare{ \ud \qgm / \ud \prob_1}  \sim \dgm_0 - \log \big(1 - \lambda^m_1 \dgm_0 \big)$. Therefore,
\[
1 = \expec_{\qgm} \big[ \dgm_0 \big] = \frac{\expec_{\prob_1} \bra{ \exp(\dgm_0) \dgm_0 / \pare{1 - \lambda^m_1 \dgm_0}} }{\expec_{\prob_1} \bra{ \exp(\dgm_0) / \pare{1 - \lambda^m_1 \dgm_0}  }} \geq \expec_{\prob_1} \big[ \dgm_0 \big], \quad \forall \mina,
\]
where the last inequality follows from $\cov_{\prob_1} \big( \exp(\dgm_0)/ \big( 1 - \lambda^m_1 \dgm_0
\big) , \dgm_0 \big) \geq 0$, holding for all $\mina$. Since
that $\expec_{\prob_1} \big[ \dgm_0 \big] \leq 1$ for all
$\mina$, $( \dgm_0)_{\mina }$ is $\lz$-bounded.

Similarly, apply \eqref{eq:equil_C} (for $i= 0$), we obtain
\[
\frac{\cgm_0}{\delta^m_0} + \log \pare{1 + \frac{\cgm_0}{\delta_1}} \sim \frac{\cpm_0}{\delta^m_0} + \log \pare{\frac{\ud \qpm}{\ud \qgm}} \sim \log \pare{\frac{\ud \prob_0}{\ud \qgm}},
\]
or, equivalently, $\log \pare{ \ud \qgm / \ud \prob_0} \sim  - \delta_1 \dgm_0 / \delta^m_0 - \log \dgm_0$. It follows that
\[
1 = \frac{1}{\expec_{\qgm} \bra{\dgm_0}} = \frac{ \expec_{\prob_0} \bra{ (1 / \dgm_0) \exp(- \delta_1 \dgm_0 / \delta^m_0  ) } }{ \expec_{\prob_0} \bra{ \exp(- \delta_1 \dgm_0 / \delta^m_0  ) } } \geq \expec_{\prob_0} \bra{1 / \dgm_0}, \quad \forall \mina,
\]
where the last inequality follows since $\cov_{\prob_0} \big( \exp(- \delta_1 \dgm_0 / \delta^m_0  ) , 1/ \dgm_0 \big) \geq 0$ holds for all $\mina$. The fact
that $\expec_{\prob_0} \big[ 1 / \dgm_0 \big] \leq 1$ holds for
all $\mina$ implies that $(1 / \dgm_0)_{\mina }$ is
$\lz$-bounded.

We then obtain that $\pare{\log \dgm_0}_{\mina}$ is bounded in $\Lb^0$, from which it follows that the family $\set{\phi^m (\dgm_0) \such \mina }$ is
bounded in $\Lb^0$. As $\pare{\cpm_0}_{\mina}$ converges in
$\lz$, \eqref{eq: approx_equations} implies that the family
$\set{\zgm_0 \such \mina}$ is bounded (in $\Real$).

Continuing, since $x \in (0,1)$ implies $\phi^m \pare{x} \leq \lambda_0^m \delta_1 \log(x)$, on the event $\set{\dgm_0 \leq 1}$ it follows that $\log \pare{\dgm_0} \geq (1 / \lambda_0^m \delta_1) \phi^m \pare{\dgm_0} = (1 / \lambda_0^m \delta_1) \pare{\zgm + \cpm_0}$ holds. A combination of the equality $\cpm_0 = \delta_0^m \log \pare{\ud \prob_0 / \ud \qpm} + \upm_0$ and \eqref{eq:qpm0} gives
\begin{equation} \label{eq:cgm0_estimate}
\frac{\cpm_0}{\lambda_0^m \delta_1} = - \log \pare{\frac{\ud \prob_1}{\ud \prob_0}} + \frac{1}{\lambda^m_1} \log \expec_{\prob_0} \bra{\pare{\frac{\ud \prob_1}{\ud \prob_0}}^{ \lambda_1^m}}  + \frac{\upm_0}{\lambda_0^m \delta_1}, \quad \forall \mina.
\end{equation}
The second and third term of the right-hand-side of the above equation converge (to $- \ent(\prob_0 \such \prob_1)$ and zero, respectively) and the sequence $\pare{ \zgm /  \lambda_0^m}_{\mina}$ is bounded in $\Real$; therefore, the existence of $a \in \Real$ such that \eqref{eq:D_domination} holds readily follows.
\end{proof}

\begin{lem} \label{lem: zb_bdd}
The sequence $\pare{\zbm}_{\mina}$ is bounded in
$\mathbb{R}$, and there exists $c\in\Real$ such that
\begin{align}
\log(1 / \dbm) &\leq c + \log \pare{ \ud \prob_1 / \ud \prob_0} \quad \text{holds on } \{\dbm\leq 1\}, \quad \forall \mina. \label{eq:logD} \\
\log \dbm  &\leq c - \log \pare{ \ud \prob_1 / \ud \prob_0} \quad \text{holds on } \{\dbm> 1\}, \quad \forall \mina. \label{eq:logD_bis}
\end{align}
\end{lem}
\begin{proof}
Recall that $\expec_{\qbm} \big[ \dbm \big] = 1$, for
all $\mina$. In view of \eqref{eq:best_resp_foc}, we obtain
\[
\log \dbm \sim - \log \pare{\frac{\ud \prob_1}{\ud \prob_0}} - \frac{\dbm}{\lambda^m_0}, \quad \forall \mina.
\]
Further, due to \eqref{eq:best_resp_qprob_dens}, it holds that $\log \pare{ \ud \qbm / \ud \prob_0} \sim - \lambda^m_0 \log \dbm   + \lambda^m_1 \log \pare{ \ud \prob_1 / \ud \prob_0}$. The last two equivalences give $\log \pare{\ud \qbm / \ud
\prob_1} \sim \dbm$. Therefore, it holds that $1 = \expec_{\qbm} \big[ \dbm \big] = \expec_{\prob_1}
\bra{ \exp\pare{\dbm} \dbm }  / \expec_{\prob_1} \bra{
\exp\pare{\dbm}} \geq \expec_{\prob_1} \big[ \dbm \big]$ for all $\mina$, where the last inequality follows from the fact that
$\cov_{\prob_1} \big( \exp(\dbm), \dbm \big) \geq 0$ holds for
all $\mina$. It follows that $\expec_{\prob_1} \big[ \dbm \big]
\leq 1$ for all $\mina$, which implies that $( \dbm)_{\mina }$ is
$\lz$-bounded. Hence, the family $\set{w^m (\dbm) \such \mina }$
is also bounded from above in $\Lb^0$. Since $\pare{\cpm_0}$ converges in $\lz$, \eqref{eq: approx_equations} implies that $\pare{\zbm}_{\mina}$ is
bounded from above (in $\Real$).

By way of contradiction, suppose that $\pare{\zbm}_{\mina}$ is not bounded
from below. Passing to a subsequence if necessary, we may
assume that $(\zbm)_{\mina}$
is a sequence of negative numbers with $\limm \zbm = -\infty$. Hence, again by \eqref{eq: approx_equations} we
get that $\limm \dbm =0$. Since $\zbm \leq 0$ for all $\mina$, and $x + \log x \leq 1 / \lambda_0^m + (1 / \lambda_0^m \delta_1) w^m(x)$ holds for all $x \geq 1$, it follows that $\dbm + \log \dbm \leq  1 / \lambda_0^m +  \cpm_0 / \lambda_0^m \delta_1$ holds on $\set{\dbm > 1}$, for all $\mina$. Given \eqref{eq:cgm0_estimate}, we conclude the existence of $\kappa \in \Real$ such that $\dbm + \log \dbm \leq \kappa + \log \pare{\ud \prob_0 / \ud \prob_1}$ holds on $\set{\dbm > 1}$, for all $\mina$.
It follows that one can use the dominated convergence theorem on the right-hand-side of the equality $\expec_{\qbm} \big[ \dbm \big] = \expec_{\prob_1}
\bra{ \exp\pare{\dbm} \dbm }  / \expec_{\prob_1} \bra{
\exp\pare{\dbm}}$, valid for all $\mina$, and obtain $\limm \expec_{\qbm} \big[ \dbm \big] = 0$, which contradicts the
fact that $\expec_{\qbm} \big[ \dbm \big]=1$ holds for all $\mina$. We
conclude that $\pare{\zbm}_{\mina}$ is bounded
from below too.

To show \eqref{eq:logD}, note that
$w^m(x) \leq \lambda_0^m \delta_1 \log(x)$ holds when $0 < x \leq 1$; hence, $\log(\dbm) \geq
 (\zbm + \cpm_0)/ (\lambda_0^m  \delta_1)$ holds on $\{\dbm\leq 1\}$ for all $\mina$. From \eqref{eq:cgm0_estimate}, we obtain that
\[
\log(1 / \dbm) - \log \pare{ \ud \prob_1 / \ud \prob_0} \leq - \frac{1}{\lambda^m_1} \log \expec_{\prob_0} \bra{\pare{\frac{\ud \prob_1}{\ud \prob_0}}^{ \lambda_1^m}} - \frac{\zbm  + \upm_0}{\lambda_0^m   \delta_1}, \quad \forall \mina.
\]
The right-hand-side of the above inequality is bounded in $\Real$; therefore, \eqref{eq:logD} follows. Similarly, since $\lambda_0^m \delta_1 \log(x) \leq w^m(x)$ holds when $x > 1$, $\log(\dbm) \leq
 (\zbm + \cpm_0)/ (\lambda_0^m  \delta_1)$ holds on $\{\dbm > 1\}$ for all $\mina$.  Using the same estimates from \eqref{eq:cgm0_estimate} that were used to establish \eqref{eq:logD}, \eqref{eq:logD_bis} follows.
\end{proof}

We now show that we have the expected limiting behaviour through subsequences.

\begin{lem} \label{lem: correct_conv}
Let $\zgi_0$ be given as in Lemma \ref{lem: zgi}. If $(\zgmk_0)_{\kin}$ is any convergent subsequence of $(\zgm_0)_{\mina}$, we have $\limk \zgmk_0 = \zgi_0$ and $\plimk \cgmk_0 = C^{\infty}_0 (\zgi_0)$. Similarly, if $\pare{\zbmk}_{\kin}$ is any convergent
subsequence of $\pare{\zbm}_{\mina}$, then it holds that $\limk
\zbmk = \zgi_0$ and $\plimk \cbmk =
C^{\infty}_0(\zgi_0)$.
\end{lem}

\begin{proof}
Set $\hzi \dfn \limk \zgmk_0$ and $\tze \dfn \limk \zbmk$. Define the function $(0, \infty) \ni x \mapsto \phi (x) = \delta_1 \pare{x - 1} + \delta_1
\log \pare{x}$, and note that both $\pare{\phi^m}_{\mina}$ of \eqref{eq: phim} and $\pare{w^m}_{\mina}$ of \eqref{eq: wm} converge uniformly to $\phi$ on compact subsets of $(0, \infty)$. This fact, combined with \eqref{eq: approx_equations}, Lemma \ref{lem: zg_bdd} and Lemma \ref{lem: zb_bdd} implies that $(\dgmk_0)_{\kin}$  has a
$(0, \infty)$-valued $\lz$-limit $\hDi = 1 + \hCi / \delta_1$ and $(D_0^{m_k,\mathsf{r}})_{\kin}$ has a $(0, \infty)$-valued $\lz$-limit $\tD^\infty_0 = 1 + \tC^\infty_0 / \delta_1$, satisfying $\phi (\hDi) = \hzi +\cpi_0$ and $\phi (\tD^\infty_0) = \tze + \cpi_0$.

We first tackle the Nash equilibrium case. In the proof of Lemma \ref{lem: zg_bdd}, the equality
\[
\frac{ \expec_{\prob_0} \bra{ (1 / \dgmk_0) \exp(- \delta_1 \dgmk_0 / \delta^{m_k}_0  ) } }{ \expec_{\prob_0} \bra{ \exp(- \delta_1 \dgmk_0 / \delta^{m_k}_0  ) } } = 1, \quad \forall \kin,
\]
was established. Since $\limk \delta^{m_k}_0 = \infty$ and $\exp(- \delta_1 \dgmk_0 / \delta^{m_k}_0  ) \leq 1$ holds for all $\kin$, \eqref{eq:D_domination} allows use of the dominated convergence theorem to obtain $\expec_{\prob_0} \big[ (1 + \hCi/ \delta_1)^{-1} \big] = \expec_{\prob_0} \big[ 1 / \hDi \big]
= 1$. Due to Lemma \ref{lem: zgi}, it then follows that
$\hzi = \zgi_0$, which also implies that $\hCi =
\cgi_0 (\zgi_0)$.

We continue to deal with the best response case. Since \eqref{eq:best_resp_qprob_dens} implies that $\log \pare{ \ud \qbmk / \ud \prob_0} \sim - \lambda^{m_k}_0 \log \dbmk   + \lambda^{m_k}_1 \log \pare{ \ud \prob_1 / \ud \prob_0}$, we obtain
\[
1 = \frac{1}{\expec_{\qbmk} \bra{\dbmk}} = \frac{\expec_{\prob_0}
\bra{ \pare{\dbmk}^{ - \lambda^{m_k}_0} \pare{ \ud \prob_1 / \ud \prob_0}^{\lambda_1^{m_k}}} }{ \expec_{\prob_0}
\bra{ \pare{\dbmk}^{\lambda^{m_k}_1} \pare{ \ud \prob_1 / \ud \prob_0}^{\lambda_1^{m_k}}} }
,\quad \forall \kin.
\]
By the domination relationship in \eqref{eq:logD}, one may apply the dominated
convergence theorem in the numerator above to obtain
\[
\limk \expec_{\prob_0}
\bra{ \pare{\dbmk}^{ - \lambda^{m_k}_0} \pare{ \ud \prob_1 / \ud \prob_0}^{\lambda_1^{m_k}}} = \expec_{\prob_0} \bra{ 1 / \tD^\infty_0 }.
\]
Similarly, the domination relationship in \eqref{eq:logD} allows one to apply the dominated convergence theorem in the denominator above to obtain
\[
\limk \expec_{\prob_0}
\bra{ \pare{\dbmk}^{\lambda^{m_k}_1} \pare{ \ud \prob_1 / \ud \prob_0}^{\lambda_1^{m_k}}} = 1.
\]
Combining everything, we obtain $\expec_{\prob_0} \big[ 1/ \tD^\infty_0 \big] = 1$,
and by Lemma \ref{lem: zgi},  it follows that $\tze = \zgi_0$
and $\tD^\infty_0 = 1 + C_0^{\infty}(\zgi_0)/\delta_0$, which in turn implies that $\tC_0^{\infty}=C_0^{\infty}(\zgi_0)$.
\end{proof}

We are now in position to finish the proof of Theorem \ref{thm:limiting_nash_one}. If $\plimm \cgm_0 = C_0^{\infty}(\zgi_0)$ would fail, there would exist $\epsilon \in (0, 1)$ and a subsequence $(\cgmk_0)_{\kin}$ of $(\cgm)_{\mina}$ such that $\expec \bra{1 \wedge |\cbmk - C_0^{\infty}(\zgi_0)|} > \epsilon$ holds for all $\kin$. Since the sequence
$\pare{\zgmk_0}_{\kin}$ is bounded by Lemma
\ref{lem: zg_bdd}, there exists a further subsequence of
$(\zgmk_0)_{\kin}$ that is convergent. Then, Lemma \ref{lem:
correct_conv} implies that there exists a further subsequence of
$(\cgmk_0)_{\kin}$ that $\lz$-converges to $C_0^{\infty}(\zgi_0)$,
contradicting that $\expec \bra{1 \wedge |\cgmk_0 -
C_0^{\infty}(\zgi_0)|} > \epsilon$ holds for all $\kin$. The proof that
$\plimm \cbm = C_0^{\infty}(\zgi_0)$ is carried out in the exact same way, using Lemma \ref{lem: zb_bdd} in place of Lemma \ref{lem: zg_bdd}.

\subsection{Proof of Proposition  \ref{prop: limiting_gain_game}}

\label{subsec: proof_of_prop_limiting_gain_game}

Recall that $\ugm_0 = - \delta^m_0 \log \expec_{\prob_0} \bra{\exp \pare{- \cgm_0 / \delta_0^m}}$, for all $\mina$. On one hand, $\ugm_0 \leq \expec_{\prob_0} \bra{\cgm_0}$ holds for all $\mina$. From \eqref{eq: C_two_agent} and \eqref{eq: approx_equations}, we get that on the event $\set{\cgm_0 > 0}$, it holds that $\cgm_0 \leq \zgm_0 + \cpm_0$ for all $\mina$. Therefore, by \eqref{eq:cgm0_estimate}, Lemma \ref{lem: zg_bdd} and Proposition \ref{prop: limit_AD}, there exists a constant $k > 0$ such that $\cgm_0 \leq k + \log_+ (\ud \prob_0 / \ud \prob_1)$ holds for all $\mina$. By Assumption \ref{ass:asymptotic} and (the reverse version of) Fatou's lemma, it follows that $\limsup_{m \to \infty} \ugm_0 \leq \expec_{\prob_0} \bra{\cgi_0}$. On the other hand, since $\limm \delta_0^m = \infty$, for all $\kin$ it follows that $\liminf_{m \to \infty} \ugm_0 \geq \liminf_{m \to \infty} \pare{- k \log \expec_{\prob_0}\bra{\exp(- \cgm_0  / k)} } = - k \log \expec_{\prob_0}\bra{\exp( - \cgi_0  / k)}$, where the last equality follows from the dominated convergence
theorem, where the fact that $- \cgm_0 \leq \delta_1$ holds for all $\mina$ is also used. Sending $k \to \infty$, $\liminf_{m \to \infty} \ugm_0 \geq \lim_{k \to \infty} \pare{- k \log \expec_{\prob_0}\bra{\exp(- \cgi_0  / k)}} =  \expec_{\prob_0} \bra{\cgi_0}$ follows. Combining with the reverse inequality, we obtain that
\[
\lim_{m \to \infty} \ugm_0 = \expec_{\prob_0} \bra{\cgi_0} = \expec_{\qgi} \bra{\pare{1 + \frac{\cgi_0}{\delta_1}} \cgi_0} = (1 / \delta_1) \var_{\qgi} \pare{\cgi_0}.
\]
Since $\lim_{m \to \infty} \upm_0 = 0$, $\lim_{m \to \infty} \pare{\ugm_0 - \upm_0} = (1 / \delta_1) \var_{\qgi} \pare{\cgi_0}$ follows.
In order to obtain the limiting loss for agent 1, note first that taking expectation with respect to $\prob_0$ on both sides of $\cgi_0 + \delta_1 \log \pare{1 +  \cgi / \delta_1} = \zgi_0 + \cpi_0$, we obtain $\expec_{\prob_0} \bra{\cgi_0} + \delta_1 \expec_{\prob_0} \bra{\log \pare{\ud \prob_0 / \ud \qgi}} = \zgi_0$, where the fact that $1 + \cgi_0 / \delta_1 = \ud \prob_0 / \ud \qgi$ is used. Recalling $\expec_{\prob_0} \bra{\cgi_0} = (1 / \delta_1) \var_{\qgi} (\cgi_0)$, we obtain the equality $\zgi_0 =  (1 / \delta_1) \var_{\qgi} (\cgi_0) + \delta_1 \ent (\prob_0 \such \qgi)$. In particular, since $\zgi_0 \in \Real_+$, it follows that $\var_{\qgi} (\cgi_0) < \infty$ and $\ent (\prob_0 \such \qgi) < \infty$. Recall that $\limm \zgm_0 = \zgi_0$ was obtained in the proof of Theorem \ref{thm:limiting_nash_one}, and that, from \eqref{eq:loss_decomp},
\[
\zgm_0 = \lambda_0^m \pare{ \upm - \ugm} - \pare{\upm_0 - \ugm_0}  = \lambda_0^m \pare{ \upm_1 - \ugm_1} - \lambda_1^m \pare{\upm_0 - \ugm_0}.
\]
Since $\limm \lambda_0^m = 1$, $\limm \lambda_1^m = 0$ and $\limm \pare{\upm_0 - \ugm_0} = (1 / \delta_1)\var_{\qgi} (\cgi_0) < \infty$,
\[
\limm \pare{ \upm_1 - \ugm_1} = \limm \zgm_0 = \zgi_0 = (1 / \delta_1) \var_{\qgi} (\cgi_0) + \delta_1 \ent (\prob_0 \such \qgi)
\]
follows, which concludes the proof.

\subsection{Proof of Theorem \ref{thm:limiting_contract_both}} \label{subsec: proof of limiting_contract_both}

Under Assumption \ref{ass:asymptotic_both}, $\delta_0^m \log \pare{\ud \prob_0^m / \ud \qpm} \sim \lambda_1 \xi_0 - \lambda_0 \xi_1$ holds for all $\mina$; therefore, $\cpm_0 = \lambda_1 \xi_0 - \lambda_0 \xi_1 - \expec_{\qpm} \bra{\lambda_1 \xi_0 - \lambda_0 \xi_1}$ holds for all $\mina$. Since $\pare{\lambda_1 \xi_0 - \lambda_0 \xi_1} \in \li$ and $\pare{\qpm}_{\mina}$ converges to $\prob$ in total variation norm, one readily obtains $\limm \expec_{\qpm} \bra{\lambda_1 \xi_0 - \lambda_0 \xi_1} = \expec_{\prob} \bra{\lambda_1 \xi_0 - \lambda_0 \xi_1} = 0$; therefore, $\limm \cpm_0 = \lambda_1 \xi_0 - \lambda_0 \xi_1$ follows.

We proceed with the limiting behaviour of the sequence
$\pare{\cgm}_{\mina}$. For each $\mina$, define the function
$(-\delta_1^m, \delta_0^m) \ni y \mapsto \psi^m(y) \dfn y+\lambda_0\delta^m_1 \log
\pare{\pare{1+ y / \delta_1^m} / \pare{1- y / \delta_0^m}}$; it then follows by \eqref{eq: C0(z)} that $\psi^m (\cgm_0) = \zgm_0 + \cpm_0$, for all $\mina$. Note that $\psi^m$ is strictly increasing with $\psi^m(0) = 0$ for
all $\mina$; furthermore, $\pare{\psi^m}_{\mina}$ converges
uniformly in compact subsets of $\Real$ to $\psi^\infty: \Real
\mapsto \Real$ defined via $\psi^\infty(y) = 2 y$ for $y \in
\Real$. 

\begin{lem} \label{lem: z_bdd_both}
The sequence $(\zgm_0)_{\mina}$ is bounded in $\Real$.
\end{lem}

\begin{proof}
We shall show that $(\zgm_0)_{\mina}$ is bounded above. A
symmetric argument applied to agent 1 shows that $(\zgm_1)_{\mina}$ is bounded above and since $\zgm_0 = - \zgm_1$
holds for all $\mina$, it will follow that $(\zgm_0)_{\mina}$ is
also bounded below.

Recall that $\expec_{\qgm} \bra{\cgm_0} = 0$ holds for all $\mina$. As in the beginning of the proof of Lemma \ref{lem: zg_bdd}, applying \eqref{eq:equil_C} for $i=1$, and using \eqref{eq:equil_Q}, it follows that
$\log \pare{ \ud \qgm / \ud \prob^m_1}  \sim \cgm_0 / \delta_1^m - \log \big(\lambda_0^m - \cgm_0 / \delta^m \big) \sim  \cgm_0 / \delta_1^m  - \log \pare{1- \cgm_0 / \delta_0^m}$.
Therefore,
\[
 0 = \expec_{\qgm} \big[ \cgm_0 \big] = \frac{\expec_{\prob_1^m}
\bra{ \cgm_0 \exp \pare{ \cgm_0/\delta_1^m } /
\pare{1-\cgm_0/\delta_1^m}} }{\expec_{\prob_1^m} \bra{ \exp
\pare{ \cgm_0/\delta_1^m } / \pare{1-\cgm_0/\delta_1^m}}} \geq
\expec_{\prob_1^m} \bra{ \cgm_0 }, \quad \forall \mina,
\]
where the last inequality follows from $\cov_{\prob^m_1} \big( \exp \pare{ \cgm_0/\delta_1^m } /
\pare{1-\cgm_0/\delta_1^m}, \cgm_0 \big) \geq 0$, holding for all
$\mina$. We obtain then that $\expec_{\prob} \bra{\exp \pare{ \xi_1 /
\delta^m_1} \cgm_0} \leq 0$ holds for all $\mina$.

Suppose that $(\zgm_0)_{\mina}$ fails to be bounded above. By passing to a subsequence if necessary, we may assume without loss of generality that $\limm
\zgm_0 = \infty$ and $(\zgm_0)_{\mina}$ is nonnegative. Note
then by \eqref{eq: C0(z)} it follows that $\limm \prob \bra{\cgm_0 > K} = 1$ holds
for all $K \in \Real_+$. Furthermore, $\cgm_0 \geq - (\cpm_0)_-$ holds
for all $\mina$, which, together with the uniform boundedness of $\pare{\cpm_0}_{\mina}$, implies the existence of $c_0 \in (0, \infty)$ such that $\cgm_0 \geq - c_0$ holds for all $\mina$. Since $\xi_1 \in \li$, $\exp \pare{ \xi_1 / \delta^m_1} \cgm_0 \geq - c$ holds for all
$\mina$, for some appropriate $c \in (0, \infty)$. The last domination from below,
combined with $\plimm \exp \pare{ \xi_1 /
\delta^m_1} = 1$ and $\limm \prob \bra{\cgm_0 > K} = 1$ holding for all $K \in \Real_+$
would imply $\limm \expec_{\prob} \bra{\exp \pare{ \xi_1 /
\delta^m_1} \cgm_0} = \infty$, which contradicts the fact that $\expec_{\prob}
\bra{\exp \pare{ \xi_1 /
\delta^m_1} \cgm_0} \leq 0$ holds for all
$\mina$. It follows that $(\zgm_0)_{\mina}$ is bounded above,
which completes the argument.
\end{proof}

\begin{lem} \label{lem: correct_conv_both}
Let $(\zgmk)_{\kin}$ be any convergent subsequence of
$(\zgm)_{\mina}$. Then, it holds that $\limk \zgmk =0$  and $\plimk \cgmk_0 = \cpi_0 /2$.\end{lem}

\begin{proof}
Let $\tz^\infty \dfn \limk \zgmk$. Due to the fact that
$\pare{\psi^m}_{\mina}$ converges uniformly to $\psi$ on compact
subsets of $\Real$ and Lemma \ref{lem: z_bdd_both}, it follows
that $(\cgmk_0)_{\kin}$ will have a  $\lz$-limit $\tC^\infty_0$,
and this limit will satisfy $2 \tC^\infty_0 = \tz^\infty +
\cpi_0$. Recall the inequality $\expec_{\prob} \bra{\exp \pare{ \xi_1 /
\delta^{m_k}_1} \cgmk_0} \leq 0$ holds for all $\kin$ that was
established in the proof of Lemma \ref{lem: z_bdd_both}.
Furthermore, since $(\zgmk)_{\kin}$ and $\pare{C^{m_k, \ast}_0}_{\kin}$ are convergent, and in particular uniformly bounded from below (the latter sequence of random variables in view of the fact that $\xi_i \in \li$ for $i \in \set{0,1}$), and $\xi_1 \in \li$, we infer the existence of $c \in (0, \infty)$ such that the uniform lower domination $\exp \pare{ \xi_1 /
\delta^{m_k}_1} \cgmk_0 \geq - c$ is valid for all $\kin$. An application of Fatou's lemma gives $\expec_{\prob} \big[ \tC^\infty_0 \big] \leq 0$. The symmetric argument
from the side of agent 1 will give $\expec_{\prob} \big[ \tC^\infty_0
\big] \geq 0$, which implies that $\expec_{\prob} \big[ \tC^\infty_0 \big]
= 0$. Since $2 \tC^\infty_0 = \tz^\infty + \cpi$ and $\expec_{\prob} \bra{\cpi} = 0$, it follows
that $\tz^\infty = 0$. We conclude that $\tC^\infty_0 = \cpi_0/2$.
\end{proof}

The proof of Theorem \ref{thm:limiting_contract_both} can now be
completed exactly as in the case of Theorem \ref{thm:limiting_nash_one}. If $\plimm \cgm_0 = \cpi_0 / 2$ fails, there
exists $\epsilon \in (0, 1)$ and a subsequence $(\cgmk_0)_{\kin}$
of $(\cgm_0)_{\mina}$ such that $\expec \bra{1 \wedge |\cgmk_0 -
\cpi_0 / 2|}  > \epsilon$ holds for all $\kin$. Since the sequence
$(\zgmk)_{\kin}$ is bounded due to Lemma \ref{lem: z_bdd_both},
there exists a further subsequence of $(\zgmk)_{\kin}$ that is
convergent. Then, Lemma \ref{lem: correct_conv_both} implies that
there exists a further subsequence of $(\cgmk_0)_{\kin}$ that
$\lz$-converges to $\cpi_0 / 2$, contradicting the fact that
$\expec \bra{1 \wedge |\cgmk_0 - \cpi_0 / 2|}  > \epsilon$ holds
for all $\kin$.

\bibliographystyle{amsalpha}
\bibliography{references}
\end{document}